\documentclass[11pt,reqno]{article}
\usepackage{amssymb,amscd,amsmath,amsthm,color}
\usepackage
{hyperref}
\newtheorem{thm}{Theorem}[section]
\newtheorem{prop}[thm]{Proposition}
\newtheorem{lemma}[thm]{Lemma}
\newtheorem{cor}[thm]{Corollary}
\newtheorem{definition}[thm]{Definition}
\newtheorem{remark}[thm]{Remark}
\newtheorem{example}[thm]{Example}

\numberwithin{equation}{section}

\def\cH{\mathcal{H}}
\def\cK{\mathcal{K}}
\def\cB{\mathcal{B}}
\def\cBH{\cB(\cH)}
\def\cBK{\cB(\cK)}
\def\cM{\mathcal{M}}
\def\cA{\mathcal{A}}
\def\cB{\mathcal{B}}
\def\cD{\mathcal{D}}
\def\cS{\mathcal{S}}
\def\cSH{\cS(\cH)}

\def\cF{\mathcal{F}}
\def\cT{\mathcal{T}}
\def\bR{\mathbb{R}}
\def\bC{\mathbb{C}}
\def\bN{\mathbb{N}}
\def\bM{\mathbb{M}}
\def\HS{\mathrm{HS}}
\def\Tr{\mathrm{Tr}\,}
\def\T{\mathrm{Tr}}
\def\id{\mathrm{id}}
\def\Sp{\mathrm{Sp}}
\def\supp{\mathrm{supp}\,}
\def\diag{\mathrm{diag}}
\def\Re{\mathrm{Re}\,}
\def\Im{\mathrm{Im}\,}
\def\sa{\mathrm{sa}}
\def\sym{\mathrm{sym}}
\def\OM{\mathrm{OM}}
\def\OMD{\mathrm{OMD}}
\def\OC{\mathrm{OC}}
\def\<{\langle}
\def\>{\rangle}
\def\eps{\varepsilon}
\def\ffi{\varphi}

\begin{document}
\allowdisplaybreaks

\centerline{\LARGE Equality cases in monotonicity of quasi-entropies,}

\medskip
\centerline{\LARGE Lieb's concavity and Ando's convexity}

\bigskip
\bigskip
\centerline{\large
Fumio Hiai\footnote{{\it E-mail:} hiai.fumio@gmail.com}}

\begin{center}
$^1$\,Graduate School of Information Sciences, Tohoku University, \\
Aoba-ku, Sendai 980-8579, Japan
\end{center}

\medskip
\begin{abstract}
We revisit and improve joint concavity/convexity and monotonicity properties of quasi-entropies due to Petz in a new fashion. Then we characterize equality cases in the monotonicity inequalities (the data-processing inequalities) of quasi-entropies in several ways as follows: Let $\Phi:\mathcal{B}(\mathcal{H})\to\mathcal{B}(\mathcal{K})$ be a trace-preserving map such that $\Phi^*$ is a Schwarz map. When $f$ is an operator monotone or operator convex function on $[0,\infty)$, we present several equivalent conditions for the equality $S_f^K(\Phi(\rho)\|\Phi(\sigma))=S_f^{\Phi^*(K)}(\rho\|\sigma)$ to hold for given positive operators $\rho,\sigma$ on $\mathcal{H}$ and $K\in\mathcal{B}(\mathcal{K})$. The conditions include equality cases in the monotonicity versions of Lieb's concavity and Ando's convexity theorems. Specializing the map $\Phi$ we have equivalent conditions for equality cases in Lieb's concavity and Ando's convexity. Similar equality conditions are discussed also for monotone metrics and $\chi^2$-divergences. We further consider some types of linear preserver problems for those quantum information quantities.


\end{abstract}

{\small
\tableofcontents
}

\section{Introduction}\label{Sec-1}

Among significant ingredients of quantum information theory is the notion of quantum divergences,
the most notable one of which is the relative entropy introduced by Umegaki \cite{Um} in 1962,
defined for positive operators $\rho,\sigma$ as follows:
\begin{align}\label{F-1.1}
D(\rho\|\sigma):=\begin{cases}\Tr\rho(\log\rho-\log\sigma) & \text{if $\rho^0\le\sigma^0$}, \\
+\infty &\text{otherwise},\end{cases}
\end{align}
where $\rho^0$ is the support projection of $\rho$. (This was later extended by Araki \cite{Ar1,Ar2}
to the general von Neumann algebra setting.) There have been several different types of extensions
of the relative entropy so far. A distinguished one is a class of so-called quasi-entropies, first
considered by Kosaki \cite{Ko} and further studied by Petz \cite{Pe2,Pe1} in the finite-dimensional
$\cBH$ setting as well as in the von Neumann algebra setting, which are given in the form
$S_f^K(\rho\|\sigma)$ for positive operators $\rho,\sigma$ with parametrization of functions $f$ on
$(0,\infty)$ and reference operators $K$. Standard $f$-divergences
$S_f(\rho\|\sigma)=S_f^I(\rho\|\sigma)$ ($I$ being the identity operator) are specialized cases of
quasi-entropies. For instance, the relative entropy $D(\rho\|\sigma)$ is $S_f(\rho\|\sigma)$ with
$f(x):=x\log x$. The operational significance of the relative entropy was established in \cite{HP0,ON},
as an optimal error exponent in the hypothesis testing of quantum Stein's lemma. Moreover, in recent
development of other types of quantum hypothesis testing, the importance of two different families
of R\'enyi divergences has emerged; see, e.g., \cite{Aud,Ha,HMO,MO,MO2,Na,NS}. The one is the
conventional (or standard) R\'enyi divergences and the other is the sandwiched R\'enyi divergences
\cite{MDSFT,WWY}, both of which are defined with parameter $\alpha>0$, $\alpha\ne1$, and the
relative entropy becomes their limit as $\alpha\to1$.

The primary property common to quantum divergences $S(\rho\|\sigma)$ is monotonicity under
quantum operations, which is stated as
\begin{align}\label{F-1.2}
S(\Phi(\rho)\|\Phi(\sigma))\le S(\rho\|\sigma)
\end{align}
for a given family of trace-preserving positive maps $\Phi$, typically TPCP (trace-preserving
completely positive) maps. This inequality is also called DPI (data-processing inequality), and
heuristically means that the distinguishability of two states becomes less after applying a
quantum operation to the states. Another common property of quantum divergences is joint
convexity stated as
\begin{align}\label{F-1.3}
S(\lambda\rho_1+(1-\lambda)\rho_2\|\lambda\sigma_1+(1-\lambda)\sigma_2)
\le\lambda S(\rho_1\|\sigma_1)+(1-\lambda)S(\rho_2\|\sigma_2)
\end{align}
for positive operators $\rho_i,\sigma_i$ ($i=1,2$) and $0<\lambda<1$. With the assumption that
$S$ satisfies the positive homogeneity $S(\lambda\rho\|\lambda\sigma)=\lambda S(\rho\|\sigma)$
for $\lambda>0$ and the additivity $S(\rho_1\oplus\rho_2\|\sigma_1\oplus\sigma_2)=
S(\rho_1\|\sigma_1)+S(\rho_2\|\sigma_2)$ for $\rho_i,\sigma_i$ on $\cH_i$ ($i=1,2$), it is widely
known (see, e.g., \cite{FL,Zh}) that joint convexity \eqref{F-1.3} is an immediate consequence of
monotonicity \eqref{F-1.2} under TPCP maps, and vice versa based on Stinespring's representation
of TPCP maps.

In the present article we are concerned with monotonicity inequalities and their equality cases for
quasi-entropies in the finite-dimensional $\cBH$ setting. In Sec.~\ref{Sec-2} we revisit joint
convexity/concavity and monotonicity properties of quasi entropies, formerly discussed in
\cite{Ko,Pe2,Pe1,Sh,Ve}, and improve the previous results by removing certain additional
assumptions on $\rho,\sigma$ and $\Phi$. To this end, in Sec.~\ref{Sec-2} we reconsider the
definition of quasi-entropies $S_f^K(\rho\|\sigma)$ along the lines of the definition of standard
$f$-divergences in \cite{HMPB,HM}. Then we update their monotonicity inequality as follows (see
Theorems \ref{T-2.14} and \ref{T-2.16}): Let $\Phi:\cBH\to \cBK$ be a trace-preserving map such
that $\Phi^*$ is a Schwarz map. Let $\rho,\sigma\in\cBH^+$ (positive operators on $\cH$) and
$K\in\cBK$. If $h$ is operator monotone on $[0,\infty)$, then
\begin{align}\label{F-1.4}
S_h^K(\Phi(\rho)\|\Phi(\sigma))\ge S_h^{\Phi^*(K)}(\rho\|\sigma).
\end{align}
If $f$ is operator convex on $(0,\infty)$ and $K$ belongs to the multiplicative domain of $\Phi^*$,
then
\begin{align}\label{F-1.5}
S_f^K(\Phi(\rho)\|\Phi(\sigma))\le S_f^{\Phi^*(K)}(\rho\|\sigma).
\end{align}
Here, we note that the reversal between inequalities \eqref{F-1.4} and \eqref{F-1.5} is natural
because an operator monotone $h$ is operator concave, i.e., $-h$ is operator convex so that a
``right'' quasi-entropy is $S_{-h}^K=-S_h^K$. Hence inequality \eqref{F-1.4} is considered as a
particular case of \eqref{F-1.5}, while the assumption of $K$ being in the multiplicative domain of
$\Phi^*$ is unnecessary for \eqref{F-1.4}.

As is well known and easily verified, the quasi-entropy $S_f^K(\rho\|\sigma)$ when $f(x)=x^\alpha$
is rewritten as
\[
\Tr K^*\rho^\alpha K\sigma^{1-\alpha},
\]
as far as $\rho,\sigma\in\cBH^{++}$ (invertible positive operators on $\cH$), so that the
concavity/convexity result of quasi-entropies (see Proposition \ref{P-2.11}) reduces to Lieb's
so-called Wigner--Yanase--Dyson--Lieb concavity theorem \cite{Li} and Ando's convexity theorem
\cite{An}. Moreover, the monotonicity inequalities in \eqref{F-1.4} and \eqref{F-1.5} give the
monotonicity inequality versions of Lieb's concavity and Ando's convexity. This type of monotonicity
inequalities have recently been studied in \cite{Ca} in detail.

Various characterizations of equality in different types of quantum information inequalities have
been known so far, including equality cases in the strong subadditivity of von Neumann entropy
and in the joint convexity/concavity of relative entropy, trace functions of Lieb type,  and so on
\cite{HJPW,HM,HMPB,Je,Je5,JP,JR,Pe5,Pe4,Ru,Sh,Ve}. Most notably, in \cite{Pe5,Pe4,JP}
Petz and Jen\v cov\'a characterized sufficiency of subalgebras and reversibility of quantum
operations via equality cases in the monotonicity inequalities of the relative entropy and the
transition probability (equivalently, the standard R\'enyi divergence for $\alpha=1/2$) in the
von Neumann algebra setting. In \cite{HMPB,HM} we discussed, in the finite-dimensional case,
reversibility and error correction of quantum operations for given states via equality cases in the
monotonicity of standard $f$-divergences (see also \cite{Hi2} for reversibility via standard
$f$-divergences in the von Neumann algebra case). Furthermore, equality conditions of monotonicity
(DPI) for quantum R\'enyi type divergences have recently been discussed by several authors
\cite{HM,Je2,Je3,Je4,LRD,WW,Zh2}. As for quasi-entropies, the equality case in the monotonicity
inequality in \eqref{F-1.5} was discussed in \cite{Sh,Ve} when $\Phi$ is a partial trace (hence any
$K\in\cBK$ is in the multiplicative domain of $\Phi^*$). Although equality conditions for DPI under
TPCP maps can essentially be reduced to the partial trace case due to Stinespring's representation
(as carried out in \cite{LRD,WW} in the case of the sandwiched R\'enyi divergence), our main aim of
the paper is to study the equality cases in \eqref{F-1.4} and \eqref{F-1.5} under a more general
class of trace-preserving maps $\Phi:\cBH\to\cBK$ such that $\Phi^*$ is a Schwarz map.

In Sec.~\ref{Sec-3} we first propose a number of conditions for the equality case in
\eqref{F-1.4} when $h$ is operator monotone on $[0,\infty)$, and further consider conditions for
the equality case in \eqref{F-1.5} when $f$ is operator convex on $[0,\infty)$. Equality cases in the
monotonicity inequality versions of Lieb's concavity and Ando's convexity show up among those
equality conditions. We analyze the equivalence between the proposed equality conditions along
the lines of arguments in \cite{HMPB,HM}, see Theorems \ref{T-3.2} and \ref{T-3.3}. The results
obtained are also specialized to equality cases of joint concavity/convexity of quasi-entropies
(including equality cases of Lieb's concavity and Ando's convexity).

Monotone metrics (also called quantum Fisher informations) studied by Petz \cite{Pe6} are
Riemannian metrics on $\cBH^{++}$ (being a Riemannian manifold) given as
$\gamma_\sigma^h(X,Y)$ for $\sigma\in\cBH^{++}$ and (self-adjoint) $X,Y\in\cBH$ with
parametrization of operator monotone functions $h>0$ on $(0,\infty)$, which are characterized
by the monotonicity property
\begin{align}\label{F-1.6}
\gamma_{\Phi(\sigma)}^h(\Phi(X),\Phi(X))\le\gamma_\sigma^h(X,X)
\end{align}
for every TPCP map $\Phi:\cBH\to\cBK$. Quantum $\chi^2$-divergences discussed in
\cite{TKRWV,Han2} are given as $\chi_k^2(\rho,\sigma)$ for density operators $\rho,\sigma$ on
$\cH$ with parametrization of operator monotone decreasing functions $k>0$ on $(0,\infty)$,
which are written as monotone metrics (up to an additive constant) and have the monotonicity
\begin{align}\label{F-1.7}
\chi_k^2(\Phi(\rho),\Phi(\sigma))\le\chi_k^2(\rho,\sigma)
\end{align}
for every $\Phi$ as above. In Sec.~\ref{Sec-4} we show some equality conditions for the
monotonicity in \eqref{F-1.6} and \eqref{F-1.7}, which are slight improvements of Jen\v cov\'a's
characterizations in \cite{Je}. The stated conditions are rather similar to those in Sec.~\ref{Sec-3}
for quasi-entropies. This might be natural because the notions of quasi-entropies and monotone
metrics are defined in a certain duality relation.

Finally, in Sec.~\ref{Sec-5} we determine special forms of maps $\Phi:\cBH\to\cBK$ when
equality conditions for $S_f^K(\rho\|\sigma)$, $\gamma_\sigma^h(K,K)$ and
$\chi_k^2(\rho,\sigma)$ discussed in the previous sections are satisfied for sufficiently many
$K\in\cBK$ with some fixed invertible $\rho,\sigma$, or for sufficiently many $\rho$
with some fixed invertible $\sigma$, see Theorems \ref{T-5.1} and \ref{T-5.3}. These are considered
as some types of linear preserver problems for those quantum information quantities.

\section{Monotonicity of quasi-entropies, revisited}\label{Sec-2}

The notion of quasi-entropies was introduced by Kosaki \cite{Ko} and Petz \cite{Pe2,Pe1} in both
the von Neumann algebra setting and the finite-dimensional setting. In this section we revisit the
monotonicity property of quasi-entropies in the finite-dimensional setting, and improve the results
in \cite{Pe1}.

Throughout the paper, let $\cH,\cK$ be finite-dimensional Hilbert spaces. Let $\cBH$ be the
space of all operators on $\cH$, and $\T$ be the usual trace on $\cBH$. We write $\cBH^+$
and $\cBH^{++}$ for the sets of positive operators and invertible positive operators in $\cBH$,
respectively. We write $\<\cdot,\cdot\>_\HS$ for the \emph{Hilbert--Schmidt inner product} on
$\cBH$, i.e., $\<X,Y\>_\HS:=\Tr X^*Y$ for $X,Y\in\cBH$, and always consider $\cBH$,
$\cBK$ as Hilbert spaces with this inner product. For $A\in\cBH$ let $L_A$ and $R_A$
denote the \emph{left} and \emph{right multiplications} on $\cBH$, respectively, i.e.,
$L_AX:=AX$ and $R_AX:=XA$ for $X\in\cBH$. For $A\in\cBH^+$ we write $\Sp(A)$, $A^0$
and $A^{-1}$ for the spectrum, the support projection and the generalized inverse (i.e., the inverse
restricted on the support) of $A$, respectively.

Let $\Phi:\cBH\to \cBK$ be a linear map, and the adjoint map $\Phi^*:\cBK\to \cBH$
be defined in terms of the Hilbert--Schmidt inner product as
$\<X,\Phi^*(Y)\>_\HS=\<\Phi(X),Y\>_\HS$ for all $X\in\cBH$, $Y\in\cBK$. The map $\Phi$ is
said to be \emph{trace-preserving} (TP) if $\Tr\Phi(X)=\Tr X$ for all $X\in\cBH$, or equivalently
$\Phi^*$ is \emph{unital}, i.e., $\Phi^*(I_\cK)=I_\cH$, where $I_\cH,I_\cK$ are the identity operators
on $\cH,\cK$, respectively (also written as $I$ simply). For each $n\in\bN$ the map $\Phi$ is said
to be \emph{$n$-positive}  if $\id_n\otimes\Phi:\bM_n\otimes \cBH\to\bM_n\otimes \cBK$
($\bM_n$ being the $n\times n$ matrix algebra) is positive, and $\Phi$ is said to be
\emph{completely positive} (CP) if it is $n$-positive for every $n\in\bN$. A TPCP map is called a
\emph{quantum channel} in quantum information. Furthermore, $\Phi$ is called a
\emph{Schwarz map} if $\Phi(X)^*\Phi(X)\le\Phi(X^*X)$ for all $X\in\cBH$. Note that if $\Phi$ is
$2$-positive, then it is a Schwarz map (however, the converse does not hold). See, e.g.,
\cite[Sec.~1.7]{Hi0} for different positivity notions of $\Phi$.

A function $f:(0,\infty)\to\bR$ is said to be \emph{operator monotone} if, for every
$A,B\in\cBH^{++}$ (of any $\cH$), $A\le B\implies f(A)\le f(B)$, where $f(A),f(B)$ are defined
via usual functional calculus. Also, $f$ is said to be \emph{operator convex} if
$f(\lambda A+(1-\lambda)B)\le\lambda f(A)+(1-\lambda)f(B)$ for all $A,B\in\cBH^{++}$ (of any
$\cH$) and all $\lambda\in(0,1)$. The function $f$ is said to be
\emph{operator monotone decreasing} if $-f$ is operator monotone, and \emph{operator concave}
if $-f$ is operator convex. The definitions are similar for $f$ on $[0,\infty)$, where
$A,B\in\cBH^{++}$ is replaced with $A,B\in\cBH^+$. In fact, when $f$ is on $[0,\infty)$, we
always assume that $\lim_{x\searrow0}f(x)=f(0)$. (Note that an operator monotone or operator
convex function $f$ on $[0,\infty)$ is automatically real analytic on $(0,\infty)$, but not necessarily
continuous at $0$.) Below we write $\OM(0,\infty)$, $\OMD(0,\infty)$ and $\OC(0,\infty)$
for the sets of operator monotone, operator monotone decreasing and operator convex functions on
$(0,\infty)$, respectively. The similar sets of those functions on $[0,\infty)$ are denoted by
$\OM[0,\infty)$ etc. Moreover, we write
\begin{align*}
\OM_+[0,\infty)&:=\{f\in\OM[0,\infty):f(x)>0\ \mbox{for all $x>0$}\}, \\
\OMD_+(0,\infty)&:=\{f\in\OMD(0,\infty):f(x)>0\ \mbox{for all $x>0$}\}.
\end{align*}

\subsection{Definition of quasi-entropies}\label{Sec-2.1}

Let $\rho,\sigma\in\cBH^+$ be given with the spectral decompositions
\begin{align}\label{F-2.1}
\rho=\sum_{a\in\Sp(\rho)}aP_a,\qquad\sigma=\sum_{b\in\Sp(\sigma)}bQ_b,
\end{align}
The \emph{relative modular operator} $\Delta_{\rho,\sigma}$ is defined by
\[
\Delta_{\rho,\sigma}X:=L_\rho R_{\sigma^{-1}}X=\rho X\sigma^{-1},\qquad X\in\cBH,
\]
which is a positive operator on $(\cBH,\<\cdot,\cdot\>_\HS)$ whose support projection is
$L_{\rho^0}R_{\sigma^0}$.

Let us begin with the definition of quasi-entropies $S_f^X(\rho\|\sigma)$ in the case where
$\rho,\sigma\in\cBH^{++}$.

\begin{definition}[\cite{Pe1}]\label{D-2.1}\rm
For any function $f:(0,\infty)\to\bR$, $X\in\cBH$ and $\rho,\sigma\in\cBH^{++}$ the
\emph{quasi-entropy} is defined by
\begin{align}\label{F-2.2}
S_f^X(\rho\|\sigma):=\<X\sigma^{1/2},f(\Delta_{\rho,\sigma})(X\sigma^{1/2})\>_\HS
=\<X,f(L_\rho R_\sigma^{-1})R_\sigma X\>_\HS,
\end{align}
where $f(\Delta_{\rho,\sigma})$ is defined via the usual functional calculus of
$\Delta_{\rho,\sigma}$.
\end{definition}

We write
\begin{align}\label{F-2.3}
J_f(\rho,\sigma):=f(\Delta_{\rho,\sigma})R_\sigma
=f(L_\rho R_\sigma^{-1})R_\sigma,\qquad\rho,\sigma\in\cBH^{++},
\end{align}
whose spectral decomposition is given in terms of \eqref{F-2.1} as
\[
J_f(\rho,\sigma)=\sum_{a\in\Sp(\rho),\,b\in\Sp(\sigma)}
bf\Bigl({a\over b}\Bigr)L_{P_a}R_{Q_b},
\]
so that we can write
\begin{align}\label{F-2.4}
S_f^X(\rho\|\sigma)=\<X,J_f(\rho,\sigma)X\>_\HS
=\sum_{a,b}bf\Bigl({a\over b}\Bigr)\Tr X^*P_aXQ_b.
\end{align}
Note that $J_f(\rho,\sigma)$ is the \emph{operator perspective} of $L_\rho,R_\sigma$ for the
function $f$ (see \cite{Eff,ENG}).

The next theorem is a slight improvement of \cite[Theorems 5 and 6]{HP} (also \cite[Sec.~1.12]{Ca}),
which gives a basis of our discussions on monotonicity of quasi-entropies.

\begin{thm}\label{T-2.2}
For any function $f:(0,\infty)\to(0,\infty)$ the following conditions are equivalent:
\begin{itemize}
\item[(I)] $f\in\OM_+[0,\infty)$ (with extension $f(0):=\lim_{x\searrow0}f(x)$).
\item[(II)] For every TPCP map $\Phi:\cBH\to \cBK$ (for arbitrary $\cH,\cK$) and for every
$\rho,\sigma\in\cBH^{++}$ such that $\Phi(\rho),\Phi(\sigma)\in\cBK^{++}$,
\[
\Phi^*J_f(\Phi(\rho),\Phi(\sigma))^{-1}\Phi\le J_f(\rho,\sigma)^{-1}
\quad\mbox{on $\cBH$}.
\]
\item[(II\,$'$)] For every trace-preserving map $\Phi:\cBH\to \cBK$ (for arbitrary $\cH,\cK$)
such that $\Phi^*$ is a Schwarz map, and for every $\rho,\sigma\in\cBH^{++}$ such that
$\Phi(\rho),\Phi(\sigma)\in\cBK^{++}$,
\[
\Phi^*J_f(\Phi(\rho),\Phi(\sigma))^{-1}\Phi\le J_f(\rho,\sigma)^{-1}
\quad\mbox{on $\cBH$}.
\]
\item[(III)] For every $\Phi$ and $\rho,\sigma$ as in (II),
\[
\Phi J_f(\rho,\sigma)\Phi^*\le J_f(\Phi(\rho),\Phi(\sigma))
\quad\mbox{on $\cBK$}.
\]
\item[(III\,$'$)] For every $\Phi$ and $\rho,\sigma$ as in (II\,$'$),
\[
\Phi J_f(\rho,\sigma)\Phi^*\le J_f(\Phi(\rho),\Phi(\sigma))
\quad\mbox{on $\cBK$}.
\]
\item[(IV)] The map $(\rho,\sigma,X)\mapsto\<X,J_f(\rho,\sigma)^{-1}X\>_\HS$ is jointly
convex on $\cBH^{++}\times \cBH^{++}\times \cBH$.
\end{itemize}
\end{thm}

\begin{proof}
Since it was proved in \cite{HP} that (I), (II), (III) and (IV) are equivalent, we need to show that
(II$'$) and (III$'$) are also equivalent to the other conditions. The proof of (II)$\iff$(III) in \cite{HP}
gives (II$'$)$\iff$(III$'$) as well, and (III$'$)$\implies$(III) is obvious. Hence it suffices to show that
(I)$\implies$(III$'$). But this was indeed proved in \cite[Theorem 4]{Pe1}. (Note that the assumption
$f(0)=0$ in \cite[Theorem 4]{Pe1} is redundant.)
\end{proof}

\begin{remark}\label{R-2.3}\rm
The proof for (I), (II), (III), (IV) in \cite{HP} is based on Lieb and Ruskai's inequality \cite{LR}. Since
this inequality was extended to $2$-positive $\Phi$ in \cite{Choi2}, the proof in \cite{HP} remains
valid when $\Phi$ is $2$-positive (as in \cite{Ca}). However, Lieb and Ruskai's inequality is
not available for $\Phi$ in (II$'$) and (III$'$).
\end{remark}

\begin{lemma}\label{L-2.4}
If $f$ is continuous on $(0,\infty)$, then $(\rho,\sigma,X)\mapsto S_f^X(\rho\|\sigma)$ is
jointly continuous on $\cBH^{++}\times \cBH^{++}\times \cBH$.
\end{lemma}

\begin{proof}
Immediate from the continuity of $(\rho,\sigma)\mapsto\Delta_{\rho,\sigma}$ for
$\rho,\sigma\in\cBH^{++}$ and of continuous functional calculus.
\end{proof}

When $f$ is convex (resp., concave) on $(0,\infty)$, we can define
\[
f(0^+):=\lim_{x\searrow0}f(x),\qquad
f'(\infty):=\lim_{x\to\infty}{f(x)\over x},
\]
which are in $(-\infty,+\infty]$ (resp., in $[-\infty,+\infty)$). The symbol $f'(\infty)$ is
natural since $f'(\infty)=\lim_{x\to\infty}f_+'(x)$ ($f_+'(x)$ being the right-derivative).
Below we will understand the expression $bf({a\over b})$ for $a=0$ or $b=0$ in the following
way:
\begin{align}\label{F-2.5}
bf\Bigl({a\over b}\Bigr):=\begin{cases}f(0^+)b & \text{for $a=0$, $b\ge0$}, \\
f'(\infty)a & \text{for $a>0$, $b=0$},
\end{cases}
\end{align}
where we use the convention that $(+\infty)0:=0$ and $(+\infty)a:=+\infty$ for $a>0$. In
particular, we set $0f({0\over0})=0$. The above expression is justified as
\begin{align}\label{F-2.6}
\lim_{\eps\searrow0}(b+\eps)f\Bigl({a+\eps\over b+\eps}\Bigr)
=\begin{cases}bf({a\over b}) & \text{if $a,b>0$}, \\
bf(0^+) & \text{if $a=0$, $b>0$}, \\
af'(\infty) & \text{if $a>0$, $b=0$}, \\
0 & \text{if $a=b=0$}.
\end{cases}
\end{align}

\begin{lemma}\label{L-2.5}
Assume that $f$ is convex (resp., concave) on $(0,\infty)$. Then for every
$\rho,\sigma\in\cBH^+$ and $X\in\cBH$, the limit
\[
\lim_{\eps\searrow0}S_f^X(\rho+\eps I\|\sigma+\eps I)
\]
exists in $(-\infty,+\infty]$ (resp., in $[-\infty,+\infty)$) and it has the expression
\begin{align}\label{F-2.7}
\lim_{\eps\searrow0}S_f^X(\rho+\eps I\|\sigma+\eps I)
=\sum_{a\in\Sp(\rho),\,b\in\Sp(\sigma)}bf\Bigl({a\over b}\Bigr)\Tr X^*P_aXQ_b
\end{align}
in terms of the spectral decompositions of $\rho,\sigma$ in \eqref{F-2.1}.
\end{lemma}

\begin{proof}
By \eqref{F-2.4} we write
\[
S_f^X(\rho+\eps I\|\sigma+\eps I)
=\sum_{a,b}(b+\eps)f\Bigl({a+\eps\over b+\eps}\Bigr)\Tr X^*P_aXQ_b.
\]
Hence the result follows from \eqref{F-2.6}.
\end{proof}

Now we extend the definition of quasi-entropies to all $\rho,\sigma\in\cBH^{++}$ as follows:

\begin{definition}\label{D-2.6}\rm
When $f$ is convex (or concave) on $(0,\infty)$, for every $\rho,\sigma\in\cBH^+$ and
$X\in\cBH$ we can define $S_f^X(\rho\|\sigma)$ by expression \eqref{F-2.7}. When
$\rho,\sigma\in\cBH^{++}$, this definition is compatible with Definition \ref{D-2.1} by
Lemma \ref{L-2.4}.
\end{definition}


We note that the \emph{standard $f$-divergence} $S_f(\rho\|\sigma)$ is the special case of
quasi-entropies (specialized to $X=I$) as
\[
S_f(\rho\|\sigma)=S_f^I(\rho\|\sigma),\qquad\rho,\sigma\in\cBH^+.
\]

Note that the expression
\[
\<X\sigma^{1/2},f(\Delta_{\rho,\sigma})(X\sigma^{1/2})\>_\HS
\]
makes sense for any $\rho,\sigma\in\cBH^+$ and any function $f$ on $[0,\infty)$ by regarding
$f(\Delta_{\rho,\sigma})$ as the functional calculus of $\Delta_{\rho,\sigma}$. The next lemma
gives the precise relation between the above expression and $S_f^X(\rho\|\sigma)$, which will be
useful in our later discussions. In fact, one can directly define $S_f^X(\rho\|\sigma)$ for general
$\rho,\sigma\in\cBH^+$ by \eqref{F-2.8}, extending \eqref{F-2.2} when
$\rho,\sigma\in\cBH^{++}$.

\begin{lemma}\label{L-2.7}
Let $f$ be a convex or concave function on $[0,\infty)$ with $f(0)=f(0^+)\in\bR$. For every
$\rho,\sigma\in\cBH^+$ and $X\in\cBH$,
\begin{align}\label{F-2.8}
S_f^X(\rho\|\sigma)
=\<X\sigma^{1/2},f(\Delta_{\rho,\sigma})(X\sigma^{1/2})\>_\HS
+f'(\infty)\Tr X^*\rho X(I-\sigma^0).
\end{align}
Hence, if $f'(\infty)=0$ or $\sigma^0=I$, then
$S_f^X(\rho\|\sigma)=\<X\sigma^{1/2},f(\Delta_{\rho,\sigma})(X\sigma^{1/2})\>_\HS$.
\end{lemma}

\begin{proof}
For $\rho,\sigma\in\cBH^+$ with the spectral decomposition \eqref{F-2.1}, by Definition
\ref{D-2.6} we have
\begin{equation}\label{F-2.9}
\begin{aligned}
S_f^X(\rho\|\sigma)&=\sum_{a,b}b\Bigl({a\over b}\Bigr)\Tr X^*P_aXQ_b \\
&=\sum_{a,b>0}bf\Bigl({a\over b}\Bigr)\Tr X^*P_aXQ_b
+\sum_{a=0,\,b>0}bf(0)\Tr X^*P_aXQ_b \\
&\qquad+\sum_{a>0,\,b=0}af'(\infty)\Tr X^*P_aXQ_b \\
&=\sum_{a,b>0}bf\Bigl({a\over b}\Bigr)\Tr X^*P_aXQ_b \\
&\qquad+f(0)\Tr X^*(I-\rho^0)X\sigma+f'(\infty)\Tr X^*\rho X(I-\sigma^0),
\end{aligned}
\end{equation}
where $bf(a/b)$ for $a=0$ or $b=0$ is given in \eqref{F-2.5}. On the other hand,
since the spectral decomposition of $\Delta_{\rho,\sigma}$ is
\[
\Delta_{\rho,\sigma}=\sum_{a,b>0}{a\over b}L_{P_a}R_{Q_b}
+0\sum_{\mbox{\scriptsize$a=0$ or $b=0$}}L_{P_a}R_{Q_b},
\]
one can write
\[
f(\Delta_{\rho,\sigma})=\sum_{a,b>0}f\Bigl({a\over b}\Bigr)L_{P_a}R_{Q_b}
+f(0)\sum_{\mbox{\scriptsize$a=0$ or $b=0$}}L_{P_a}R_{Q_b}
\]
and hence
\begin{equation}\label{F-2.10}
\begin{aligned}
&\<X\sigma^{1/2},f(\Delta_{\rho,\sigma})X\sigma^{1/2}\>_\HS \\
&\quad=\sum_{a,b>0}f\Bigl({a\over b}\Bigr)\Tr\sigma^{1/2}X^*P_aX\sigma^{1/2}Q_b
+f(0)\sum_{\mbox{\scriptsize$a=0$ or $b=0$}}\Tr\sigma^{1/2}X^*P_aX\sigma^{1/2}Q_b \\
&\quad=\sum_{a,b>0}bf\Bigl({a\over b}\Bigr)\Tr X^*P_aXQ_b
+f(0)\sum_{a=0,\,b>0}b\Tr X^*P_aXQ_b \\
&\quad=\sum_{a,b>0}bf\Bigl({a\over b}\Bigr)\Tr X^*P_aXQ_b
+f(0)\Tr X^*(I-\rho^0)X\sigma.
\end{aligned}
\end{equation}
Therefore, \eqref{F-2.8} follows from \eqref{F-2.9} and \eqref{F-2.10}. The latter
assertion is now immediate.
\end{proof}

Concerning the term $\Tr X^*\rho X(I-\sigma^0)$ in \eqref{F-2.8}, another simple lemma is
included here for later use.

\begin{lemma}\label{L-2.8}
For any $\rho,\sigma\in\cB(\cH)^+$ and $X\in\cBH$, the following conditions are equivalent:
\begin{itemize}
\item[(a)] $\Tr X^*\rho X(I-\sigma^0)=0$;
\item[(b)] $\rho^0X(I-\sigma^0)=0$;
\item[(c)] $\sigma^0X^*\rho X\sigma^0=X^*\rho X$.
\end{itemize}
\end{lemma}

\begin{proof}
Assume (a). Since $\Tr(\rho^{1/2}(I-\sigma^0)X)^*(\rho^{1/2}X(I-\sigma^0))=0$, one has
$\rho^{1/2}X(I-\sigma^0)=0$ and hence (b) follows. If (b) holds, then 
\[
\sigma^0X^*\rho X\sigma^0=(\sigma^0X^*\rho^0)\rho(\rho^0 X\sigma^0)
=(X^*\rho^0)\rho(\rho^0 X)=X^*\rho X.
\]
Hence (b)$\implies$(c). Finally, since
\[
\Tr X^*\rho X(I-\sigma^0)=\Tr X^*\rho X-\Tr\sigma^0X^*\rho X\sigma^0,
\]
(c)$\implies$(a) follows.
\end{proof}

\begin{remark}\label{R-2.9}\rm
It might be instructive to make a historical comment on the definition of quasi-entropies. Petz
\cite{Pe1} introduced the quasi-entropy $S_f^X(\rho\|\sigma)$ with \eqref{F-2.2}, while he
restricted $\rho,\sigma$ to operators in $\cBH^{++}$. To extend the definition to all
$\rho,\sigma\in\cBH^+$, we take the limit in Lemma \ref{L-2.5}, and as it turns out,
$S_f^X(\rho\|\sigma)$ has the expression in the RHS of \eqref{F-2.7} in terms of the spectral
decompositions of $\rho,\sigma$ with specification \eqref{F-2.5}. This way of definition is
compatible with that of standard $f$-divergences in \cite{HMPB,HM}. On the other hand,
$S_f^X(\rho\|\sigma)$ was defined in \cite{Sh,Ve} by \eqref{F-2.2} directly for all
$\rho,\sigma\in\cBH^+$, which does not seem a rigorous way. Although the expression in the
RHS of \eqref{F-2.7} was mentioned in \cite{Sh,Ve}, the boundary values in \eqref{F-2.5} when
$a=0$ or $b=0$ are not explicitly specified there. Here, we stress the boundary term
$f'(\infty)\Tr X^*\rho X(I-\sigma^0)$ that appears in \eqref{F-2.8}, which will play a crucial role in
our discussions below. Without this term, for instance, $S_{x\log x}^I(\rho\|\sigma)$ cannot
coincide with the relative entropy $D(\rho\|\sigma)$ in \eqref{F-1.1}.
\end{remark}

\begin{remark}\label{R-2.10}\rm
One can extend $J_f(\rho,\sigma)$ in \eqref{F-2.4} and Theorem \ref{T-2.2} for
$\rho,\sigma\in\cBH^{++}$ to all $\rho,\sigma\in\cBH^+$ as
\[
J_f(\rho,\sigma):=\sum_{a,b}b\Bigl({a\over b}\Bigr)L_{P_a}R_{Q_b}
\]
with convention \eqref{F-2.5}, which is though no longer bounded, but lower semibounded when
$f$ is convex on $(0,\infty)$. Then one may regard expression \eqref{F-2.7} as
$\<X,J_f(\rho,\sigma)X\>_\HS$ ($\in(-\infty,+\infty]$). It might be more natural to consider the
operator perspective of $\rho,\sigma$ (instead of commuting $L_\rho,R_\sigma$). Indeed, for
any convex function $f$ on $(0,\infty)$, the extended operator perspective, denoted by
$\phi_f(\rho,\sigma)$, was recently introduced in \cite[Sec.~7]{HUW} for every
$\rho,\sigma\in\cBH^+$ in the method of the Pusz--Woronowicz functional calculus. Then one
can think of another type of quasi-entropy given by
$\<X,\phi_f(\rho,\sigma)X\>_\HS=\Tr X^*\phi(\rho,\sigma)X$ ($\in(-\infty,+\infty]$). When $f$ is
operator convex and $X=I$, we note that $\Tr\phi_f(\rho,\sigma)$ becomes the so-called
\emph{maximal $f$-divergence} studied in \cite{Ma,HM,Hi2}, which is different from the standard
$f$-divergence $S_f(\rho\|\sigma)$.
\end{remark}

\subsection{Monotonicity properties of quasi-entropies}\label{Sec-2.2}

In this subsection we improve, based on the definition described in Sec.~\ref{Sec-2.1}, the
monotonicity properties (DPIs) of quasi-entropies, as well as joint convexity/concavity and lower
semicontinuity, formerly shown in \cite{Pe1}.

\begin{prop}\label{P-2.11}
Assume that $f$ is operator convex (resp., operator concave) on $(0,\infty)$. Then for every
$X\in\cBH$, $(\rho,\sigma)\mapsto S_f^X(\rho\|\sigma)$ is jointly convex (resp., jointly
concave) and jointly subadditive (resp., jointly superadditive) on $\cBH^+\times \cBH^+$.
\end{prop}

\begin{proof}
Assume that $f\in\OC(0,\infty)$. From Definitions \ref{D-2.1} and \ref{D-2.6}
it is easily verified that $S_f^X(\rho\|\sigma)$ is positively homogeneous in $(\rho,\sigma)$, i.e.,
$S_f^X(\lambda\rho\|\lambda\sigma)=\lambda S_f^X(\rho\|\sigma)$ for all $\lambda\ge0$ and
$\rho,\sigma\in\cBH^+$. So we may only show joint subadditivity. Let
$\rho_i,\sigma_i\in\cBH^{++}$, $i=1,2$. Since the operator perspective for an operator convex
function is jointly convex (equivalently, jointly subadditive) (see \cite{Eff,ENG}), we have
\[
f(L_{\rho_1+\rho_2}R_{\sigma_1+\sigma_2}^{-1})R_{\sigma_1+\sigma_2}
\le f(L_{\rho_1}R_{\sigma_1}^{-1})R_{\sigma_1}
+f(L_{\rho_2}R_{\sigma_2}^{-1})R_{\sigma_2},
\]
and hence
\[
S_f^X(\rho_1+\rho_2\|\sigma_1+\sigma_2)
\le S_f^X(\rho_1\|\sigma_1)+S_f^X(\rho_2\|\sigma_2).
\]
For general $\rho_i,\sigma_i\in\cBH^+$, by Definition \ref{D-2.6} we have
\begin{align*}
S_f^X(\rho_1+\rho_2\|\sigma_1+\sigma_2)
&=\lim_{\eps\searrow0}S_f^X(\rho_1+\rho_2+2\eps I\|\sigma_1+\sigma_2+2\eps I) \\
&\le\lim_{\eps\searrow0}\bigl[S_f^X(\rho_1+\eps I\|\sigma_1+\eps I)
+S_f^X(\rho_2+\eps I\|\sigma_2+\eps I)\bigr] \\
&=S_f^X(\rho_1\|\sigma_1)+S_f^X(\rho_2\|\sigma_2).
\end{align*}
The proof for operator concave $f$ is similar (or just apply the above case to $-f$).
\end{proof}

\begin{lemma}\label{L-2.12}
Assume that $f\in\OC(0,\infty)$ with $f(1)=0$. Then for every $\rho,\sigma\in\cBH^+$
and $X\in\cBH$,
\[
S_f^X(\rho\|\sigma)=\sup_{\eps>0}S_f^X(\rho+\eps I\|\sigma+\eps I).
\]
\end{lemma}

\begin{proof}
Note that $S_f^X(I\|I)=0$ from the assumption $f(1)=0$. For $0<\eps'<\eps$ we have by
Proposition \ref{P-2.11}
\[
S_f^X(\rho+\eps I\|\sigma+\eps I)
\le S_f^X(\rho+\eps'I\|\sigma+\eps'I)+(\eps-\eps')S_f^X(I\|I)
=S_f^X(\rho+\eps'I\|\sigma+\eps'I).
\]
Hence
\[
S_f^X(\rho\|\sigma)=\lim_{\eps\searrow0}S_f^X(\rho+\eps I\|\sigma+\eps I)
=\sup_{\eps>0}S_f^X(\rho+\eps I\|\sigma+\eps I),
\]
as desired.
\end{proof}

\begin{prop}\label{P-2.13}
If $f\in\OC(0,\infty)$, then $(\rho,\sigma,X)\mapsto S_f^X(\rho\|\sigma)$ is jointly lower
continuous on $\cBH^+\times \cBH^+\times \cBH$.
\end{prop}

\begin{proof}
Note that
\begin{align*}
S_f^X(\rho\|\sigma)&=S_{f-f(1)}^X(\rho\|\sigma)+S_{f(1)}^X(\rho\|\sigma) \\
&=\sup_{\eps>0}S_{f-f(1)}^X(\rho+\eps I\|\sigma+\eps I)+f(1)\Tr\sigma X^*X
\end{align*}
by Lemma \ref{L-2.12}. It follows from Lemma \ref{L-2.4} that
$(\rho,\sigma,X)\mapsto S_f^X(\rho+\eps I\|\sigma+\eps I)$ is jointly continuous on
$\cBH^+\times \cBH^+\times \cBH$. Also $(\sigma,X)\mapsto\Tr\sigma X^*X$ is jointly
continuous on $\cBH^+\times \cBH$. Hence the result follows.
\end{proof}

The next theorem improves \cite[Theorem 4]{Pe1}.

\begin{thm}\label{T-2.14}
Assume that $h\in\OM_+[0,\infty)$. Let $\Phi:\cBH\to \cBK$ be a trace-preserving map such that
$\Phi^*$ is a Schwarz map. Then for every $\rho,\sigma\in\cBH^+$ and $K\in\cBK$,
\begin{align}\label{F-2.11}
S_h^K(\Phi(\rho)\|\Phi(\sigma))\ge S_h^{\Phi^*(K)}(\rho\|\sigma).
\end{align}
\end{thm}

\begin{proof}
In the case that $\rho,\sigma\in\cBH^{++}$ and $\Phi(\rho),\Phi(\sigma)\in\cBK^{++}$, the
result is included in Theorem \ref{T-2.2} (as (I)$\implies$(III$'$)). Let now $\rho,\sigma\in\cBH^+$
be arbitrary. For any $\delta\in(0,1)$ define a unital map $\Psi_\delta:\cBK\to \cBH$ by
\[
\Psi_\delta(Y):=(1-\delta)\Phi^*(Y)+\delta\tau(Y)I_\cH,\qquad Y\in\cBK,
\]
where $\tau:=(\dim\cK)^{-1}\Tr$ is the normalized trace on $\cBK$. Note that $\Psi_\delta$ is a
Schwarz map since
\begin{align*}
\Psi_\delta(Y^*Y)-\Psi_\delta(Y)^*\Psi_\delta(Y)
&\ge(1-\delta)\Phi^*(Y)^*\Phi^*(Y)+\delta|\tau(Y)|^2I
-\Psi_\delta(Y)^*\Psi_\delta(Y) \\
&=\delta(1-\delta)[\Phi^*(Y)-\tau(Y)I]^*[\Phi^*(Y)-\tau(Y)I]\ge0.
\end{align*}
Hence $\Phi_\delta:=\Psi_\delta^*:\cBH\to \cBK$ is a trace-preserving map such that
$\Phi_\delta^*=\Psi_\delta$ is a Schwarz map. Moreover, since the dual map of
$Y\in\cBK\mapsto\tau(Y)I_\cH$ is $X\in\cBH\mapsto(\dim\cK)^{-1}(\Tr X)I_\cK$,
$\Phi_\delta(\rho+\eps I)$ and $\Phi_\delta(\sigma+\eps I)$ are in $\cBK^{++}$
for any $\eps>0$. Therefore, from the above case we see that
\[
S_h^K(\Phi_\delta(\rho+\eps I)\|\Phi_\delta(\sigma+\eps I))
\ge S_h^{\Phi_\delta^*(K)}(\rho+\eps I\|\sigma+\eps I).
\]
By Proposition \ref{P-2.13} applied to operator convex $f=-h$ and Lemma \ref{L-2.4} we have
\begin{align*}
S_h^K(\Phi(\rho+\eps I)\|\Phi(\sigma+\eps I))
&\ge\limsup_{\delta\searrow0}
S_h^K(\Phi_\delta(\rho+\eps I)\|\Phi_\delta(\sigma+\eps I)) \\
&\ge\lim_{\delta\searrow0}S_h^{\Phi_\delta^*(K)}(\rho+\eps I\|\sigma+\eps I) \\
&=S_h^{\Phi^*(K)}(\rho+\eps I\|\sigma+\eps I).
\end{align*}
By Proposition \ref{P-2.13} again and Definition \ref{D-2.6}, letting $\eps\searrow0$ yields
\eqref{F-2.11}.
\end{proof}

\begin{remark}\label{R-2.15}\rm
The assumption $h(0)\ge0$ is essential in Theorem \ref{T-2.14}. Indeed, when $h(x)\equiv\alpha$,
inequality \eqref{F-2.11} reduces to
$\alpha\Tr\sigma\Phi^*(K^*K)\ge\alpha\Tr\sigma\Phi^*(K)^*\Phi^*(K)$, which is not necessarily true
for $\alpha<0$.
\end{remark}

The \emph{multiplicative domain} of a unital Schwarz map $\Psi:\cBK\to \cBH$ is defined by
\begin{align}\label{F-2.12}
\cM_\Psi:=\{X\in\cBK:\Psi(X^*X)=\Psi(X)^*\Psi(X),\,\Psi(XX^*)=\Psi(X)\Psi(X)^*\}.
\end{align}
An important fact proved by Choi \cite{Choi} for unital $2$-positive maps and then extended
in \cite[Lemma 3.9]{HMPB} to the case of unital Schwarz maps is the following:
\begin{equation}\label{F-2.13}
\begin{aligned}
\cM_\Psi&=\{X\in\cBK:\Psi(YX)=\Psi(Y)\Psi(X), \\
&\hskip2.7cm\Psi(XY)=\Psi(X)\Psi(Y)\ \mbox{for all $Y\in\cBK$}\},
\end{aligned}
\end{equation}
and hence $\cM_\Psi$ is a $C^*$-subalgebra of $\cBK$. More specifically,
$\Psi(X^*X)=\Psi(X)^*\Psi(X)$ (resp., $\Psi(XX^*)=\Psi(X)\Psi(X)^*$) holds if and
only if $\Psi(YX)=\Psi(Y)\Psi(X)$ (resp., $\Psi(XY)=\Psi(X)\Psi(Y)$) for all $Y\in\cBK$.

The next theorem improves \cite[Theorem 1]{Pe1}.

\begin{thm}\label{T-2.16}
Assume that $f\in\OC(0,\infty)$. Let $\Phi:\cBH\to \cBK$ be a trace-preserving map such that
$\Phi^*$ is a Schwarz map. Then for every $\rho,\sigma\in\cBH^+$ and any $K\in\cM_{\Phi^*}$,
\begin{align}\label{F-2.14}
S_f^K(\Phi(\rho)\|\Phi(\sigma))\le S_f^{\Phi^*(K)}(\rho\|\sigma).
\end{align}
\end{thm}

\begin{proof}
According to \cite[Sec.~III]{Hi1} (also \cite[Sec.~2.2]{Hi2}), one has a sequence of functions
$f_n\in\OC(0,\infty)$ of the form
\begin{align}\label{F-2.15}
f_n(x)=f_n(0^+)+f_n'(\infty)x-h_n(x),\qquad x\in(0,\infty),
\end{align}
where $f_n(0^+)$, $f_n'(\infty)$ are finite and $h_n$ is in $\OM[0,\infty)$ with
$h_n(0)=h_n'(\infty)=0$ such that
\begin{align}\label{F-2.16}
f_n(0^+)\,\nearrow\,f(0^+),\quad
f_n'(\infty)\,\nearrow\,f'(\infty),\quad
f_n(x)\,\nearrow\,f(x),\ \ x\in(0,\infty).
\end{align}
From the assumption $K\in\cM_{\Phi^*}$ and Proposition \ref{T-2.14} one has
\begin{align*}
S_{f_n}^K(\Phi(\rho)\|\Phi(\sigma))
&=f_n(0^+)\Tr\Phi(\sigma)K^*K+f_n'(\infty)\Tr\Phi(\rho)KK^*
-S_{h_n}^K(\Phi(\rho)\|\Phi(\sigma)) \\
&\le f_n(0^+)\Tr\sigma\Phi^*(K)^*\Phi^*(K)+f_n'(\infty)\Tr\rho\Phi^*(K)\Phi^*(K)
-S_{h_n}^{\Phi^*(K)}(\rho\|\sigma) \\
&=S_{f_n}^{\Phi^*(K)}(\rho\|\sigma).
\end{align*}
With the spectral decompositions in \eqref{F-2.1} we see by expression \eqref{F-2.4} that
\begin{align*}
S_{f_n}^{\Phi^*(K)}(\rho\|\sigma)
&=\sum_{a,b}bf_n\Bigl({a\over b}\Bigr)\Tr\Phi^*(K)^*P_a\Phi^*(K)Q_b \\
&\nearrow\,\sum_{a,b}bf\Bigl({a\over b}\Bigr)\Tr\Phi^*(K)^*P_a\Phi^*(K)Q_b
=S_f^{\Phi^*(K)}(\rho\|\sigma)
\end{align*}
as $n\to\infty$ thanks to \eqref{F-2.16}. Similarly, by taking the spectral decompositions of
$\Phi(\rho),\Phi(\sigma)$ we have
$S_{f_n}^K(\Phi(\rho)\|\Phi(\sigma))\nearrow S_f^K(\Phi(\rho)\|\Phi(\sigma))$ as $n\to\infty$.
Therefore, the required inequality follows.
\end{proof}

\begin{remark}\label{R-2.17}\rm
The assumption $K\in\cM_{\Phi^*}$ is indeed necessary to guarantee \eqref{F-2.14} for
$f(x)\equiv1$ and $f(x)=x$, so Theorem \ref{T-2.16} seems the best possible result on monotonicity
(DPI) of quasi-entropies for operator convex $f$.
\end{remark}

\begin{remark}\label{R-2.18}\rm
Theorem \ref{T-2.16} was shown in \cite{Pe1} when $\Phi$ is a conditional expectation, and in
\cite{Sh,Ve} when $\Phi$ is a partial trace. In fact, a slightly more general monotonicity
\begin{align}\label{F-2.17}
S_f^K(\T_\cK\rho\|\T_\cK\sigma)\le S_f^{K\otimes V}(\rho\|\sigma)
\end{align}
was shown in \cite{Ve} for $\rho,\sigma\in\cBH^+$, $K\in\cBH$ and any unitary $V$ on $\cK$,
where $\T_\cK:\cB(\cH\otimes\cK)\to\cBH$ is the partial trace over $\cK$ and $f$ is
in $\OC(0,\infty)$. Since, as easily verified,
\begin{equation}\label{F-2.18}
\begin{aligned}
\T_\cK\rho&=\T_\cK((I_\cH\otimes V)^*\rho(I_\cH\otimes V)), \\
S_f^{K\otimes V}(\rho\|\sigma)
&=S_f^{K\otimes I_\cK}((I_\cH\otimes V)^*\rho(I_\cH\otimes V)\|\sigma),
\end{aligned}
\end{equation}
we note that \eqref{F-2.17} can be reduced to
$S_f^K(\T_\cK\rho\|\T_\cK\sigma)\le S_f^{K\otimes I_\cK}(\rho\|\sigma)$, that is \eqref{F-2.14} for
$\Phi=\T_\cK$. In connection with this remark, see also Remark \ref{R-3.5} in Sec.~\ref{Sec-3.1}.
\end{remark}

The main properties of the standard $f$-divergence $S_f(\rho\|\sigma)$ -- joint convexity, joint
lower semicontinuity and monotonicity (DPI) -- shown in \cite{HMPB,HM,Hi1} are also recaptured
by Propositions \ref{P-2.11}, \ref{P-2.13} and Theorem \ref{T-2.16}.

\subsection{Lieb's concavity/convexity and Ando's convexity, revisited}\label{Sec-2.3}

For each $\alpha\in\bR$ we define the function $f_\alpha$ on $(0,\infty)$ by
$f_\alpha(x):=x^\alpha$ for $x\in(0,\infty)$, which continuously extends to $[0,\infty)$ when
$\alpha\ge0$. As is well known, $f_\alpha\in\OM_+[0,\infty)$ if $\alpha\in[0,1]$, and
$f_\alpha\in\OC(0,\infty)$ if $\alpha\in[-1,0]\cup[1,2]$. For any $\rho,\sigma\in\cBH^+$
and any $X\in\cBH$ we write
\[
\<X\sigma^{1/2},f_\alpha(\Delta_{\rho,\sigma})(X\sigma^{1/2})\>_\HS
=\begin{cases}\Tr X^*\rho^\alpha X\sigma^{1-\alpha} & \text{if $\alpha\ne0$}, \\
\Tr\sigma X^*X & \text{if $\alpha=0$},
\end{cases}
\]
where $\sigma^0$ (when $\alpha=1$) is the support projection of $\sigma$. Then, in view of
Lemma \ref{L-2.7}, the following can be regarded as special cases of Proposition \ref{P-2.11}:
\begin{itemize}
\item \emph{Lieb's concavity} \cite[Corollary 1.1]{Li}: The function
$(\rho,\sigma)\mapsto\Tr X^*\rho^\alpha X\sigma^{1-\alpha}$ is jointly concave on
$\cBH^+\times \cBH^+$ for every $\alpha\in[0,1]$ and any $X\in\cBH$.
\item \emph{Ando's convexity} \cite[p.~221]{An}: The function
$(\rho,\sigma)\mapsto\Tr X^*\rho^\alpha X\sigma^{1-\alpha}$ is jointly convex on
$\cBH^{++}\times \cBH^{++}$ for every $\alpha\in[-1,0]\cup[1,2]$ and any $X\in\cBH$.
\end{itemize}
Furthermore, Theorem \ref{T-2.2}\,(IV) for $h(x)=x^\alpha$, $\alpha\in[0,1]$, means
\begin{itemize}
\item \emph{Lieb's $3$-variable convexity} \cite[Corollary 2.1]{Li}: The function
$(\rho,\sigma,X)\mapsto\Tr X^*\rho^{-\alpha}X\sigma^{\alpha-1}$ is jointly concave on
$\in\cBH^{++}\times\cBH^{++}\times\cBH$ for every $\alpha\in[0,1]$.
\end{itemize}

The inequalities in the next corollary are the monotonicity versions of Lieb's
concavity/convexity and Ando's convexity.

\begin{cor}\label{C-2.19}
Let $\Phi:\cBH\to \cBK$ be a trace-preserving map such that $\Phi^*$ is a Schwarz map.
\begin{itemize}
\item[(1)] For every $\alpha\in[0,1]$, any $\rho,\sigma\in\cBH^+$ and any $K\in\cBK$, 
\[
\Tr K^*\Phi(\rho)^\alpha K\Phi(\sigma)^{1-\alpha}
\ge\Tr\Phi^*(K)^*\rho^\alpha\Phi^*(K)\sigma^{1-\alpha}.
\]
\item[(2)] For every $\alpha\in[-1,0]\cup[1,2]$, any $\rho,\sigma\in\cBH^{++}$ such that
$\Phi(\rho),\Phi(\sigma)\in\cBK^{++}$, and for any $K\in\cM_{\Phi^*}$,
\begin{align}\label{F-2.19}
\Tr K^*\Phi(\rho)^\alpha K\Phi(\sigma)^{1-\alpha}
\le\Tr\Phi^*(K)^*\rho^\alpha\Phi^*(K)\sigma^{1-\alpha}.
\end{align}
\item[(3)] For every $\alpha\in[0,1]$, any $\rho,\sigma\in\cBH^{++}$ such that
$\Phi^*(\rho),\Phi^*(\sigma)\in\cBK^{++}$, and for any $K\in\cBH$,
\[
\Tr\Phi(K)^*\Phi(\rho)^\alpha\Phi(K)\Phi(\sigma)^{\alpha-1}
\le\Tr K^*\rho^\alpha K\sigma^{\alpha-1}.
\]
\end{itemize}
\end{cor}

\begin{proof}
For $\alpha\in(0,1)$ the assertion (1) is a consequence of Theorem \ref{T-2.14} by
Lemma \ref{L-2.7}, while the $\alpha=0,1$ cases follow by taking the limits from the relevant
inequality for $\alpha\in(0,1)$ as $\alpha\searrow0$ or $\alpha\nearrow1$. (2) and (3) are
consequences of Theorem \ref{T-2.16} and Theorem \ref{T-2.2}\,(II$'$), respectively.
\end{proof}

Corollary \ref{C-2.19} slightly improves monotonicity results in Carlen's recent survey paper
\cite{Ca}. For instance, \eqref{F-2.19} was given in \cite[Theorem 1.10]{Ca} in the restricted case
that $\Phi$ is a partial trace (note that $\cM_{\Phi^*}=\cBK$ in this case, see the proof of
Corollary \ref{C-3.4}).

\begin{remark}\label{R-2.20}\rm
Corollary \ref{C-2.19}\,(2) holds true for slightly more general $\rho,\sigma$. When $\alpha\in[1,2]$,
by Theorem \ref{T-2.16} and Lemma \ref{L-2.7} it holds for any $\rho\in\cBH^+$ and
$\sigma\in\cBH^{++}$ such that $\Phi(\sigma)\in\cBK^{++}$. Furthermore, when $\alpha=1$,
equality holds in \eqref{F-2.19} for any $\rho\in\cBH^+$ and
$\sigma\in\cBH^{++}$. Indeed,
\begin{align*}
\Tr K^*\Phi(\rho)K\Phi(\sigma)^0&\le\Tr K^*\Phi(\rho)K=\Tr\rho\Phi^*(KK^*) \\
&=\Tr\rho\Phi^*(K)\Phi^*(K)^*=\Tr\Phi^*(K)^*\rho\Phi^*(K)\sigma^0,
\end{align*}
where the second equality is due to $K\in\cM_{\Phi^*}$, and the reverse inequality is by Corollary
\ref{C-2.19}\,(1). On the other hand, the restriction of $K$ to $K\in\cM_{\Phi^*}$ is essential in
Corollary \ref{C-2.19}\,(2) (see Example \ref{E-3.18} below).
\end{remark}

\section{Equality cases in monotonicity of quasi-entropies}\label{Sec-3}

In this section we present several equivalent conditions for equality cases in the monotonicity
inequalities of quasi-entropies given in Theorems \ref{T-2.14} and \ref{T-2.16} for operator
monotone/convex functions as well as in Lieb's concavity and Ando's convexity described in
Sec.~\ref{Sec-2.3}.

First of all, we recall the integral expressions of operator monotone functions and operator convex
functions on $[0,\infty)$. It is well known (see, e.g., \cite{Bh,Hi0}) that a function
$h\in\OM[0,\infty)$ has the integral expression
\begin{align}\label{F-3.1}
h(x)=h(0)+h'(\infty)x+\int_{(0,\infty)}{x(1+t)\over x+t}\,d\mu_h(t),
\qquad x\in[0,\infty),
\end{align}
where $\mu_h$ is a (unique) finite positive measure on $(0,\infty)$. We write $\supp\mu_h$ for
the (topological) support of $\mu_h$ consisting of $t\in(0,\infty)$ such that
$\mu_h((t-\eps,t+\eps))>0$ for any $\eps>0$.

It is also known (see \cite[Theorem 8.1]{HMPB}) that a function $f\in\OC[0,\infty)$ has
the integral expression
\begin{align}\label{F-3.2}
f(x)=f(0)+ax+bx^2+\int_{(0,\infty)}\biggl({x\over1+t}-{x\over x+t}\biggr)\,d\nu_f(t),
\qquad x\in[0,\infty),
\end{align}
where $a\in\bR$, $b\ge0$ and $\nu_f$ is a positive measure with
$\int_{(0,\infty)}(1+t)^{-2}\,d\nu_f(t)<+\infty$ ($a,b$ and $\nu_f$ are unique). The support
$\supp\nu_f$ of $\nu_f$ is similar to that of $\mu_h$ above. For notational convenience, we set
for $t\in(0,\infty)$,
\begin{align}
\ffi_t(x)&:={x\over x+t}=1-{t\over x+t},\hskip2.8cm x\in[0,\infty), \label{F-3.3}\\
\psi_t(x)&:={x\over1+t}-{x\over x+t}={x\over1+t}-\ffi_t(x),\qquad x\in[0,\infty). \label{F-3.4}
\end{align}

\subsection{List of equality conditions and theorems}\label{Sec-3.1}

Throughout this section, let $\Phi:\cBH\to \cBK$ be a trace-preserving map such that $\Phi^*$
is a Schwarz map. Moreover, let $\rho,\sigma\in\cBH^+$ and $K\in\cBK$. Let us first
summarize several equality conditions for the monotonicity of quasi-entropies with operator
monotone functions and for the monotonicity version of Lieb's concavity theorem as follows:
\begin{itemize}
\item[(i)] $S_h^K(\Phi(\rho)\|\Phi(\sigma))=S_h^{\Phi^*(K)}(\rho\|\sigma)$ for all
$h\in\OM[0,\infty)$.
\item[(ii)] $S_h^K(\Phi(\rho)\|\Phi(\sigma))=S_h^{\Phi^*(K)}(\rho\|\sigma)$ for all
$h\in\OM[0,\infty)$ with $h(0)=h'(\infty)=0$.
\item[(iii)] $S_h^K(\Phi(\rho)\|\Phi(\sigma))=S_h^{\Phi^*(K)}(\rho\|\sigma)$ for some
$h\in\OM_+[0,\infty)$ such that $\supp\mu_h$ has a limit point in $(0,\infty)$.
\item[(iv)] $S_h^K(\Phi(\rho)\|\Phi(\sigma))=S_h^{\Phi^*(K)}(\rho\|\sigma)$ for some
$h\in\OM_+[0,\infty)$ such that
\begin{align}\label{F-3.5}
|\supp\mu_h|\ge|\Sp(\Delta_{\rho,\sigma})\cup\Sp(\Delta_{\Phi(\rho),\Phi(\sigma)})|.
\end{align}
\item[(v)] $\Tr Y\Phi(\rho)^zK\Phi(\sigma)^{1-z}
=\Tr\Phi^*(Y)\rho^z\Phi^*(K)\sigma^{1-z}$ for all $Y\in\cBK$ and all $z\in\bC$.
\item[(vi)] $\Tr K^*\Phi(\rho)^z K\Phi(\sigma)^{1-z}
=\Tr\Phi^*(K)^*\rho^z\Phi^*(K)\sigma^{1-z}$ for all $z\in\bC$.
\item[(vii)] $\Tr K^*\Phi(\rho)^\alpha K\Phi(\sigma)^{1-\alpha}
=\Tr\Phi^*(K)^*\rho^\alpha\Phi^*(K)\sigma^{1-\alpha}$ for some $\alpha\in(0,1)$.
\item[(viii)] $\sigma^0\Phi^*(Y\Phi(\rho)^zK\Phi(\sigma)^{-z})\sigma^0
=\sigma^0\Phi^*(Y)\rho^z\Phi^*(K)\sigma^{-z}$ for all $Y\in\cBK$ and all $z\in\bC$.
\item[(ix)] $\sigma^0\Phi^*(K^*\Phi(\rho)^\alpha K\Phi(\sigma)^{-\alpha})\sigma^0
=\sigma^0\Phi^*(K)^*\rho^\alpha\Phi^*(K)\sigma^{-\alpha}$ for some $\alpha\in(0,1)$.
\item[(x)] $\Phi^*(\Phi(\rho)^{it}K\Phi(\sigma)^{-it})\sigma^0=\rho^{it}\Phi^*(K)\sigma^{-it}$
for all $t\in\bR$.
\end{itemize}

In conditions (v)--(x), $\sigma^z$ and $\Phi(\sigma)^z$ for $z\in\bC$ are defined with restriction
to the respective supports of $\sigma$ and $\Phi(\sigma)$, and similarly for $\rho^z$ and
$\Phi(\rho)^z$.

In addition to the equality conditions (i)--(x) listed above, we consider the following equality
conditions for the monotonicity of quasi-entropies with operator convex functions and for the
monotonicity version of Ando's convexity:
\begin{itemize}
\item[(i$'$)] $S_f^K(\Phi(\rho)\|\Phi(\sigma))=S_f^{\Phi^*(K)}(\rho\|\sigma)$ for all
$f\in\OC(0,\infty)$.
\item[(iv$'$)] $S_f^K(\Phi(\rho)\|\Phi(\sigma))=S_f^{\Phi^*(K)}(\rho\|\sigma)<+\infty$ for some
$f\in\OC[0,\infty)$ with $f(0)=0$ and $f'(\infty)=+\infty$ such that
\eqref{F-3.5} is satisfied for $\nu_f$ in place of $\mu_h$.
\item[(vii$'$)] $\Tr K^*\Phi(\rho)^\alpha K\Phi(\sigma)^{1-\alpha}=
\Tr\Phi^*(K)^*\rho^\alpha\Phi^*(K)\sigma^{1-\alpha}$ for some $\alpha\in(1,2)$.
\item[(ix$'$)] $\sigma^0\Phi^*(K^*\Phi(\rho)^\alpha K\Phi(\sigma)^{-\alpha})\sigma^0=
\sigma^0\Phi^*(K)^*\rho^\alpha\Phi^*(K)\sigma^{-\alpha}$ for some $\alpha\in(1,2)$.
\end{itemize}

For later discussions, it is convenient to summarize easy implications between the above equality
conditions in the next lemma.

\begin{lemma}\label{L-3.1}
For the conditions stated above, the following implications hold:
\begin{align*}
&(i)\implies(ii)\implies(vii)\implies(iii)\implies(iv), \\
&(viii)\implies\left\{\begin{array}{cc}(v)\implies(vi)\\(ix)\end{array}\right\}\implies(vii), \\
&(vi)\implies(vii'),\qquad(viii)\implies(ix'),\qquad(ix')\implies(vii').
\end{align*}
\end{lemma}

\begin{proof}
When $h(x)=x^\alpha$ with $\alpha\in(0,1)$, note that $\supp\mu_h=(0,\infty)$ and
$S_h^X(\rho\|\sigma)=\Tr X^*\rho^\alpha X\sigma^{1-\alpha}$ for any $X\in\cBH$. Hence
(ii)$\implies$(vii)$\implies$(iii) holds. (viii)$\implies$(v) follows by multiplying $\sigma$ to both
sides of the equality in (viii) from the left and taking the trace. (ix)$\implies$(vii) and
(ix$'$)$\implies$(vii$'$) are similarly shown by replacing $Y$ with $K$. All other implications are
obvious.
\end{proof}

The following are our main results of this section, whose proofs will be given in the next
subsection.

\begin{thm}\label{T-3.2}
Let $\rho$, $\sigma$, $\Phi$ and $K$ be as stated above.
\begin{itemize}
\item[(1)] Conditions (ii), (iii), (vi) and (vii) are equivalent.
\item[(2)] Assume that $\rho^0=\sigma^0$. Then conditions (ii)--(ix) are equivalent.
\item[(3)] Assume that $K\in\cM_{\Phi^*}$ (see \eqref{F-2.12} and \eqref{F-2.13}). Then conditions
(i)--(x) are all equivalent.
\end{itemize}
\end{thm}

\begin{thm}\label{T-3.3}
Let $\rho$, $\sigma$, $\Phi$ and $K$ be as above.
\begin{itemize}
\item[(1)] Conditions (i$'$) and (i) are equivalent.
\item[(2)] Assume that $K\in\cM_{\Phi^*}$. Then condition (iv$'$) holds if and only if (iv) holds
together with $\rho^0\Phi^*(K)(I_\cH-\sigma^0)=0$.
\item[(3)] Assume that $K\in\cM_{\Phi^*}$ and $\rho^0\Phi^*(K)(I_\cH-\sigma^0)=0$
(in particular, the latter is the case if $\sigma\in\cBH^{++}$). Then conditions (i)--(x) (i$'$),
(iv$'$), (vii$'$) and (ix$'$) are all equivalent.
\end{itemize}
\end{thm}

\begin{cor}\label{C-3.4}
Let $\Phi:=\T_\cK:\cB(\cH\otimes\cK)\to\cBH$ be the partial trace over $\cK$. Then conditions
(i)--(x) are all equivalent for every $\rho,\sigma\in\cB(\cH\otimes\cK)^+$ and any $K\in\cBH$.
Moreover, if $\rho^0(K\otimes I_\cK)(I_{\cH\otimes\cK}-\sigma^0)=0$ (in particular, if
$\sigma\in\cB(\cH\otimes\cK)^{++}$), then conditions (i)--(x) (i$'$), (iv$'$), (vii$'$) and (ix$'$) are
all equivalent.
\end{cor}

\begin{proof}
Since $\Phi^*(K)=K\otimes I_\cK$ for all $K\in\cBH$, $\Phi^*$ is a *-homomorphism so that
any $K\in\cBH$ is obviously in $\cM_{\Phi^*}$. Hence the result follows from
Theorems \ref{T-3.2}\,(3) and \ref{T-3.3}\,(3).
\end{proof}

\begin{remark}\label{R-3.5}\rm
In the situation of Theorem \ref{T-3.2}\,(2), condition (x) does not imply (ix) unlike
Theorem \ref{T-3.2}\,(3). For instance, let $\Phi=E_\cA:\cBH\to \cBH$ be the trace-preserving
conditional expectation onto a $C^*$-subalgebra $\cA$ of $\cBH$. Note that $E_\cA^*=E_\cA$ and
\[
\cA=\cM_{E_\cA}=\{X\in\cBH:E_\cA(X^*X)=E_\cA(X)^*E_\cA(X)\},
\]
as easily seen from $E_\cA((X-E_\cA(X))^*(X-E_\cA(X)))=E_\cA(X^*X)-E_\cA(X)^*E_\cA(X)$.
Let $\rho=\sigma=I$. Then (x) reduces to $E_\cA(K)=E_\cA(K)$ that is null, but (ix) means
$E_\cA(K^*K)=E_\cA(K)^*E_\cA(K)$ and hence $K\in\cA$.
\end{remark}

\begin{remark}\label{R-3.6}\rm
Let $\Phi=E_\cA:\cBH\to \cBH$ be as given in Remark \ref{R-3.5}. For every
$\rho,\sigma\in\cBH^+$ and any $K\in\cA$, Theorem \ref{T-3.2}\,(3) says that the following
conditions are equivalent, which are (i), (vii) and (x) in the present setting:
\begin{itemize}
\item $S_h^K(E_\cA(\rho)\|E_\cA(\sigma))=S_h^K(\rho\|\sigma)$ for all $h\in\OM[0,\infty)$,
\item $\Tr K^*E_\cA(\rho)^\alpha KE_\cA(\sigma)^{1-\alpha}=
\Tr K^*\rho^\alpha K\sigma^{1-\alpha}$ for some (equivalently, any) $\alpha\in(0,1)$,
\item $E_\cA(\rho)^{it}KE_\cA(\sigma)^{-it}\sigma^0=\rho^{it}K\sigma^{-it}$ for all $t\in\bR$.
\end{itemize}
In the particular case that $K=I$, the equivalence of the above conditions was given in
\cite[Theorem 9]{JR}.
\end{remark}

\begin{remark}\label{R-3.7}\rm
Sharma \cite{Sh} and Vershynina \cite{Ve} discussed equality conditions for the monotonicity
of quasi-entropies under partial traces (as in Corollary \ref{C-3.4}), to be more precise, conditions
for the equality
\begin{align}\label{F-3.6}
S_f^K(\T_\cK\rho\|\T_\cK\sigma)=S_f^{K\otimes I_\cK}(\rho\|\sigma),
\end{align}
where $\rho,\sigma\in\cB(\cH\otimes\cK)$, $K\in\cBH$ and $f\in\OC(0,\infty)$ under a
somewhat strong assumption on $\nu_f$. Furthermore, the paper \cite{Ve} treated a slightly more
general equality
\begin{align}\label{F-3.7}
S_f^K(\T_\cK\rho\|\T_\cK\sigma)=S_f^{K\otimes V}(\rho\|\sigma)
\end{align}
with a unitary $V\in\cBK$. In view of \eqref{F-2.18}, equality \eqref{F-3.7} is rewritten as
$S_f^K(\T_\cK\rho'\|\T_\cK\sigma)=S_f^{K\otimes I_\cK}(\rho'\|\sigma)$ where
$\rho':=(I_\cH\otimes V)^*\rho(I_\cH\otimes V)$, and the condition
$\rho^0(K\otimes V)(I_{\cH\otimes\cK}-\sigma^0)=0$ in Corollary \ref{C-3.4} is equivalent to
$\rho'^0(K\otimes I_\cK)(I_{\cH\otimes\cK}-\sigma^0)=0$. In this way, we see that the
characterization problem for \eqref{F-3.7} is equivalent to that for \eqref{F-3.6}. Conditions (vi) and
(x) (in the partial trace case) showed up in \cite{Sh} and \cite{Ve}, respectively. The monotonicity
and its equality case of quasi-entropies under TPCP maps (i.e., quantum channels) are essentially
reduced to the partial trace case in view of Stinespring's representation of those maps, while in
the present paper we treat more general trace-preserving positive maps than TPCP maps (as in
Theorems \ref{T-3.2} and \ref{T-3.3}).
\end{remark}

\begin{remark}\label{R-3.8}\rm
In particular, when $K=I_\cK$, since the two assumptions in Theorem \ref{T-3.3}\,(3) are obviously
satisfied for any $\rho,\sigma\in\cBH^+$ with $\rho^0\le\sigma^0$, Theorems \ref{T-3.2}\,(3)
and \ref{T-3.3}\,(3) together reduce to \cite[Theorem 5.1]{HMPB} (also \cite[Theorem 3.18]{HM}).
Here, we note that the assumption $\rho^0\Phi^*(K)(I_\cH-\sigma^0)=0$ is essential in
Theorem \ref{T-3.3}\,(3) as the counterpart of $\rho^0\le\sigma^0$ assumed in
\cite[Theorem 5.1]{HMPB}, without which condition (iv$'$) never shows up (in view of
Lemma \ref{L-2.7}).
\end{remark}

\subsection{Proofs of theorems}\label{Sec-3.2}

Throughout the rest of this section, for notational simplicity we write
\[
\tilde\rho:=\Phi(\rho),\qquad\tilde\sigma:=\Phi(\sigma),\qquad
\Delta:=\Delta_{\rho,\sigma},\qquad\tilde\Delta:=\Delta_{\tilde\rho,\tilde\sigma}.
\]

\medskip\noindent
{\bf Proof of Theorem \ref{T-3.2}\,(1).}\enspace
By Lemma \ref{L-3.1} it suffices to prove that (ii)$\implies$(vi) and (iii)$\implies$(ii) hold.

(ii)$\implies$(vi).\enspace
Apply (ii) to $h(x)=x^\alpha$ for any $\alpha\in(0,1)$. We then have
\[
\Tr K^*\tilde\rho^\alpha K\tilde\sigma^{1-\alpha}
=\Tr\Phi^*(K)^*\rho^\alpha\Phi^*(K)\sigma^{1-\alpha},\qquad\alpha\in(0,1),
\]
which implies (vi) by analytic continuation.

(iii)$\implies$(ii).\enspace
Assume that $S_h^K(\tilde\rho\|\tilde\sigma)=S_h^{\Phi^*(K)}(\rho\|\sigma)$ for some
$h\in\OM_+[0,\infty)$ and with the condition stated in (iii). By \eqref{F-3.1} we write
$h(x)=h(0)+h'(\infty)x+h_0(x)$, where $h_0\in\OM[0,\infty)$ with $h_0(0)=h_0'(\infty)=0$.
Then the assumption says that
\begin{align*}
&h(0)\Tr\tilde\sigma K^*K+h'(\infty)\Tr\tilde\rho KK^*
+S_{h_0}^K(\tilde\rho\|\tilde\sigma) \\
&\qquad=h(0)\Tr\sigma\Phi^*(K)^*\Phi^*(K)+h'(\infty)\Tr\rho\Phi^*(K)\Phi^*(K)^*
+S_{h_0}^{\Phi^*(K)}(\rho\|\sigma).
\end{align*}
Since
\begin{align*}
\Tr\tilde\sigma K^*K&=\Tr\sigma\Phi^*(K^*K)\ge\Tr\sigma\Phi^*(K)^*\Phi^*(K), \\
\Tr\tilde\rho KK^*&=\Tr\rho\Phi^*(KK^*)\ge\Tr\rho\Phi^*(K)\Phi^*(K)^*
\end{align*}
as well as $S_{h_0}^K(\tilde\rho\|\tilde\sigma)\ge S_{h_0}^{\Phi^*(K)}(\rho\|\sigma)$ by
Proposition \ref{T-2.14}, we have
$S_{h_0}^K(\tilde\rho\|\tilde\sigma)=S_{h_0}^{\Phi^*(K)}(\rho\|\sigma)$. Note that the integral
expression in \eqref{F-3.1} for $h_0$ is written as
\begin{align}\label{F-3.8}
h_0(x)=\int_{(0,\infty)}\ffi_t(x)(1+t)\,d\mu_{h_0}(t),\qquad x\in[0,\infty),
\end{align}
where $\ffi_t$ is given in \eqref{F-3.3}. Hence it follows that
\begin{equation}\label{F-3.9}
\begin{aligned}
S_{h_0}^K(\tilde\rho\|\tilde\sigma)
&=\int_{(0,\infty)}S_{\ffi_t}^K(\tilde\rho\|\tilde\sigma)(1+t)\,d\mu_{h_0}(t), \\
S_{h_0}^{\Phi^*(K)}(\rho\|\sigma)
&=\int_{(0,\infty)}S_{\ffi_t}^{\Phi^*(K)}(\rho\|\sigma)(1+t)\,d\mu_{h_0}(t),
\end{aligned}
\end{equation}
and by Proposition \ref{T-2.14},
\[
S_{\ffi_t}^K(\tilde\rho\|\tilde\sigma)\ge S_{\ffi_t}^{\Phi^*(K)}(\rho\|\sigma),
\qquad t\in(0,\infty).
\]
Therefore, we find a $T\subset\supp\mu_{h_0}$ ($=\supp\mu_h$) such that $T$ has a limit point
in $(0,\infty)$ and $S_{\ffi_t}^K(\tilde\rho\|\tilde\sigma)=S_{\ffi_t}^{\Phi^*(K)}(\rho\|\sigma)$ for all
$t\in T$. By Lemma \ref{L-2.7} this means that
\begin{equation}\label{F-3.10}
\begin{aligned}
&\<K\tilde\sigma^{1/2},\tilde\Delta(\tilde\Delta+tI_{\cBK})^{-1}(K\tilde\sigma^{1/2}\>_\HS \\
&\qquad=\<\Phi^*(K)\sigma^{1/2},\Delta(\Delta+tI_{\cBH})^{-1}(\Phi^*(K)\sigma^{1/2})\>_\HS,
\qquad t\in T.
\end{aligned}
\end{equation}
Both side functions in $t\in(0,\infty)$ of the above equality extend to analytic functions in
$\{z\in\bC:\Re z>0\}$. Hence the coincidence theorem for analytic functions yields that
\eqref{F-3.10} holds for all $t\in(0,\infty)$, i.e.,
$S_{\ffi_t}^K(\tilde\rho\|\tilde\sigma)=S_{\ffi_t}^{\Phi^*(K)}(\rho\|\sigma)$ for all $t\in(0,\infty)$. Thus,
(ii) follows by making use of the integral expressions for any $h$ stated in (ii) like \eqref{F-3.9}.\qed

\medskip
Next, we prove the statement (3) before (2) of Theorem \ref{T-3.2}.

\medskip\noindent
{\bf Proof of Theorem \ref{T-3.2}\,(3).}\enspace
By Lemma \ref{L-3.1} it suffices to prove the implications (iv)$\implies$(x)$\implies$(viii) and
(v)$\implies$(i).

(v)$\implies$(i).\enspace
By (v) with $Y=K^*$ and $z=n\in\bN\cup\{0\}$ we have
\[
\Tr K^*\tilde\rho^nK\tilde\sigma^{1-n}=\Tr\Phi^*(K)^*\rho^n\Phi^*(K)\sigma^{1-n},
\]
which is rewritten as
\begin{align}\label{F-3.11}
\<K\tilde\sigma^{1/2},\tilde\Delta^n(K\tilde\sigma^{1/2})\>_\HS
=\<\Phi^*(K)\sigma^{1/2},\Delta^n(\Phi^*(K)\sigma^{1/2})\>_\HS,
\qquad n\in\bN\cup\{0\}.
\end{align}
Let $f$ be any continuous function $f$ on $[0,\infty)$. Approximate $f$ uniformly  on an interval
$[0,c]$ including $\Sp(\Delta)\cup\Sp(\tilde\Delta)$ by polynomials and use \eqref{F-3.11} to obtain
\begin{align}\label{F-3.12}
\<K\tilde\sigma^{1/2},f(\tilde\Delta)(K\tilde\sigma^{1/2})\>_\HS
=\<\Phi^*(K)\sigma^{1/2},f(\Delta)(\Phi^*(K)\sigma^{1/2})\>_\HS.
\end{align}
Now, for any $h\in\OM[0,\infty)$, write $h(x)=h(0)+h'(\infty)x+h_0(x)$ with
$h_0(0)=h_0'(\infty)=0$ as in the proof of Theorem \ref{T-3.2}\,(1). From Lemma \ref{L-2.7} and
\eqref{F-3.12} for $f=h_0$ it follows that
\[
S_{h_0}^K(\tilde\rho\|\tilde\sigma)=S_{h_0}^{\Phi^*(K)}(\rho\|\sigma).
\]
Moreover, by the assumption $K\in\cM_{\Phi^*}$,
\begin{equation}\label{F-3.13}
\begin{aligned}
S_1^K(\tilde\rho\|\tilde\sigma)=\Tr\tilde\sigma K^*K
&=\Tr\sigma\Phi^*(K)^*\Phi^*(K)=S_1^{\Phi^*(K)}(\rho\|\sigma), \\
S_x^K(\tilde\rho\|\tilde\sigma)=\Tr\tilde\rho KK^*
&=\Tr\rho\Phi^*(K)\Phi^*(K)^*=S_x^{\Phi^*(K)}(\rho\|\sigma).
\end{aligned}
\end{equation}
Combining the above equalities yields (i).

(iv)$\implies$(x).\enspace
Assume that $S_h^K(\tilde\rho\|\tilde\sigma)=S_h^{\Phi^*(K)}(\rho\|\sigma)$ for some
$h\in\OM_+[0,\infty)$ and \eqref{F-3.5}. Then, as in the proof of Theorem \ref{T-3.2}\,(1),
we can choose a $T\subset\supp\mu_h$ such that $|T|\ge|\Sp(\Delta)\cup\Sp(\tilde\Delta)|$ and
\begin{align}\label{F-3.14}
S_{\ffi_t}^K(\tilde\rho\|\tilde\sigma)=S_{\ffi_t}^{\Phi^*(K)}(\rho\|\sigma),
\qquad t\in T.
\end{align}
Since
\[
\|\Phi^*(Y)\sigma^{1/2}\|_\HS^2=\Tr\sigma\Phi^*(Y^*)\Phi^*(Y)
\le\Tr\sigma\Phi^*(Y^*Y)=\|Y\tilde\sigma^{1/2}\|_\HS^2,\qquad
Y\in\cBK,
\]
one can define a linear contraction $V:\cBK\to \cBH$ by
$V(Y\tilde\sigma^{1/2}):=\Phi^*(Y)\sigma^{1/2}$ for $Y\in\cBK$ and $V(Y):=0$ for
$Y\in(\cBK\tilde\sigma)^\perp$, that is,
\begin{align}\label{F-3.15}
V(Y):=\Phi^*(Y\tilde\sigma^{-1/2})\sigma^{1/2},\qquad Y\in\cBK.
\end{align}
Then one has for every $Y\in\cBK$,
\begin{align*}
\<Y,V^*\Delta V(Y)\>_\HS
&=\<\Phi^*(Y\tilde\sigma^{-1/2})\sigma^{1/2},
\Delta\Phi^*(Y\tilde\sigma^{-1/2})\sigma^{1/2}\>_\HS \\
&=\<\Phi^*(Y\tilde\sigma^{-1/2})\sigma^{1/2},
\rho\Phi^*(Y\tilde\sigma^{-1/2})\sigma^{-1/2}\>_\HS \\
&=\Tr\rho\Phi^*(Y\tilde\sigma^{-1/2})\sigma^0\Phi^*(Y\tilde\sigma^{-1/2})^* \\
&\le\Tr\rho\Phi^*(Y\tilde\sigma^{-1}Y^*)=\Tr\tilde\rho Y\tilde\sigma^{-1}Y^*
=\<Y,\tilde\Delta Y\>_\HS,
\end{align*}
which shows that $V^*\Delta V\le\tilde\Delta$ on $(\cBK,\<\cdot,\cdot\>_\HS)$. Therefore, it
follows from \cite[Theorem 2.1]{HaPe} that
\begin{align}\label{F-3.16}
\ffi_t(\tilde\Delta)\ge\ffi_t(V^*\Delta V)\ge V^*\ffi_t(\Delta)V.
\end{align}
On the other hand, since
$V(K\tilde\sigma^{1/2})=\Phi^*(K)\sigma^{1/2}$, by \eqref{F-3.14} and Lemma \ref{L-2.7} (with
$\ffi_t'(\infty)=0$) one has
\[
\<K\tilde\sigma^{1/2},\bigl[\ffi_t(\tilde\Delta)-V^*\ffi_t(\Delta)V\bigr](K\tilde\sigma^{1/2})
\>_\HS=0,\qquad t\in T.
\]
This with \eqref{F-3.16} yields
\begin{align}\label{F-3.17}
\ffi_t(\tilde\Delta)(K\tilde\sigma^{1/2})=V^*\ffi_t(\Delta)(\Phi^*(K)\sigma^{1/2}),
\qquad t\in T,
\end{align}
which means by \eqref{F-3.3} that
\begin{align}\label{F-3.18}
[I-t(\tilde\Delta+tI)^{-1}](K\tilde\sigma^{1/2})=V^*[I-t(\Delta+tI)^{-1}](\Phi^*(K)\sigma^{1/2}),
\qquad t\in T.
\end{align}
Moreover, one has for every $Y\in\cBK$,
\begin{align*}
\<V(Y),\Phi^*(K)\sigma^{1/2}\>_\HS
&=\<\Phi^*(Y\tilde\sigma^{-1/2})\sigma^{1/2},\Phi^*(K)\sigma^{1/2}\>_\HS \\
&=\Tr\Phi^*(Y\tilde\sigma^{-1/2})^*\Phi^*(K)\sigma
=\Tr\Phi^*(\tilde\sigma^{-1/2}Y^*K)\sigma \\
&=\Tr\tilde\sigma^{-1/2}Y^*K\tilde\sigma=\<Y,K\tilde\sigma^{1/2}\>_\HS,
\end{align*}
where the third equality above is due to $K\in\cM_{\Phi^*}$ (see \eqref{F-2.13}). Therefore,
$K\tilde\sigma^{1/2}=V^*(\Phi^*(K)\sigma^{1/2})$. By this and \eqref{F-3.18} we have
\[
(\tilde\Delta+tI)^{-1}K\tilde\sigma^{1/2}=V^*(\Delta+tI)^{-1}\Phi^*(K)\sigma^{1/2},
\qquad t\in T.
\]
Since $|T|\ge|\Sp(\Delta)\cup\Sp(\tilde\Delta)|$, by a similar argument to
\cite[around (5.5)--(5.7)]{HMPB}, we have for any function $f$ on $[0,\infty)$,
\[
V\bigl(f(\tilde\Delta)(K\tilde\sigma^{1/2})\bigr)=f(\Delta)(\Phi^*(K)\sigma^{1/2}),
\]
that is, by \eqref{F-3.15},
\[
\Phi^*\bigl(f(\tilde\Delta)(K\tilde\sigma^{1/2})\tilde\sigma^{-1/2}\bigr)\sigma^{1/2}
=f(\Delta)(\Phi^*(K)\sigma^{1/2}).
\]
Taking $f(x)=x^z$ on $(0,\infty)$ and $f(0)=0$ for any $z\in\bC$ yields
\[
\Phi^*\bigl(\tilde\rho^zK\tilde\sigma^{1/2}\tilde\sigma^{-z}\tilde\sigma^{-1/2}\bigr)\sigma^{1/2}
=\rho^z\Phi^*(K)\sigma^{1/2}\sigma^{-z}
\]
so that
\begin{align}\label{F-3.19}
\Phi^*(\tilde\rho^zK\tilde\sigma^{-z})\sigma^0=\rho^z\Phi^*(K)\sigma^{-z},
\qquad z\in\bC,
\end{align}
showing (x).

(x)$\implies$(viii).\enspace
Put $u_t:=\tilde\rho^{it}K\tilde\sigma^{-it}$ and  $w_t:=\rho^{it}\Phi^*(K)\sigma^{-it}$ for
$t\in\bR$. By (x) we have $\Phi^*(u_t)\sigma^0=w_t$ and $\sigma^0\Phi^*(u_t^*)=w_t^*$
so that
\begin{align*}
\Tr\sigma\bigl(\Phi^*(u_t^*u_t)-\Phi^*(u_t^*)\Phi^*(u_t)\bigr)
&=\Tr\sigma\bigl(\Phi^*(u_t^*u_t)-\sigma^0\Phi^*(u_t^*)\Phi^*(u_t)\sigma^0\bigr) \\
&=\Tr\tilde\sigma u_t^*u_t-\Tr\sigma w_t^*w_t \\
&=\Tr\tilde\sigma\tilde\sigma^{it}K^*\tilde\rho^0K\tilde\sigma^{-it}
-\Tr\sigma\sigma^{it}\Phi^*(K)^*\rho^0\Phi^*(K)\sigma^{-it} \\
&=\Tr\sigma\Phi^*(K^*\tilde\rho^0K\tilde\sigma^0)\sigma^0
-\Tr\sigma\Phi^*(K)^*\rho^0\Phi^*(K)\sigma^0 \\
&=\Tr\sigma\Phi^*(K)^*\Phi^*(\tilde\rho^0K\tilde\sigma^0)\sigma^0
-\Tr\sigma\Phi^*(K)^*\rho^0\Phi^*(K)\sigma^0,
\end{align*}
where the last equality is due to $K\in\cM_{\Phi^*}$. Letting $t=0$ in (x) gives
\[
\Phi^*(\tilde\rho^0K\tilde\sigma^0)\sigma^0=\rho^0\Phi^*(K)\sigma^0.
\]
We hence have
\[
\Tr\sigma\bigl(\Phi^*(u_t^*u_t)-\Phi^*(u_t^*)\Phi^*(u_t)\bigr)=0
\]
so that
\[
\sigma^0(\Phi^*(u_t^*u_t)-\Phi^*(u_t^*)\Phi^*(u_t))\sigma^0=0.
\]
For any $Y\in\cBK$ and any $t\in\bR$, by Lemma \ref{L-3.9} below we have
\begin{align*}
\sigma^0\Phi^*(Y\tilde\rho^{it}K\tilde\sigma^{-it})\sigma^0
&=\sigma^0\Phi^*(Yu_t)\sigma^0=\sigma^0\Phi^*(Y)\Phi^*(u_t)\sigma^0 \\
&=\sigma^0\Phi^*(Y)w_t=\sigma^0\Phi^*(Y)\rho^{it}\Phi^*(K)\sigma^{-it},
\end{align*}
which implies (viii) by analytic continuation.\qed

\medskip
The next lemma is a slight modification of \cite[Theorem 3.1]{Choi} and
\cite[Lemma 3.9]{HMPB}, while we give a proof for the convenience of the reader.

\begin{lemma}\label{L-3.9}
Let $\Psi:\cBK\to \cBH$ be a unital Schwarz map, $K\in\cBK$ and $A\in\cBH$. If
\[
A^*\Psi(K^*K)A=A^*\Psi(K)^*\Psi(K)A,
\]
then for every $Y\in\cBK$,
\begin{align}\label{F-3.20}
A^*\Psi(YK)A&=A^*\Psi(Y)\Psi(K)A.
\end{align}
\end{lemma}

\begin{proof}
Define a $\cBH$-valued sesqui-linear form on $\cBK\times \cBK$ by
\[
D(Y_1,Y_2):=A^*(\Psi(Y_1^*Y_2)-\Psi(Y_1)^*\Psi(Y_2))A,
\qquad Y_1,Y_2\in\cBK.
\]
By assumption, $D(K,K)=0$. For any $B\in\cBH^+$ and any $Y\in\cBK$, by applying
the Schwarz inequality to the positive sesqui-linear form $\Tr BD(Y_1,Y_2)$ we have
\[
|\Tr BD(Y,K)|\le[\Tr BD(Y,Y)]^{1/2}[\Tr BD(K,K)]^{1/2}=0,
\]
so that $\Tr BD(Y,K)=0$ for all $B\in\cBH^+$. Hence $D(Y,K)=0$. This implies equality
\eqref{F-3.20}.
\end{proof}

To prove Theorem \ref{T-3.2}\,(2), it is convenient to consider the restriction map
$\hat\Phi:\cB(\sigma^0\cH)\to\cB(\tilde\sigma^0\cK)$, which we define below. Assume
that $\rho^0\le\sigma^0$; hence $\tilde\rho^0\le\tilde\sigma^0$ holds as well. Since
$\Phi(\sigma^0)^0=\tilde\sigma^0$, we can define
\[
\hat\Phi:\cB(\sigma^0\cH)=\sigma^0\cBH\sigma^0\to
\cB(\tilde\sigma^0\cK)=\tilde\sigma^0\cBK\tilde\sigma^0
\]
by
\begin{align}\label{F-3.21}
\hat\Phi(X):=\Phi(X)|_{\tilde\sigma^0\cK}
=\tilde\sigma^0\Phi(X)\tilde\sigma^0|_{\tilde\sigma^0\cK},\qquad X\in\cB(\sigma^0\cH).
\end{align}
Since $\rho^0\le\sigma^0$, note that $\hat\Phi(\rho)=\Phi(\rho)=\tilde\rho$ and
$\hat\Phi(\sigma)=\Phi(\sigma)=\tilde\sigma$, where we use the same $\rho,\sigma$ when they
are considered as operators in $\cB(\sigma^0\cH)$ and also the same $\tilde\rho,\tilde\sigma$
when considered in $\cB(\tilde\sigma^0\cK)$. Note that the modular operator of $\rho,\sigma$ on
$\cB(\sigma^0\cH)$ is the restriction of $\Delta:=\Delta_{\rho,\sigma}$ to $\cB(\sigma^0\cH)$ and
that of $\tilde\rho,\tilde\sigma$ on $\cB(\tilde\sigma^0\cK)$ is the restriction of
$\tilde\Delta:=\Delta_{\tilde\rho,\tilde\sigma}$ to $\cB(\tilde\sigma^0\cK)$. So we use the same
symbols $\Delta$ and $\tilde\Delta$ when they are considered as modular operators on
$\cB(\sigma^0\cH)$ and $\cB(\tilde\sigma^0\cK)$, respectively, as there arises no confusion.
Moreover, we write
\begin{align}\label{F-3.22}
\hat K:=\tilde\sigma^0K\tilde\sigma^0\in\cB(\tilde\sigma^0\cK).
\end{align}

\begin{lemma}\label{L-3.10}
With the above assumption and notations, $\hat\Phi$ is a trace-preserving map (hence
$\hat\Phi^*(\tilde\sigma^0)=\sigma^0$), and $\hat\Phi^*$ is a Schwarz map. Furthermore,
\begin{align}\label{F-3.23}
\hat\Phi^*(\tilde\sigma^0Y\tilde\sigma^0)=\sigma^0\Phi^*(Y)\sigma^0,
\qquad Y\in\cBK.
\end{align}
In particular,
\begin{align}\label{F-3.24}
\hat\Phi^*(\hat K)=\sigma^0\Phi^*(K)\sigma^0.
\end{align}
\end{lemma}

\begin{proof}
It is clear that $\hat\Phi$ is a trace-preserving map and hence
$\hat\Phi^*(\tilde\sigma^0)=\sigma^0$ (since $\sigma^0,\tilde\sigma^0$ are the identities of
$\cB(\sigma^0\cH),\cB(\tilde\sigma^0\cK)$, respectively). For every $X\in\cB(\sigma^0\cH)$ and
$Y\in\cBK$,
\begin{align*}
\<X,\hat\Phi^*(\tilde\sigma^0Y\tilde\sigma^0)\>_\HS
&=\<\hat\Phi(X),\tilde\sigma^0Y\tilde\sigma^0\>_\HS=\<\Phi(X),Y\>_\HS \\
&=\<X,\Phi^*(Y)\>_\HS=\<X,\sigma^0\Phi^*(Y)\sigma^0\>_\HS.
\end{align*}
Hence \eqref{F-3.23} follows. Furthermore, for every $Y\in\cB(\tilde\sigma^0\cK)$,
\[
\hat\Phi^*(Y^*Y)=\sigma^0\Phi^*(Y^*Y)\sigma^0
\ge\sigma^0\Phi^*(Y^*)\Phi^*(Y)\sigma^0
\ge\hat\Phi^*(Y^*)\hat\Phi^*(Y).
\]
Hence $\hat\Phi^*$ is a Schwarz map.
\end{proof}

\begin{lemma}\label{L-3.11}
Let $\rho,\sigma,\Phi,K$ be as above, and assume that $\rho^0\le\sigma^0$. Then conditions (ii),
(iii) and (v)--(ix) given in Sec.~\ref{Sec-3.1} for $\rho,\sigma,\Phi,K$ are, respectively, equivalent
to (ii), (iii) and (v)--(ix) for $\rho,\sigma,\hat\Phi,\hat K$. Condition (iv) for $\rho,\sigma,\Phi,K$
implies (iv) for $\rho,\sigma,\hat\Phi,\hat K$.
\end{lemma}

\begin{proof}
Let $h\in\OM[0,\infty)$ with $h(0)=h'(\infty)=0$. Since
$h(\tilde\Delta)=\tilde\Delta^0h(\tilde\Delta)\tilde\Delta^0$, by Lemma \ref{L-2.7} we have
\begin{equation}\label{F-3.25}
\begin{aligned}
S_h^K(\tilde\rho\|\tilde\sigma)
&=\<K\tilde\sigma^{1/2},h(\tilde\Delta)K\tilde\sigma^{1/2}\>_\HS \\
&=\<\tilde\Delta^0K\tilde\sigma^{1/2},h(\tilde\Delta)\tilde\Delta^0K\tilde\sigma^{1/2}\>_\HS \\
&=\<\hat K\tilde\sigma^{1/2},h(\tilde\Delta)\hat K\tilde\sigma^{1/2}\>_\HS
=S_h^{\hat K}(\tilde\rho\|\tilde\sigma),
\end{aligned}
\end{equation}
where the third equality above follows from $\tilde\Delta^0=L_{\tilde\rho^0}R_{\tilde\sigma^0}$
and $\tilde\rho^0\le\tilde\sigma^0$. Similarly we have
\begin{equation}\label{F-3.26}
\begin{aligned}
S_h^{\Phi^*(K)}(\rho\|\sigma)
&=\<\Phi^*(K)\sigma^{1/2},h(\Delta)\Phi^*(K)\sigma^{1/2}\>_\HS \\
&=\<\Delta^0\Phi^*(K)\sigma^{1/2},h(\Delta)\Delta^0\Phi^*(K)\sigma^{1/2}\>_\HS \\
&=\<\hat\Phi^*(\hat K)\sigma^{1/2},h(\Delta)\hat\Phi^*(\hat K)\sigma^{1/2}\>_\HS
=S_h^{\hat\Phi^*(\hat K)}(\rho\|\sigma),
\end{aligned}
\end{equation}
where the third equality follows from \eqref{F-3.24}. Therefore, the assertion for (ii) holds.

Next, assume condition (iii) or (iv) for $\rho,\sigma,\Phi,K$. As in the proof of
Theorem \ref{T-3.2}\,(1), writing $h(x)=h(0)+h'(\infty)x+h_0(x)$ with $h_0(0)=h_0'(\infty)=0$, we
have $S_{h_0}^K(\tilde\rho\|\tilde\sigma)=S_{h_0}^{\Phi^*(K)}(\rho\|\sigma)$. Similarly,
$S_{h_0}^{\hat K}(\tilde\rho\|\tilde\sigma)=S_{h_0}^{\hat\Phi^*(\hat K)}(\rho\|\sigma)$ if (iii) or (iv)
holds for $\rho,\sigma,\hat\Phi,\hat K$. Hence the assertion for (iii) holds by \eqref{F-3.25} and
\eqref{F-3.26}. Moreover, condition \eqref{F-3.5} on $\supp\mu_h$ ($=\supp\mu_{h_0}$) still holds
when $\Delta$, $\tilde\Delta$ are restricted on $\cB(\sigma^0\cH)$, $\cB(\tilde\sigma^0\cH)$.
Hence the assertion for (iv) holds too.

The assertions for the remaining conditions (v)--(ix) are immediately shown by the
following equalities: for any $Y\in\cBK$ and $z\in\bC$,
\begin{align*}
\Tr Y\tilde\rho^zK\tilde\sigma^{1-z}
&=\Tr Y\tilde\sigma^0\tilde\rho^z\tilde\sigma^0K
\tilde\sigma^0\tilde\sigma^{1-z}\tilde\sigma^0
=\Tr\hat Y\tilde\rho^z\hat K\tilde\sigma^{1-z}, \\
\Tr\Phi^*(Y)\rho^z\Phi^*(K)\sigma^{1-z}
&=\Tr\Phi^*(Y)\sigma^0\rho^z\sigma^0\Phi^*(K)\sigma^0\sigma^{1-z}\sigma^0
=\Tr\hat\Phi^*(\hat Y)\rho^z\hat\Phi^*(\hat K)\sigma^{1-z}, \\
\sigma^0\Phi^*(Y^*\tilde\rho^z K\tilde\sigma^{-z})\sigma^0
&=\hat\Phi^*(\tilde\sigma^0Y\tilde\sigma^0\tilde\rho^z\tilde\sigma^0K
\tilde\sigma^0\tilde\sigma^{-z}\tilde\sigma^0)
=\hat\Phi^*(\hat Y\tilde\rho^z\hat K\tilde\sigma^{-z}), \\
\sigma^0\Phi^*(Y)\rho^z\Phi^*(K)\sigma^{-z}
&=\sigma^0\Phi^*(Y)\sigma^0\rho^z\sigma^0\Phi^*(K)\sigma^0\sigma^{-z}
=\hat\Phi^*(\hat Y)\rho^z\hat\Phi^*(\hat K)\sigma^{-z},
\end{align*}
where $\hat Y:=\tilde\sigma^0Y\tilde\sigma^0$ and we have used \eqref{F-3.23} and
\eqref{F-3.24} repeatedly.
\end{proof}

\noindent
{\bf Proof of Theorem \ref{T-3.2}\,(2).}\enspace
By Lemma \ref{L-3.1} it suffices to prove that (v)$\implies$(ii) and (iv)$\implies$(viii). To do so,
by Lemma \ref{L-3.11} we may and do assume that $\sigma,\tilde\sigma$ are invertible. Hence
$\rho,\tilde\rho$ are invertible too by assumption $\rho^0=\sigma^0$.

(v)$\implies$(ii).\enspace
Let us prove the stronger implication (v)$\implies$(i) in the present situation.  By (v) with
$Y=K^*$ and $z=n\in\bN\cup\{0\}$ we have
\[
\Tr K^*\tilde\rho^nK\tilde\sigma^{1-n}=\Tr\Phi^*(K)^*\rho^n\Phi^*(K)\sigma^{1-n},
\]
which is rewritten as
\begin{align}\label{F-3.27}
\<K\tilde\sigma^{1/2},\tilde\Delta^n(K\tilde\sigma^{1/2})\>_\HS
=\<\Phi^*(K)\sigma^{1/2},\Delta^n(\Phi^*(K)\sigma^{1/2})\>_\HS,
\qquad n\in\bN\cup\{0\}.
\end{align}
For any continuous function $f$ on $[0,\infty)$, as in the proof of (v)$\implies$(i) of
Theorem \ref{T-3.2}\,(3), by \eqref{F-3.27} we have
\[
\<K\tilde\sigma^{1/2},f(\tilde\Delta)(K\tilde\sigma^{1/2})\>_\HS
=\<\Phi^*(K)\sigma^{1/2},f(\Delta)(\Phi^*(K)\sigma^{1/2})\>_\HS,
\]
which is a stronger version of (i) in view of Lemma \ref{L-2.7} (since $\sigma$ and $\tilde\sigma$
are invertible here).

(iv)$\implies$(viii).\enspace
Assume (iv). By the second paragraph of the proof of Lemma \ref{L-3.11}, we may and do assume
that (iv) holds with an $h\in\OM[0,\infty)$ with $h(0)=h'(\infty)=0$, so that $h$ is
given as in \eqref{F-3.8} and \eqref{F-3.3}. Then, as in the proof of Theorem \ref{T-3.2}\,(1)
and also as in the proof of (iv)$\implies$(x) of Theorem \ref{T-3.2}\,(3), we can choose a
$T\subset\supp\mu_h$ such that $|T|\ge|\Sp(\Delta)\cup\Sp(\tilde\Delta)|$ and \eqref{F-3.14} holds.
One can define a linear contraction $V:\cBK\to \cBH$ as in \eqref{F-3.15}, which satisfies
\eqref{F-3.16} and hence \eqref{F-3.17} as before.

Now we can modify the proof of \cite[(5.5)]{HMPB}. We write
$\Sp(\Delta)\cup\Sp(\tilde\Delta)=\{x_i\}_{i=1}^m\subset(0,\infty)$, where
$\Sp(\Delta)=\{x_i\}_{i=1}^k$ and $\Sp(\tilde\Delta)=\{x_i\}_{i=l}^m$ with $1\le k,l\le m$ and
$l\le k+1$. Choose distinct $t_1,\dots,t_m$ from $T$ and consider an $m\times m$ matrix
\[
C:=\bigl[\ffi_{t_j}(x_i)\bigr]_{i,j=1}^m=\biggl[{x_i\over x_i+t_j}\biggr]_{i,j=1}^m.
\]
Note that $C$ is invertible, since it is the product of $\diag(x_1,\dots,x_m)$ and a Cauchy matrix
$\Bigl[{1\over x_i+t_j}\Bigr]_{i,j=1}^m$. Hence, for any continuous function $f:[0,\infty)\to\bC$
define
\[
(y_1,\dots,y_m)^t:=C^{-1}(f(x_1),\dots,f(x_m))^t.
\]
Then, with the spectral decompositions $\Delta=\sum_{i=1}^kx_iE_i$ and
$\tilde\Delta=\sum_{i=l}^mx_i\hat E_i$ we have
\begin{equation}\label{F-3.28}
\begin{aligned}
f(\tilde\Delta)(K\tilde\sigma^{1/2})
&=\sum_{i=l}^mf(x_i)\hat E_i(K\tilde\sigma^{1/2})
=\sum_{i=l}^m\Biggl(\sum_{j=1}^m\ffi_{t_j}(x_i)y_j\Biggr)\hat E_i(K\tilde\sigma^{1/2}) \\
&=\sum_{j=1}^my_j\ffi_{t_j}(\tilde\Delta)(K\tilde\sigma^{1/2})
=\sum_{j=1}^my_jV^*\ffi_{t_j}(\Delta)(\Phi^*(K)\sigma^{1/2}) \\
&=V^*\sum_{i=1}^k\Biggl(\sum_{j=1}^m\ffi_{t_j}(x_i)y_j\Biggr)E_i(\Phi^*(K)\sigma^{1/2}) \\
&=V^*\sum_{i=1}^kf(x_i)E_i(\Phi^*(K)\sigma^{1/2})
=V^*f(\Delta)(\Phi^*(K)\sigma^{1/2}),
\end{aligned}
\end{equation}
where \eqref{F-3.17} has been used for the fourth equality. When $f(x)\equiv1$, since
\eqref{F-3.28} gives
\begin{align*}
\|K\tilde\sigma^{1/2}\|_\HS^2&=\|V^*\Phi^*(K)\sigma^{1/2}\|_\HS^2
\le\|\Phi^*(K)\sigma^{1/2}\|_\HS^2=\Tr\sigma\Phi^*(K)^*\Phi^*(K) \\
&\le\Tr\sigma\Phi^*(K^*K)=\Tr\tilde\sigma K^*K=\|K\tilde\sigma^{1/2}\|_\HS^2,
\end{align*}
one has $\Tr\sigma\Phi^*(K)^*\Phi^*(K)=\Tr\sigma\Phi^*(K^*K)$ and hence
$\Phi^*(K^*K)=\Phi^*(K)^*\Phi^*(K)$ since $\sigma\in\cBH^{++}$. When $f(x)=x^{1/2}$, since
\eqref{F-3.28} gives
\[
\tilde\rho^{1/2}K=\tilde\Delta^{1/2}(K\tilde\sigma^{1/2})
=V^*\Delta^{1/2}(\Phi^*(K)\sigma^{1/2})=V^*\rho^{1/2}\Phi^*(K),
\]
one has
\begin{align*}
\|\tilde\rho^{1/2}K\|_\HS^2&\le\|\rho^{1/2}\Phi^*(K)\|_\HS=\Tr\rho\Phi^*(K)\Phi^*(K)^* \\
&\le\Tr\rho\Phi^*(KK^*)=\Tr\tilde\rho KK^*=\|\tilde\rho^{1/2}K\|_\HS^2
\end{align*}
and hence $\Phi^*(KK^*)=\Phi^*(K)\Phi^*(K)^*$ since $\rho\in\cBH^{++}$. Therefore,
$K\in\cM_{\Phi^*}$ in the present situation, so we can use Theorem \ref{T-3.2}\,(3) already shown
to obtain (iv)$\implies$(viii).\qed

\medskip
We now proceed to prove Theorem \ref{T-3.3}.

\medskip\noindent
{\bf Proof of Theorem \ref{T-3.3}\,(1).}\enspace
(i$'$)$\implies$(i) is obvious since $-h\in\OC[0,\infty)$ if $h\in\OM[0,\infty)$ (see, e.g.,
\cite[Corollary 2.5.6]{Hi0}). The converse implication can be shown similarly to the proof of
Proposition \ref{T-2.16}. Indeed, for any $f\in\OC(0,\infty)$, choose sequences of functions
$f_n\in\OC(0,\infty)$ and $h_n\in\OM[0,\infty)$ with $h_n(0)=h_n'(\infty)=0$ for which \eqref{F-2.15}
and \eqref{F-2.16} hold. Applying (i) to the functions $1$, $x$ and $h_n$, we have in turn
\begin{align*}
S_{f_n}^K(\tilde\rho\|\tilde\sigma)
&=f_n(0)\Tr\tilde\sigma K^*K+f_n'(\infty)\Tr\tilde\rho KK^*
-S_{h_n}^K(\tilde\rho\|\tilde\sigma) \\
&=f_n(0)\Tr\sigma\Phi^*(K)^*\Phi^*(K)+f_n'(\infty)\Tr\rho\Phi^*(K)\Phi^*(K)^*
-S_{h_n}^{\Phi^*(K)}(\rho\|\sigma) \\
&=S_{f_n}^{\Phi^*(K)}(\rho\|\sigma).
\end{align*}
Letting $n\to\infty$ yields $S_f^K(\tilde\rho\|\tilde\sigma)=S_f^{\Phi^*(K)}(\rho\|\sigma)$. Hence
(i)$\implies$(i$'$) follows.\qed

\medskip\noindent
{\bf Proof of Theorem \ref{T-3.3}\,(2).}\enspace
Assume (iv$'$). Since $S_f^K(\tilde\rho\|\tilde\sigma),S_f^{\Phi^*(K)}(\rho\|\sigma)<+\infty$ and
$f'(\infty)=+\infty$, Lemma \ref{L-2.7} yields
\[
\Tr K^*\tilde\rho K(I-\tilde\sigma^0)=\Tr\Phi^*(K)^*\rho\Phi^*(K)(I-\sigma^0)=0.
\]
By Lemma \ref{L-2.8} the last equality to $0$ is equivalent to $\rho^0\Phi^*(K)(I-\sigma^0)=0$.
Making use of the integral expression in \eqref{F-3.2} with $f(0)=0$, we can rewrite
$S_f^K(\tilde\rho\|\tilde\sigma)=S_f^{\Phi^*(K)}(\rho\|\sigma)<+\infty$ as
\begin{equation}\label{F-3.29}
\begin{aligned}
&a\Tr\tilde\rho KK^*+bS_{x^2}^K(\tilde\rho\|\tilde\sigma)
+\int_{(0,\infty)}S_{\psi_t}^K(\tilde\rho\|\tilde\sigma)\,d\nu_f(t) \\
&\quad=a\Tr\rho\Phi^*(K)\Phi^*(K)^*+bS_{x^2}^{\Phi^*(K)}(\rho\|\sigma)
+\int_{(0,\infty)}S_{\psi_t}^{\Phi^*(K)}(\rho\|\sigma)\,d\nu_f(t)<+\infty,
\end{aligned}
\end{equation}
where $\psi_t$ is given in \eqref{F-3.4}. Since $K\in\cM_{\Phi^*}$, we have
$\Tr\tilde\rho KK^*=\Tr\rho\Phi^*(K)\Phi^*(K)^*$ and moreover by Proposition \ref{T-2.16},
\[
S_{x^2}^K(\tilde\rho\|\tilde\sigma)\le S_{x^2}^{\Phi^*(K)}(\rho\|\sigma),\qquad
S_{\psi_t}^K(\tilde\rho\|\tilde\sigma)\le S_{\psi_t}^{\Phi^*(K)}(\rho\|\sigma),\quad t\in(0,\infty).
\]
Hence equality \eqref{F-3.29} with $b\ge0$ implies that
\[
S_{\psi_t}^K(\tilde\rho\|\tilde\sigma)=S_{\psi_t}^{\Phi^*(K)}(\rho\|\sigma)\quad
\mbox{for $\nu_f$-almost everywhere\ $t\in(0,\infty)$}.
\]
Thanks to condition \eqref{F-3.5} for $\nu_f$ one can take a $T\subset\supp\nu_f$
such that $|T|=|\Sp(\Delta)\cup\Sp(\tilde\Delta)|$ and
\[
S_{\psi_t}^K(\tilde\rho\|\tilde\sigma)=S_{\psi_t}^{\Phi^*(K)}(\rho\|\sigma),\qquad t\in T.
\]
Therefore,
\begin{align*}
S_{\ffi_t}^K(\tilde\rho\|\tilde\sigma)
&={\Tr\tilde\rho KK^*\over1+t}-S_{\psi_t}^K(\tilde\rho\|\tilde\sigma)
={\Tr\rho\Phi^*(K)\Phi^*(K)^*\over1+t}-S_{\psi_t}^{\Phi^*(K)}(\rho\|\sigma) \\
&=S_{\ffi_t}^{\Phi^*(K)}(\rho\|\sigma),\qquad t\in T,
\end{align*}
from which (iv) holds for $h:=\sum_{t\in T}\ffi_t\in\OM[0,\infty)$ with $h(0)=h'(\infty)=0$.

Conversely, assume (iv) and $\rho^0\Phi^*(K)(I-\sigma^0)=0$. Since (vi) holds by
Theorem \ref{T-3.2}\,(3), one has
\begin{align}\label{F-3.30}
\Tr K^*\tilde\rho^\beta K\tilde\sigma^{1-\beta}
=\Tr\Phi^*(K)^*\rho^\beta\Phi^*(K)\sigma^{1-\beta}
\end{align}
for any fixed $\beta\in(1,2)$. One has $\Tr\Phi^*(K)^*\rho\Phi^*(K)(I-\sigma^0)=0$ by
Lemma \ref{L-2.8}, and moreover by Corollary \ref{C-2.19}\,(1) (for $\alpha=1$),
\[
\Tr K^*\tilde\rho K\tilde\sigma^0\ge\Tr\Phi^*(K)^*\rho\Phi^*(K)\sigma^0
=\Tr\Phi^*(K)^*\rho\Phi^*(K).
\]
Furthermore, since $K\in\cM_{\Phi^*}$,
$\Tr\Phi^*(K)^*\rho\Phi^*(K)=\Tr\rho\Phi^*(KK^*)=\Tr K^*\tilde\rho K$ so that one has
$\Tr K^*\tilde\rho K(I-\tilde\sigma^0)=0$ as well. Therefore, by Lemma \ref{L-2.7} and \eqref{F-3.30}
we see that $S_{x^\beta}^K(\tilde\rho\|\tilde\sigma)=S_{x^\beta}^{\Phi^*(K)}(\rho\|\sigma)$, that is,
(iv$'$) holds with $f(x)=x^\beta$.\qed

\medskip\noindent
{\bf Proof of Theorem \ref{T-3.3}\,(3).}\enspace
Assume that $K\in\cM_{\Phi^*}$ and $\rho^0\Phi(K)(I-\sigma^0)=0$. From the last part of the
above proof of (2) we have (vii$'$)$\implies$(iv$'$). Hence the conclusion of (3) follows from
Lemma \ref{L-3.1}, Theorem \ref{T-3.2}\,(3) and Theorem \ref{T-3.3}\,(1), (2) together.\qed

\begin{remark}\label{R-3.12}\rm
Let $\hat\Phi$ and $\hat K$ be as defined in \eqref{F-3.21} and \eqref{F-3.22}. Note that (vii$'$)
for $\rho,\sigma,\Phi,K$ is equivalent to (vii$'$) for $\rho,\sigma,\hat\Phi,\hat K$, like (vii) in
Lemma \ref{L-3.11}. However, since $K\in\cM_{\Phi^*}$ does not imply $\hat K\in\cM_{\hat\Phi^*}$,
the proof method of Theorem \ref{T-3.2}\,(2) breaks down when one would try to remove the
assumption $\rho^0\Phi^*(K)(I-\sigma^0)=0$ for the equivalence of (vii) and (vii$'$) in
Theorem \ref{T-3.3}\,(3). Indeed, Example \ref{E-3.21} in Sec.~\ref{Sec-3.4} shows that (vii$'$) is
not equivalent to (vii), without $\rho^0\Phi^*(K)(I-\sigma^0)=0$, even when $K\in\cM_{\Phi^*}$ and
$\rho^0=\sigma^0$.
\end{remark}

\subsection{Further results on equality cases}\label{Sec-3.3}

The next proposition is a supplement of Theorem \ref{T-3.2}, which clarifies when condition (i)
holds in the situation of Theorem \ref{T-3.2}\,(2).

\begin{prop}\label{P-3.13}
Let $\rho,\sigma,\Phi,K$ be as before, and assume that $\rho^0=\sigma^0$. If
$\Phi(\sigma)^0K=K\Phi(\sigma)^0$, then (i)--(ix) are all equivalent. Moreover, if $\sigma$ is
invertible and one (hence all) of (ii)--(ix) holds, then (i) holds if and only if
$\Phi(\sigma)^0K=K\Phi(\sigma)^0$.
\end{prop}

\begin{proof}
To show the first assertion, assume that $\tilde\sigma^0K=K\tilde\sigma^0$. By Lemma \ref{L-3.1}
and Theorem \ref{T-3.2}\,(2) we need to prove that if all of conditions (ii)--(ix) hold, then (i) does too.
Since we have (ii), it suffices to show that
$S_h^K(\tilde\rho\|\tilde\sigma)=S_h^{\Phi^*(K)}(\rho\|\sigma)$ for $h(x)\equiv1$ and $h(x)=x$, i.e.,
\begin{align}
\Tr\tilde\sigma K^*K&=\Tr\sigma\Phi^*(K)^*\Phi^*(K), \label{F-3.31}\\
\Tr\tilde\rho KK^*&=\Tr\rho\Phi^*(K)\Phi^*(K)^*.\label{F-3.32}
\end{align}
Letting $\alpha\searrow0$ and $\alpha\nearrow1$ in (vii) yields
\begin{align}
\Tr K^*\tilde\rho^0K\tilde\sigma=\Tr\Phi^*(K)^*\rho^0\Phi^*(K)\sigma, \label{F-3.33}\\
\Tr K^*\tilde\rho K\tilde\sigma^0=\Tr\Phi^*(K)^*\rho\Phi^*(K)\sigma^0. \label{F-3.34}
\end{align}
Since $\tilde\rho^0=\tilde\sigma^0$ follows from $\rho^0=\sigma^0$, we have
\begin{align*}
\Tr K^*K\tilde\sigma&=\Tr K^*K\tilde\sigma^0\tilde\sigma
=\Tr K^*\tilde\rho^0K\tilde\sigma \\
&=\Tr\Phi^*(K)^*\rho^0\Phi^*(K)\sigma\quad(\mbox{by \eqref{F-3.33}}) \\
&\le\Tr\Phi^*(K)^*\Phi^*(K)\sigma\le\Tr\Phi^*(K^*K)\sigma=\Tr K^*K\tilde\sigma,
\end{align*}
which gives \eqref{F-3.31}. Similarly,
\begin{align*}
\Tr K^*\tilde\rho K&=\Tr K^*\tilde\rho\tilde\sigma^0K
=\Tr K^*\tilde\rho K\tilde\sigma^0 \\
&=\Tr\Phi^*(K)^*\rho\Phi^*(K)\sigma^0\quad(\mbox{by \eqref{F-3.34}}) \\
&\le\Tr\rho\Phi^*(K)\Phi^*(K)^*\le\Tr\rho\Phi^*(KK^*)=\Tr\tilde\rho KK^*,
\end{align*}
which gives \eqref{F-3.32}.

To prove the second assertion, assume that $\rho^0=\sigma^0=I_\cH$ and (i) holds
(hence all (i)--(ix) hold). By \eqref{F-3.31} and \eqref{F-3.33} one has
\[
\Tr\tilde\sigma K^*K=\Tr\sigma\Phi^*(K)^*\rho^0\Phi^*(K)
=\Tr\tilde\sigma K^*\tilde\rho^0K=\Tr\tilde\sigma K^*\tilde\sigma^0K.
\]
Hence $\Tr\tilde\sigma(\tilde\sigma^0K^*(I-\tilde\sigma^0)K\tilde\sigma^0)=0$, which gives
$(I-\tilde\sigma^0)K\tilde\sigma^0=0$. Therefore, $K\tilde\sigma^0=\tilde\sigma^0K\tilde\sigma^0$.
Also, by \eqref{F-3.32} and \eqref{F-3.34} one has
\begin{align*}
\Tr\tilde\rho KK^*&=\Tr\rho\Phi^*(K)\sigma^0\Phi^*(K)^*
=\Tr\tilde\rho K\tilde\sigma^0K^*.
\end{align*}
This gives $K^*\tilde\sigma^0=\tilde\sigma^0K^*\tilde\sigma^0$ as above. Therefore,
$\tilde\sigma^0K=\tilde\sigma^0K\tilde\sigma^0=K\tilde\sigma^0$, as desired.
\end{proof}

For a trace-preserving positive map $\Phi:\cBH\to \cBK$ and a given
$\sigma\in\cBH^+$, define a positive map $\Phi_\sigma:\cBH\to \cBK$ by
\begin{align}\label{F-3.35}
\Phi_\sigma(X):=\tilde\sigma^{-1/2}\Phi(\sigma^{1/2}X\sigma^{1/2})\tilde\sigma^{-1/2},
\qquad X\in\cBH.
\end{align}
Since $\Phi_\sigma(I_\cH)=\tilde\sigma^0$, $\Phi_\sigma$ is unital when we regard it as a
map from $\cBH$ into $\cB(\tilde\sigma^0\cK)=\tilde\sigma^0\cBK\tilde\sigma^0$. It is easy
to see that the adjoint map $\Phi_\sigma^*:\cBK\to \cBH$ of $\Phi_\sigma$ is given by
\begin{align}\label{F-3.36}
\Phi_\sigma^*(Y)=\sigma^{1/2}\Phi^*(\tilde\sigma^{-1/2}Y\tilde\sigma^{-1/2})\sigma^{1/2},
\qquad Y\in\cBK.
\end{align}
We call the map $\Phi_\sigma^*$, as usual, the \emph{Petz recovery map} relative to $\sigma$.
Note that
\begin{align}\label{F-3.37}
 \Phi_\sigma^*(\tilde\sigma)=\sigma,
\end{align}
and $\Phi_\sigma^*$ is trace-preserving when restricted to $\cB(\tilde\sigma^0\cK)$
($\subset\cBK$). According to \cite[Proposition 2]{Je}, $\Phi_\sigma$ is not necessarily
a Schwarz map even when $\Phi^*$ is a Schwarz map. But, if $\Phi$ is $2$-positive, then
so are $\Phi^*$, $\Phi_\sigma$ and $\Phi_\sigma^*$.

\begin{prop}\label{P-3.14}
Let $\rho,\sigma\in\cBH^+$, $\Phi:\cBH\to \cBK$ and $K\in\cBK$ be as before.
\begin{itemize}
\item[(1)] Assume that either (a) $\rho^0=\sigma^0$, or (b) $K\in\cM_{\Phi^*}$. If
one (hence all) of conditions (ii)--(ix) holds, then the following holds:
\begin{itemize}
\item[(xi)] $\sigma^0\Phi^*(K)^*\rho\Phi^*(K)\sigma^0=\Phi_\sigma^*(K^*\Phi(\rho)K)$.
\end{itemize}
\item[(2)] Assume that $\Phi_\sigma$ is a Schwarz map, $K\in\cM_{\Phi^*}$ is a unitary,
and $\rho^0\Phi^*(K)(I_\cH-\sigma^0)=0$. Then condition (xi) becomes
$\Phi^*(K)^*\rho\Phi^*(K)=\Phi_\sigma^*(K^*\Phi(\rho)K)$, and it is equivalent to any of conditions
(i)--(x), (i$'$), (iv$'$), (vii$'$) and (ix$'$).
\end{itemize}
\end{prop}

\begin{proof}
(1)\enspace
{\it Case (a).}\enspace
Let $\hat\Phi:\cB(\sigma^0\cH)\to\cB(\tilde\sigma^0\cK)$ and $\hat K$ be those defined in
\eqref{F-3.21} and \eqref{F-3.22}. Let $\hat\Phi_\sigma^*$ be the Petz recovery map for $\hat\Phi$
relative to $\sigma$. By \eqref{F-3.24} the LHS of (xi) is
\[
\sigma^0\Phi^*(K)^*\sigma^0\rho\sigma^0\Phi^*(K)\sigma^0
=\hat\Phi^*(\hat K)^*\rho\hat\Phi^*(\hat K).
\]
On the other hand, since, by \eqref{F-3.23},
\begin{align*}
\Phi_\sigma^*(Y)&=\sigma^{1/2}\sigma^0\Phi^*(\tilde\sigma^{-1/2}Y\tilde\sigma^{-1/2})
\sigma^0\sigma^{1/2} \\
&=\sigma^{1/2}\hat\Phi^*(\tilde\sigma^0\tilde\sigma^{-1/2}Y\tilde\sigma^{-1/2}
\tilde\sigma^0)\sigma^{1/2}=\hat\Phi_\sigma^*(\tilde\sigma^0Y\tilde\sigma^0),
\qquad Y\in\cBK,
\end{align*}
the RHS of (xi) is
\[
\hat\Phi_\sigma^*(\tilde\sigma^0K^*\tilde\rho K\tilde\sigma^0)
=\hat\Phi_\sigma^*(\hat K^*\tilde\rho\hat K).
\]
By the above expressions and Lemma \ref{L-3.11} we may and do assume that
$\rho$, $\sigma$, $\tilde\rho$ and $\tilde\sigma$ are all invertible. Assume that one
(hence all) of (ii)--(ix) holds. Then by (viii) with $Y=\tilde\sigma^{-1/2}K^*\tilde\rho^{1/2}$ and
$z=1/2$ we have
\begin{align}\label{F-3.38}
\Phi^*(\tilde\sigma^{-1/2}K^*\tilde\rho K\tilde\sigma^{-1/2})
=\Phi^*(\tilde\sigma^{-1/2}K^*\tilde\rho^{1/2})\rho^{1/2}\Phi^*(K)\sigma^{-1/2}.
\end{align}
From the proof (iv)$\implies$(viii) of Theorem \ref{T-3.2}\,(2), it follows that
$K\in\cM_{\Phi^*}$. Then from the proof of (iv)$\implies$(x) of Theorem \ref{T-3.2}\,(3), we see
that \eqref{F-3.19} holds and its adjoint gives
\begin{align}\label{F-3.39}
\Phi^*(\tilde\sigma^{-1/2}K^*\tilde\rho^{1/2})=\sigma^{-1/2}\Phi^*(K)^*\rho^{1/2}.
\end{align}
Therefore, by \eqref{F-3.38} and \eqref{F-3.39} we have
\[
\Phi^*(\tilde\sigma^{-1/2}K^*\tilde\rho K\tilde\sigma^{-1/2})
=\sigma^{-1/2}\Phi^*(K)^*\rho\Phi^*(K)\sigma^{-1/2},
\]
which implies (xi) by \eqref{F-3.36}.

{\it Case (b).}\enspace
The proof is similar to case (a). By (viii) with $Y=\tilde\sigma^{-1/2}K^*\tilde\rho^{1/2}$
and $z=1/2$ we have
\begin{align}\label{F-3.40}
\sigma^0\Phi^*(\tilde\sigma^{-1/2}K^*\tilde\rho K\tilde\sigma^{-1/2})\sigma^0
=\sigma^0\Phi^*(\tilde\sigma^{-1/2}K^*\tilde\rho^{1/2})\rho^{1/2}\Phi^*(K)\sigma^{-1/2}.
\end{align}
Note that (x) holds by Theorem \ref{T-3.2}\,(3), and the analytic continuation of the adjoint of (x)
gives
\[
\sigma^0\Phi^*(\tilde\sigma^zK^*\tilde\rho^{-z})=\sigma^z\Phi^*(K)^*\rho^{-z},
\qquad z\in\bC.
\]
Hence with $z=-1/2$ we have
\begin{align}\label{F-3.41}
\sigma^0\Phi^*(\tilde\sigma^{-1/2}K^*\tilde\rho^{1/2})=\sigma^{-1/2}\Phi^*(K)^*\rho^{1/2}.
\end{align}
Combining \eqref{F-3.40} and \eqref{F-3.41} gives
\[
\sigma^0\Phi^*(\tilde\sigma^{-1/2}K^*\tilde\rho K\tilde\sigma^{-1/2})\sigma^0
=\sigma^{-1/2}\Phi^*(K)^*\rho\Phi^*(K)\sigma^{-1/2},
\]
so that (xi) follows.

(2)\enspace
Assume that $K\in\cM_{\Phi^*}$ is a unitary and $\rho^0\Phi^*(K)(I_\cH-\sigma^0)=0$. The first
assertion is seen by Lemma \ref{L-2.8}. For the second assertion, by Theorem \ref{T-3.3}\,(3)
as well as the above statement (1), it suffices to show that condition (xi) implies (vii). Since
\[
\Phi^*(K)^*\Phi^*(K)=\Phi^*(K^*K)=\Phi^*(I_\cK)=I_\cH,
\]
it follows that $\Phi^*(K)$ is a unitary too. From \eqref{F-3.37} and (xi) we have for any
$\alpha\in(0,1)$,
\begin{align*}
\Tr\Phi^*(K)^*\rho^\alpha\Phi^*(K)\sigma^{1-\alpha}
&=\Tr(\Phi^*(K)\rho\Phi^*(K))^\alpha\sigma^{1-\alpha} \\
&\le\Tr(K^*\tilde\rho K)^\alpha\tilde\sigma^{1-\alpha}
=\Tr K^*\tilde\rho^\alpha K\tilde\sigma^{1-\alpha},
\end{align*}
where the inequality above is due to the monotonicity of $S_{x^\alpha}^I$ under $\Phi_\sigma^*$.
Since Corollary \ref{C-2.19}\,(1) gives the reverse inequality, (vii) follows.
\end{proof}

\begin{remark}\label{R-3.15}\rm
The formula in (xi) is considered as a type of reversibility of $\rho$ from $\Phi(\rho)$.
In fact, in the situation of Proposition \ref{P-3.14}\,(2), the operator $\Phi^*(K)$ is a unitary (as
seen in the above proof) and (xi) gives
\begin{align}\label{F-3.42}
\rho=\Phi^*(K)\Phi_\sigma^*(K^*\Phi(\rho)K)\Phi^*(K)^*.
\end{align}
In particular, when $K=I_\cK$ and $\rho^0\le\sigma^0$, condition (xi) reduces to
\begin{align}\label{F-3.43}
\rho=\Phi_\sigma^*(\Phi(\rho)),
\end{align}
which is the standard expression of reversibility, so that Proposition \ref{P-3.14}\,(2)
more or less reduces to \cite[Theorem 5.1]{HMPB} and \cite[Theorem 3.18]{HM}. Here, note that
\eqref{F-3.43} is equivalent to any of (i)--(ix$'$) (with $K=I_\cK$) stated in
Proposition \ref{P-3.14}\,(2) even without the assumption that $\Phi_\sigma$ is a Schwarz map
(hence without the $2$-positivity assumption of $\Phi$ assumed in the above cited theorems in
\cite{HMPB,HM}). Indeed, since $\Phi_\sigma^*$ restricted to $\cB(\tilde\sigma^0\cK)$ is a
trace-preserving positive map, condition \eqref{F-3.43}, together with \eqref{F-3.37}, implies that
$D(\Phi(\rho)\|\Phi(\sigma))=D(\rho\|\sigma)$ thanks to the monotonicity property of the relative
entropy \cite{MR}, which is condition (iv$'$) (with $K=I_\cK$) for $f(x)=x\log x$. Hence the assertion
follows from Theorem \ref{T-3.3}\,(3) and Proposition \ref{P-3.14}\,(1).
\end{remark}

In the next proposition, for any $n\in\bN$, $n\ge2$, we consider the TPCP map
\begin{align}\label{F-3.44}
\Phi:\cB(\underbrace{\cH\oplus\dots\oplus\cH}_{n})=\cB(\cH\otimes\bC^n)\to\cBH,\quad
\begin{bmatrix}A_{11}&\cdots&A_{1n}\\\vdots&\ddots&\vdots\\A_{n1}&\cdots&A_{nn}
\end{bmatrix}\mapsto\sum_{k=1}^nA_{kk},
\end{align}
which is indeed the partial trace $\T_{\bC^n}:\cB(\cH\otimes\bC^n)\to\cBH$. In the next
proposition, rewriting conditions (i)--(x) in this setting, we have equivalent conditions for equality
cases in joint concavity of quasi-entropies and in Lieb's concavity, as well as in joint convexity
of quasi-entropies and in Ando's convexity.

\begin{prop}\label{P-3.16}
Let $\rho_k,\sigma_k\in\cBH^+$, $1\le k\le n$, and let $K\in\cBH$.
\begin{itemize}
\item[(1)] The following conditions (I)--(X) are equivalent:
\begin{itemize}
\item[(I)] $S_h^K\bigl(\sum_{k=1}^n\rho_k\big\|\sum_{k=1}^n\sigma_k\bigr)
=\sum_{k=1}^nS_h^K(\rho_k\|\sigma_k)$ for all $h\in\OM[0,\infty)$.
\item[(IV)] $S_h^K\bigl(\sum_{k=1}^n\rho_k\big\|\sum_{k=1}^n\sigma_k\bigr)
=\sum_{k=1}^nS_h^K(\rho_k\|\sigma_k)$ for some $h\in\OM_+[0,\infty)$ with
$|\supp\mu_h|\ge\big|\bigcup_{k=1}^n\Sp(\Delta_{\rho_k,\sigma_k})
\cup\Sp\bigl(\Delta_{\sum_{k=1}^n\rho_k,\sum_{k=1}^n\sigma_k}\bigr)\big|$.
\item[(V)] $\Tr Y\bigl(\sum_{k=1}^n\rho_k\bigr)^zK\bigl(\sum_{k=1}^n\sigma_k\bigr)^{1-z}
=\sum_{k=1}^n\Tr Y\rho_k^zK\sigma_k^{1-z}$ for all $Y\in\cBH$ and $z\in\bC$.
\item[(VII)] $\Tr K^*\bigl(\sum_{k=1}^n\rho_k\bigr)^\alpha K
\bigl(\sum_{k=1}^n\sigma_k\bigr)^{1-\alpha}
=\sum_{k=1}^n\Tr K^*\rho_k^\alpha K\sigma_k^{1-\alpha}$ for some $\alpha\in(0,1)$.
\item[(IX)] $\sigma_j^0K^*\bigl(\sum_{k=1}^n\rho_k\bigr)^\alpha K
\bigl(\sum_{k=1}^n\sigma_k\bigr)^{-\alpha}\sigma_j^0
=\sigma_j^0K^*\rho_j^\alpha K\sigma_j^{-\alpha}$, $1\le j\le n$, for some $\alpha\in(0,1)$.
\item[(X)] $\bigl(\sum_{k=1}^n\rho_k\bigr)^{it}K
\bigl(\sum_{k=1}^n\sigma_k\bigr)^{-it}\sigma_j^0
=\rho_j^{it}K\sigma_j^{-it}$, $1\le j\le n$, for all $t\in\bR$.
\end{itemize}

If the above equivalent conditions hold, then the following holds too:
\begin{itemize}
\item[(XI)] $\sigma_j^0K^*\rho_jK\sigma_j^0=\sigma_j^{1/2}
\bigl(\sum_{k=1}^n\sigma_k\bigr)^{-1/2}K^*\bigl(\sum_{k=1}^n\rho_k\bigr)K
\bigl(\sum_{k=1}^n\sigma_k\bigr)^{-1/2}\sigma_j^{1/2}$, $1\le j\le n$.
\end{itemize}

\item[(2)] If $\rho_k^0K(I_\cH-\sigma_k^0)=0$ for $1\le k\le n$ (in particular, if
$\sigma_k\in\cBH^{++}$ for $1\le k\le n$), then conditions (I)--(X) in (1) and the following
conditions (I$'$)--(IX$'$) are equivalent:
\begin{itemize}
\item[(I\,$'$)] $S_f^K\bigl(\sum_{k=1}^n\rho_k\big\|\sum_{k=1}^n\sigma_k\bigr)
=\sum_{k=1}^nS_f^K(\rho_k\|\sigma_k)$ for all $f\in\OC(0,\infty)$.
\item[(IV\,$'$)] $S_f^K\bigl(\sum_{k=1}^n\rho_k\big\|\sum_{k=1}^n\sigma_k\bigr)
=\sum_{k=1}^nS_f^K(\rho_k\|\sigma_k)$ for some $f\in\OC[0,\infty)$
with $f(0)=0$, $f'(\infty)=+\infty$ and
$|\supp\nu_f|\ge\big|\bigcup_{k=1}^n\Sp(\Delta_{\rho_k,\sigma_k})
\cup\Sp\bigl(\Delta_{\sum_{k=1}^n\rho_k,\sum_{k=1}^n\sigma_k}\bigr)\big|$.
\item[(VII\,$'$)] $\Tr K^*\bigl(\sum_{k=1}^n\rho_k\bigr)^\alpha K
\bigl(\sum_{k=1}^n\sigma_k\bigr)^{1-\alpha}
=\sum_{k=1}^n\Tr K^*\rho_k^\alpha K\sigma_k^{1-\alpha}$ for some $\alpha\in(1,2)$.
\item[(IX\,$'$)] $\sigma_j^0K^*\bigl(\sum_{k=1}^n\rho_k\bigr)^\alpha K
\bigl(\sum_{k=1}^n\sigma_k\bigr)^{-\alpha}\sigma_j^0
=\sigma_j^0K^*\rho_j^\alpha K\sigma_j^{-\alpha}$, $1\le j\le n$, for some $\alpha\in(1,2)$.
\end{itemize}

\item[(3)] If $K$ is a unitary and $\rho_k^0K(I_\cH-\sigma_k^0)=0$ for $1\le k\le n$, then condition
(XI) is equivalent to any of (I)--(X) in (1) and (I\,$'$)--(IX\,$'$) in (2), where the LHS of the equality
in (XI) becomes $K^*\rho_kK$.
\end{itemize}
\end{prop}

\begin{proof}
Let $\rho:=\rho_1\oplus\dots\oplus\rho_n$, $\sigma:=\sigma_1\oplus\dots\oplus\sigma_n$,
and $\Phi$ be defined by \eqref{F-3.44}. Then conditions (i)--(x) given in Sect.~\ref{Sec-3.1} and
(xi) in Proposition \ref{P-3.14} are respectively rewritten as (I)--(XI) above (while not all of (i)--(x)
are rewritten here). Also, conditions (i$'$), (iii$'$), (v$'$) and (vii$'$) are respectively rewritten as
(I$'$), (IV$'$), (VII$'$) and (IX$'$), and $\rho^0\Phi^*(K)(I_{\cH\otimes\bC^n}-\sigma^0)=0$ is
equivalent to $\rho_k^0\Phi^*(K)(I_\cH-\sigma_k^0)=0$ for $1\le k\le n$. Hence (1)
and (2) follow from Corollary \ref{C-3.4}, and (3) does from Proposition \ref{P-3.14}\,(2).
\end{proof}

\begin{remark}\label{R-3.17}\rm
In the particular case that $K=I_\cH$, the equivalence of conditions (VII) (also (VII) for all
$\alpha\in(0,1)$), (X) and (XI), as well as the equivalence of (I$'$), (IV$'$), (VII$'$) (also (VII$'$) for
all $\alpha\in(1,2)$), (X) and (XI) when $\rho_k^0\le\sigma_k^0$ for $1\le k\le n$, was given in
\cite[Theorem 10]{JR} and \cite[Corollary 5.3]{HMPB}.
\end{remark}

\subsection{Examples}\label{Sec-3.4}

In this subsection we present four examples, which may help the us to understand the results in
Secs.~\ref{Sec-3.1} and \ref{Sec-3.3} in a more concrete fashion. The first example shows that
equivalent conditions (i)--(x) in Corollary \ref{C-3.4} can give the strongest restriction on $K$ even
when $\rho=\sigma$.

\begin{example}\label{E-3.18}\rm
In the situation of Corollary \ref{C-3.4}, we consider the special case that
$\rho=\sigma=|\psi\>\<\psi|$, a pure state in $\cH\otimes\cK$. For every $K\in\cBH$ and any
$\alpha>0$, note that
\begin{align*}
\Tr\Phi^*(K)^*\sigma^\alpha\Phi^*(K)\sigma^{1-\alpha}
&=\Tr(K^*\otimes I)|\psi\>\<\psi|(K\otimes I)|\psi\>\<\psi| \\
&=\<\psi,(K\otimes I)\psi\>\Tr(K^*\otimes I)|\psi\>\<\psi| \\
&=\Tr(K\otimes I)|\psi\>\<\psi|\cdot\Tr(K^*\otimes I)|\psi\>\<\psi| \\
&=\Tr K\tilde\sigma\cdot\Tr K^*\tilde\sigma=|\Tr K\tilde\sigma|^2,
\end{align*}
and the Schwarz inequality gives
\begin{align*}
|\Tr K\tilde\sigma|^2
&=|\Tr\tilde\sigma^{\alpha\over2}K\tilde\sigma^{1-\alpha\over2}\tilde\sigma^{1/2}|^2
=|\<\tilde\sigma^{1/2},\tilde\sigma^{\alpha\over2}K\tilde\sigma^{1-\alpha\over2}\>_\HS|^2 \\
&\le\|\tilde\sigma^{1/2}\|_\HS^2\|
\tilde\sigma^{\alpha\over2}K\tilde\sigma^{1-\alpha\over2}\|_\HS^2
=\Tr\tilde\sigma^{1-\alpha\over2}K^*\tilde\sigma^\alpha K\tilde\sigma^{1-\alpha\over2}
=\Tr K^*\tilde\sigma^\alpha K\tilde\sigma^{1-\alpha}.
\end{align*}
Therefore, for any $\alpha>0$,
\begin{align}\label{F-3.45}
\Tr K^*\tilde\sigma^\alpha K\tilde\sigma^{1-\alpha}
\ge\Tr\Phi^*(K)^*\sigma^\alpha\Phi^*(K)\sigma^{1-\alpha}
\end{align}
and from the equality case of the Schwarz inequality, equality holds in \eqref{F-3.45} (i.e.,
condition (vii) in the present case) if and only if
$\tilde\sigma^{\alpha\over2}K\tilde\sigma^{1-\alpha\over2}=\lambda\tilde\sigma^{1/2}$ for
some $\lambda\in\bC$. The last condition is equivalent to
\begin{align}\label{F-3.46}
\tilde\sigma^0K\tilde\sigma^0=\lambda\tilde\sigma^0\quad
\mbox{for some $\lambda\in\bC$},
\end{align}
which is a characterization for (i)--(x) to hold in this special case. When $\cK=\cH$ and
$|\psi\>\<\psi|$ is a maximal entangled state (see, e.g., \cite{Pe3}) so that
$\tilde\sigma^0=I_\cH$, this becomes quite a strong condition that $K=\lambda I_\cH$ for some
$\lambda\in\bC$.
\end{example}

In the next two examples, we specialize the conditional expectation considered in
Remark \ref{R-3.6} to show that conditions (i) and (xi) cannot be included in
Theorem \ref{T-3.2}\,(2) and Proposition \ref{P-3.14}\,(1), respectively.

\begin{example}\label{E-3.19}\rm
Let $\{e_i\}_{i=1}^d$ be an orthonormal basis of $\cH$ and $\Phi:\cBH\to \cBH$ be
defined by
\[
\Phi(K):=\sum_{i=1}^d\<e_i,Ke_i\>|e_i\>\<e_i|,\qquad K\in\cBH,
\]
that is the trace-preserving conditional expectation $E_\cA$ onto the abelian subalgebra
$\cA$ generated by $\{|e_i\>\<e_i|\}_{i=1}^d$, so that $\Phi^*=\Phi$. Let
$\rho=\sigma=|e_1\>\<e_1|$. Then $\Phi(\rho)=\Phi(\sigma)=|e_1\>\<e_1|$ and $\hat\Phi$
is the identity map on the one-dimensional $\cB(\sigma^0\cH)=\cB(\bC e_1)\cong\bC$. Hence
by Lemma \ref{L-3.11} and Theorem \ref{T-3.2}\,(2), conditions (ii)--(ix) hold trivially for any
$K\in\cBH$. Indeed, it is also easy to confirm this directly without using Lemma \ref{L-3.11}.
On the other hand, (i) holds for $\rho,\sigma$ and $K$ if and only if \eqref{F-3.33} and
\eqref{F-3.34} hold, i.e.,
\begin{align*}
\Tr|e_1\>\<e_1|K^*K=\Tr|e_1\>\<e_1|\Phi(K)^*\Phi(K), \\
\Tr|e_1\>\<e_1|KK^*=\Tr|e_1\>\<e_1|\Phi(K)\Phi(K)^*,
\end{align*}
that is,
\[
\|Ke_1\|^2=\|\Phi(K)e_1\|^2,\qquad\|K^*e_1\|^2=\|\Phi(K)^*e_1\|^2.
\]
These means that $\<e_i,Ke_1\>=\<e_i,K^*e_1\>=0$ for all $i\ne1$, or equivalently, $\bC e_1$
is a reducing subspace of $K$. This example shows that condition (i) is strictly stronger than
(ii)--(ix) in the situation of Theorem \ref{T-3.2}\,(2), and also that the assertion in
Lemma \ref{L-3.11} is invalid for condition (i).
\end{example}

\begin{example}\label{E-3.20}\rm
Let $\Phi$ be the trace-preserving conditional expectation $E$ from $\cB(\bC^2)=\bM_2$
(the $2\times2$ matrices) onto the $2\times2$ diagonal matrices, and set
\[
K=\begin{bmatrix}x_{11}&x_{12}\\x_{21}&x_{22}\end{bmatrix},\qquad
\sigma=I_2,\qquad
\rho^{1/2}=\begin{bmatrix}a&c\\\overline c&b\end{bmatrix}.
\]
Assume that $a,b>0$ and $|c|^2<ab$, hence $\rho$ is invertible so that case (a) of
Proposition \ref{P-3.14}\,(1) is met. Hence Proposition \ref{P-3.14}\,(1) says that condition
(vii) for $\alpha=1/2$ implies (xi). Here we show that the converse implication fails. In the present
case, (vii) for $\alpha=1/2$ means that
\begin{align}\label{F-3.47}
\Tr K^*E(\rho)^{1/2}K=\Tr E(K)^*\rho^{1/2}E(K),
\end{align}
while (xi) means that
\begin{align}\label{F-3.48}
E(K)^*\rho E(K)=E(K^*E(\rho)K).
\end{align}
A direct computation shows that \eqref{F-3.47} is explicitly rewritten as
\[
\sqrt{a^2+|c|^2}(|x_{11}|^2+|x_{12}|^2)+\sqrt{b^2+|c|^2}(|x_{21}|^2+|x_{22}|^2)
=a|x_{11}|^2+b|x_{22}|^2.
\]
In view of the assumption $a,b>0$, this holds if and only if
\begin{align}\label{F-3.49}
c=x_{12}=x_{21}=0\quad\mbox{or}\quad K=0.
\end{align}
On the other hand, we can directly rewrite \eqref{F-3.48} as
\begin{align*}
&\begin{bmatrix}(a^2+|c|^2)|x_{11}|^2&(a+b)c\overline{x}_{11}x_{22}\\
(a+b)\overline cx_{11}\overline{x}_{22}&(b^2+|c|^2)|x_{22}|^2\end{bmatrix} \\
&\quad=\begin{bmatrix}(a^2+|c|^2)|x_{11}|^2+(b^2+|c|^2)|x_{21}|^2&0\\
0&(a^2+|c|^2)|x_{12}|^2+(b^2+|c|^2)|x_{22}|^2\end{bmatrix},
\end{align*}
which holds if and only if
\begin{align}\label{F-3.50}
x_{12}=x_{21}=0\quad\mbox{and}\quad
[c=0\ \ \mbox{or}\ \ x_{11}=0\ \ \mbox{or}\ \ x_{22}=0].
\end{align}
It is obvious that \eqref{F-3.49}$\implies$\eqref{F-3.50}, as Proposition \ref{P-3.14}\,(1)
says. However, when $x_{11}=x_{12}=x_{21}=0$ and $x_{22}\ne0$ as well as $0<|c|^2<ab$,
\eqref{F-3.50} holds but \eqref{F-3.49} does not hold, so (xi)$\implies$(vi) fails, even when
both assumptions (a) and (b) of Proposition \ref{P-3.14}\,(1) are met.
\end{example}

The next example shows that the assumption $\rho^0\Phi^*(K)(I_\cH-\sigma^0)=0$ is
essential in Theorem \ref{T-3.3}\,(3) .

\begin{example}\label{E-3.21}\rm
Let $\cH$, $\cK$ and $\cH_1$ be finite-dimensional Hilbert spaces of dimension $\ge2$. Let
$\Phi:=\Tr_\cK:\cB(\cH\otimes\cK)\to\cBH$ be the partial trace over $\cK$, and choose a pure
state $\sigma=|\psi\>\<\psi|$ and $K\in\cBH$ such that
$\tilde\sigma^0K\tilde\sigma^0\not\in\bC\tilde\sigma^0$ where
$\tilde\sigma:=\Phi(\sigma)=\Tr_\cK\sigma$. Then, as shown in Example \ref{E-3.18}
(in particular, see \eqref{F-3.45} and \eqref{F-3.46}), it holds that
\begin{align}\label{F-3.51}
\Tr K^*\Phi(\sigma)^\alpha K\Phi(\sigma)^{1-\alpha}
>\Tr\Phi^*(K)^*\sigma^\alpha\Phi^*(K)\sigma^{1-\alpha},\qquad\alpha\in(0,\infty).
\end{align}
Next, let $\rho_1,\sigma_1\in\cB(\cH_1)^{++}$ be non-commuting. Then by
\cite[Corollary A.5]{HMPB} (i.e., non-equality cases of H\"older's and inverse H\"older's
inequalities),
\begin{align}
(\Tr\rho_1)^\alpha(\Tr\sigma_1)^{1-\alpha}&>\Tr\rho_1^\alpha\sigma_1^{1-\alpha},
\qquad\alpha\in(0,1), \label{F-3.52}\\
(\Tr\rho_1)^\beta(\Tr\sigma_1)^{1-\beta}&<\Tr\rho_1^\beta\sigma_1^{1-\beta},
\qquad\beta\in(1,2). \label{F-3.53}
\end{align}
Now define $\cH_2:=(\cH\otimes\cK)\oplus\cH_1$, $\cK_2:=\cH\oplus\bC$, a CPTP map
$\Phi_2:\cB(\cH_2)\to\cB(\cK_2)$ by
\[
\Phi_2:\begin{bmatrix}A&C\\D&B\end{bmatrix}\mapsto
\begin{bmatrix}\Phi(A)&0\\0&\Tr B\end{bmatrix}
\quad\mbox{for $A\in\cB(\cH\otimes\cK)$, $B\in\cB(\cH_1)$},
\]
and moreover $K_2:=\begin{bmatrix}K&0\\0&1\end{bmatrix}\in\cB(\cK\oplus\bC)$. Then it is
immediate to see $\Phi_2^*(K_2)=\begin{bmatrix}\Phi^*(K)&0\\0&I_{\cH_1}\end{bmatrix}$ and
hence $K_2\in\cM_{\Phi_2^*}$. For any $\beta\in(1,2)$, from \eqref{F-3.51} and \eqref{F-3.53}
we can choose a $\lambda>0$ such that
\[
\Tr K^*\Phi(\sigma)^\beta K\Phi(\sigma)^{1-\beta}
+\lambda(\Tr\rho_1)^\beta(\Tr\sigma_1)^{1-\beta}
=\Tr\Phi^*(K)^*\sigma^\beta\Phi^*(K)\sigma^{1-\beta}
+\lambda\Tr\rho_1^\beta\sigma_1^{1-\beta}.
\]
Replacing $\rho_1,\sigma_1$ with $\lambda\rho_1,\lambda\sigma_1$ we may assume
$\lambda=1$. Define $\rho_2:=\begin{bmatrix}\sigma&0\\0&\rho_1\end{bmatrix}$
$\sigma_2:=\begin{bmatrix}\sigma&0\\0&\sigma_1\end{bmatrix}\in\cB(\cH_2)^+$;
then we have
\[
\Tr K_2^*\Phi_2(\rho_2)^\beta K_2\Phi_2(\sigma_2)^{1-\beta}
=\Tr\Phi_2^*(K_2^*)\rho_2^\beta\Phi_2(K_2)\sigma_2^{1-\beta}.
\]
On the other hand, it follows from \eqref{F-3.51} and \eqref{F-3.52} that
\[
\Tr K_2^*\Phi_2(\rho_2)^\alpha K_2\Phi_2(\sigma_2)^{1-\alpha}
>\Tr\Phi_2^*(K_2^*)\rho_2^\alpha\Phi_2(K_2)\sigma_2^{1-\alpha},
\qquad\alpha\in(0,1).
\]
Therefore, for $\rho_2,\sigma_2,\Phi_2,K_2$, condition (vii$'$) holds but (vii) does not, so
(vii$'$)$\implies$(vii) is not true, where $K_2\in\cM_{\Phi_2^*}$ and $\rho_2^0=\sigma_2^0$
are satisfied but $\rho_2^0\Phi_2^*(K_2)(I_{\cH_2}-\sigma_2^0)\ne0$ (otherwise, (vii) must hold
due to Theorem \ref{T-3.3}\,(3)).
\end{example}

\section{Equality cases for monotone metrics and $\chi^2$-divergences}\label{Sec-4}

Recall that $\cBH^{++}$ is a Riemannian manifold and the tangent space at each foot point
$\sigma\in\cBH^{++}$ is identified with $\cBH^\sa$, the self-adjoint operators on $\cH$. By
\emph{monotone metrics} we mean a family of Riemannian metrics $\gamma_\sigma(X,Y)$ for
$\sigma\in\cBH^{++}$ and $X,Y\in\cBH^\sa$ with any finite-dimensional $\cH$ such that
\[
\gamma_{\Phi(\sigma)}(\Phi(X),\Phi(X))\le\gamma_\sigma(X,X),
\qquad\sigma\in\cBH^{++},\ X\in\cBH^\sa
\]
for every TPCP map $\Phi:\cBH\to\cBK$. In \cite{Pe6} Petz showed that such monotone metrics,
when they satisfy $\gamma_\sigma(X,X)=\Tr\sigma^{-1}X^2$ if $\sigma X=X\sigma$, are in the
form
\[
\gamma_\sigma(X,Y)=\Tr XJ_h(\sigma,\sigma)^{-1}Y,\qquad
\sigma\in\cBH^{++},\ X,Y\in\cBH^\sa,
\]
where $h\in\OM_+[0,\infty)$ with $h(1)=1$ and $J_h(\sigma,\sigma)=h(\Delta_\sigma)R_\sigma$
(see \eqref{F-2.3}). (This characterization was refined later in \cite{Ku}.) Monotone metrics are also
called \emph{quantum Fisher informations} because they are quantum analogs of the classical
Fisher information.

In this section, following \cite{HP}, we discuss the extended version of monotone metrics in two
points $\rho,\sigma$ and non-self-adjoint $X,Y$, which are defined by
\begin{align}\label{F-4.1}
\gamma_{\rho,\sigma}^h(X,Y):=\<X,J_h(\rho,\sigma)^{-1}Y\>_\HS,\qquad
\rho,\sigma\in\cBH^{++},\ X,Y\in\cBH,
\end{align}
with parametrization $h\in\OM_+[0,\infty)$ (without assuming $h(1)=1$). This extension is
meaningful in the sense that it includes functions of Lieb type (as in (vi) of Theorem \ref{T-4.1}).
Note that (II$'$) of Theorem \ref{T-2.2} says the monotonicity property
\begin{align}\label{F-4.2}
\gamma_{\Phi(\rho),\Phi(\sigma)}^h(\Phi(X),\Phi(X))\le\gamma_{\rho,\sigma}^h(X,X)
\end{align}
for any trace-preserving map $\Phi$ such that $\Phi^*$ is a Schwarz map, where the above LHS
is defined as a two-point version of monotone metrics on $\cB(Q\cH)^{++}$ where $Q:=\Phi(I_\cH)^0$,
since $\Phi(\rho),\Phi(\sigma)\in\cB(Q\cH)^{++}$ and $\Phi(X)\in\cB(Q\cH)$. Note also that (IV) of
Theorem \ref{T-2.2} says the joint convexity of
$(\rho,\sigma,X)\mapsto\gamma_{\rho,\sigma}^h(X,X)$ on $\cBH^{++}\times\cBH^{++}\times\cBH$.
Since $\gamma_{\rho,\sigma}^h(X,Y)$ is positively homogeneous, i.e.,
$\gamma_{\alpha\rho,\alpha\sigma}^h(\alpha X,\alpha Y)=\alpha\gamma_{\rho,\sigma}^h(X,Y)$ for
$\alpha>0$ as immediately verified, it is clear that the above joint convexity is equivalent to the joint
subadditivity
\[
\gamma_{\rho_1+\rho_2}^h(X_1+X_2,X_1+X_2)
\le\gamma_{\rho_1,\sigma_1}^h(X_1,X_1)+\gamma_{\rho_2,\sigma_2}^h(X_2,X_2)
\]
for $\rho_i,\sigma_i\in\cBH^{++}$, $X_i\in\cBH$.

Furthermore, it is worth pointing out that
\begin{align}\label{F-4.3}
\gamma_{\rho,\sigma}^h(X,X)=S_{h^*}^X(\rho^{-1}\|\sigma^{-1}),\qquad
\rho,\sigma\in\cBH^{++},\ X\in\cBH,
\end{align}
since $J_h(\rho,\sigma)^{-1}=h^*(\Delta_{\rho^{-1},\sigma^{-1}})R_{\sigma^{-1}}$, where
$h^*(x):=h(x^{-1})^{-1}$ for $x>0$, a function in $\OM_+[0,\infty)$ called the \emph{adjoint} of $h$
(see \cite{KA}, \cite[p.~194]{Hi0}). It is well known \cite{Han,AH} that a function
$k\in\OMD_+(0,\infty)$ has the integral expression
\begin{align}\label{F-4.4}
k(x)=a+{b\over x}+\int_{(0,\infty)}{1+t\over x+t}\,d\lambda_k(t),\qquad x\in(0,\infty),
\end{align}
where $a,b\ge0$ and $\lambda_k$ is a (unique) finite positive measure on $(0,\infty)$. When
$h\in\OM_+[0,\infty)$ and $k=1/h$ on $(0,\infty)$, the integral expression of $h^*$ (like \eqref{F-3.1})
is given as
\begin{equation}\label{F-4.5}
\begin{aligned}
h^*(x)&=a+bx+\int_{(0,\infty)}{1+t\over x^{-1}+t}\,d\lambda_{1/h}(t) \\
&=a+bx+\int_{(0,\infty)}{x(1+t)\over x+t}\,d\lambda_{1/h}(t^{-1}),\qquad x\in(0,\infty),
\end{aligned}
\end{equation}
so that $d\mu_{h^*}(t)=d\lambda_{1/h}(t^{-1})$. Expression \eqref{F-4.5} gives the integral
expression
\begin{align}\label{F-4.6}
J_h(\rho,\sigma)^{-1}=aR_\sigma^{-1}+bL_\rho^{-1}
+\int_{(0,\infty)}J_{h_t}(\rho,\sigma)^{-1}(1+t)\,d\lambda_{1/h}(t),
\end{align}
where $h_t(x):=x+t$ and $J_{h_t}(\rho,\sigma)^{-1}=(\Delta_{\rho,\sigma}+tI)^{-1}R_\sigma^{-1}$.

The next theorem is concerned with equality conditions for the monotonicity of
$\gamma_{\rho,\sigma}^h(X,X)$ and for the monotonicity version of Lieb's $3$-variable convexity
theorem \cite[Corollary 2.1]{Li} (see Corollary \ref{C-2.19}\,(3)). The theorem is a modification of
Jen\v cov\'a's result in \cite[Proposition 4]{Je}, while we treat two-point version
$\gamma_{\rho,\sigma}^h$ and our assumption on $\Phi$ is slightly weaker than that in \cite{Je}
(see Remark \ref{R-4.4}). The proof is similar to but simpler than those in Sec.~\ref{Sec-3.2}
because of simplicity of \eqref{F-4.1} for $\rho,\sigma\in\cBH^{++}$ compared with \eqref{F-2.8} for
$\rho,\sigma\in\cBH^+$. We give a sketchy proof here for completeness and for the convenience of
the reader. 

\begin{thm}\label{T-4.1}
Let $\Phi:\cBH\to\cBK$ be a trace-preserving map such that $\Phi^*$ is a Schwarz map. Let
$\rho,\sigma\in\cBH^{++}$ and $K\in\cBH$. Then the following conditions are equivalent:
\begin{itemize}
\item[(i)] $\gamma_{\Phi(\rho),\Phi(\sigma)}^h(\Phi(K),\Phi(K))=\gamma_{\rho,\sigma}^h(K,K)$ for
all $h\in\OM_+[0,\infty)$.
\item[(ii)] $\gamma_{\Phi(\rho),\Phi(\sigma)}^h(\Phi(K),\Phi(K))=\gamma_{\rho,\sigma}^h(K,K)$ for
some $h\in\OM_+[0,\infty)$ with $|\supp\lambda_{1/h}|\ge
|\Sp(\Delta_{\rho,\sigma})\cup\Sp(\Delta_{\Phi(\rho),\Phi(\sigma)})|$ (see \eqref{F-4.4}).
\item[(iii)] $S_h^{\Phi(K)}(\Phi(\rho)^{-1}\|\Phi(\sigma)^{-1})=S_h^K(\rho^{-1}\|\sigma^{-1})$ for all
$h\in\OM_+[0,\infty)$.
\item[(iv)] $S_h^{\Phi(K)}(\Phi(\rho)^{-1}\|\Phi(\sigma)^{-1})=S_h^K(\rho^{-1}\|\sigma^{-1})$ for some
$h\in\OM_+[0,\infty)$ with \eqref{F-3.5}.
\item[(v)] $\Tr\Phi(K)^*\Phi(\rho)^{-z}\Phi(K)\Phi(\sigma)^{z-1}=\Tr K^*\rho^{-z}K\sigma^{z-1}$ for all
$z\in\bC$.
\item[(vi)] $\Tr\Phi(K)^*\Phi(\rho)^{-\alpha}\Phi(K)\Phi(\sigma)^{\alpha-1}=
\Tr K^*\rho^{-\alpha}K\sigma^{\alpha-1}$ for some $\alpha\in(0,1)$.
\item[(vii)] $\Phi^*(\Phi(\rho)^{-z}\Phi(K)\Phi(\sigma)^{z-1})=\rho^{-z}K\sigma^{z-1}$ for all $z\in\bC$.
\end{itemize}
\end{thm}

\begin{proof} (Sketch)\enspace
By \eqref{F-4.3} and \eqref{F-4.5} we have (i)$\iff$(iii) and (ii)$\iff$(iv). (i)$\implies$(ii) and
(v)$\implies$(vi) are trivial. (vi)$\implies$(ii) is immediately seen since (vi) means that (ii) holds for
$h(x)=x^\alpha$ for some $\alpha\in(0,1)$. (vii)$\implies$(v) is easy as well.

(i)$\implies$(vii).\enspace
Condition (i) means that
\[
\<K,[\Phi^*J_h(\Phi(\rho),\Phi(\sigma))^{-1}\Phi]K\>_\HS
=\<K,J_h(\rho,\sigma)^{-1}K\>_\HS,
\]
which yields by Theorem \ref{T-2.2}\,(II$'$),
\[
\bigl[J_h(\rho,\sigma)^{-1}-\Phi^*J_h(\Phi(\rho),\Phi(\sigma))^{-1}\Phi\bigr]K=0.
\]
Hence (vii) follows by letting $h(x)=x^\alpha$ for any $\alpha\in(0,1)$ and then by analytic
continuation.

(ii)$\implies$(i).\enspace
In view of expression \eqref{F-4.6} and Theorem \ref{T-2.2}\,(II$'$), condition (ii) implies that there
exists a $T\subset\supp\lambda_{1/h}$ such that
$|T|\ge|\Sp(\Delta_{\rho,\sigma})\cup\Sp(\Delta_{\Phi(\rho),\Phi(\sigma)})|$ and
\begin{align}\label{F-4.7}
\bigl[J_{h_t}(\rho,\sigma)^{-1}-\Phi^*J_{h_t}(\Phi(\rho),\Phi(\sigma))^{-1}\Phi\bigr]K=0,
\qquad t\in T.
\end{align}
One can define a linear contraction $V:\cBH\to\cBK$ by
$V(X\sigma^{-1/2}):=\Phi(X)\Phi(\sigma)^{-1/2}$, $X\in\cBH$. Then \eqref{F-4.7} is equivalently
written as
\[
(\Delta_{\rho,\sigma}+tI)^{-1}(K\sigma^{-1/2})
=V^*(\Delta_{\Phi(\rho),\Phi(\sigma)}+tI)^{-1}(\Phi(K)\Phi(\sigma)^{-1/2}),\qquad t\in T.
\]
For an arbitrary function $f$ on $(0,\infty)$, similarly to \cite[around (5.5)--(5.7)]{HMPB} one can
obtain
\[
V^*f(\Delta_{\Phi(\rho),\Phi(\sigma)})(\Phi(K)\Phi(\sigma)^{-1/2})
=f(\Delta_{\rho,\sigma})(K\sigma^{-1/2}),
\]
which shows (i) by letting $f=1/h$.
\end{proof}

\begin{remark}\label{R-4.2}\rm
The point of Theorem \ref{T-4.1} is that the equality case in the monotonicity in \eqref{F-4.2} of
monotone metrics for a special function $h(x)=x^\alpha$ (i.e., condition (vi) of the theorem)
implies that for all operator monotone functions $h$ (i.e., condition (i)). A similar feature takes place
in Theorems \ref{T-3.2} and \ref{T-3.3} as well.
\end{remark}

\begin{remark}\label{R-4.3}\rm
In condition (ii) of Theorem \ref{T-4.1}, it might be more natural to take $\mu_h$ in place of
$\lambda_{1/h}$ in the support condition. This would be possible thanks to \eqref{F-4.5} if we have
$|\supp\mu_h|=|\supp\mu_{h^*}|$. However, an explicit relation between $\supp\mu_h$ and
$\supp\mu_{h^*}$ does not seem known, although the representing measures $\mu_h$ and $\mu_{h^*}$
are determined in principle by the Stieltjes inversion formula (see, e.g., \cite[Theorem V.4.12]{Bh}). For
example, if $h(x)=2x/(x+1)$ and so $h^*(x)=(x+1)/2$, then $\mu_h=\delta_1$ (point mass at $1$) and
$\mu_{h^*}=0$, so that $|\supp\mu_h|\ne|\supp\mu_{h^*}|$.
\end{remark}

\begin{remark}\label{R-4.4}\rm
Assume that $\rho=\sigma$ in Theorem \ref{T-4.1}. Then the conditions in the theorem is also
equivalent to
\begin{itemize}
\item[(viii)] $\Phi_\sigma^*(\Phi(K))=K$, where $\Phi_\sigma^*$ is the Petz recovery map given in
\eqref{F-3.36}.
\end{itemize}
Indeed, since this is the $z=1/2$ case of the equality in (vii), it is clear that
(vii)$\implies$(viii)$\implies$(vi). Furthermore, there is an alternative proof of
(viii)$\implies$(ii) even without the assumption that $\Phi_\sigma$ is a Schwarz map. Set
\[
\cD_{\sigma,K}:=\{\sigma+s\Re K+t\Im K:s,t\in\bR\}\cap\cBH^{++},
\]
where $\Re K$ and $\Im K$ are the real and imaginary parts of $K$. Since $\Phi_\sigma^*$ is a
trace-preserving positive map with $\Phi_\sigma^*(\Phi(\sigma))=\sigma$ (see \eqref{F-3.37}), it
follows from the monotonicity property of the relative entropy \cite{MR} that (viii) implies the
equality $D(\Phi(\sigma_1)\|\Phi(\sigma_2))=D(\sigma_1\|\sigma_2)$ for all
$\sigma_1,\sigma_2\in\cD_{\sigma,K}$ (see Remark \ref{R-3.15} for a similar argument).
Recall \cite{LeRu} that the so-called \emph{Bogoliubov} (or \emph{Kubo--Mori}) \emph{metric} is
the monotone metric corresponding to $h_0(x):={x-1\over\log x}\in\OM_+[0,\infty)$ and
\[
\gamma_\sigma^{h_0}(X,Y)
=-{\partial^2\over\partial s\partial t}D(\sigma+sX\|\sigma+tY)\Big|_{s=t=0},
\qquad X,Y\in\cBH^\sa.
\]
Therefore, we see that (viii) yields $\gamma_{\Phi(\sigma)}^{h_0}(\Phi(X),\Phi(Y))=
\gamma_\sigma^{h_0}(X,Y)$ for all $X,Y\in\{\Re K,\Im K\}$, which implies (ii).
\end{remark}

Rewriting conditions of Theorem \ref{T-4.1} in the situation where $\Phi$ is given in \eqref{F-3.44}
as in Proposition \ref{P-3.16}, we have equivalent conditions for equality cases in joint convexity
(equivalent to joint subadditivity as mentioned above) of monotone metrics and in Lieb's convexity
(\cite[Corollary 2.1]{Li}). For $\rho_i,\sigma_i\in\cBH^{++}$ and $K_i\in\cBH$, $1\le i\le n$, the
following are among such equivalent conditions, corresponding to (i), (iii), (vi) and (vii):
\begin{itemize}
\item[(I)] $\gamma_{\sum_{i=1}^n\rho_i,\sum_{i=1}^n\sigma_i}^h\bigl(\sum_{i=1}^nK_i,
\sum_{i=1}^nK_i\bigr)=\sum_{i=1}^n\gamma_{\rho_i,\sigma_i}^h(K_i,K_i)$ for
all $h\in\OM_+[0,\infty)$.
\item[(III)] $S_h^{\sum_{i=1}^nK_i}\bigl((\sum_{i=1}^n\rho_i)^{-1}\|(\sum_{i=1}^n\sigma_i)^{-1}\bigr)=
\sum_{i=1}^nS_h^{K_i}(\rho_i^{-1}\|\sigma_i^{-1})$ for all $h\in\OM_+[0,\infty)$.
\item[(VI)] $\Tr(\sum_{i=1}^nK_i)^*(\sum_{i=1}^n\rho_i)^{-\alpha}(\sum_{i=1}^nK_i)
(\sum_{i=1}^n\sigma_i)^{\alpha-1}=\sum_{i=1}^n\Tr K_i^*\rho_i^{-\alpha}K_i\sigma_i^{\alpha-1}$ for
some $\alpha\in(0,1)$.
\item[(VII)] $(\sum_{i=1}^n\rho_i)^{-z}(\sum_{i=1}^nK_i)(\sum_{i=1}^n\sigma_i)^{z-1}=
\rho_i^{-z}K_i\sigma_i^{z-1}$ for all $z\in\bC$.
\end{itemize}

For each $k\in\OMD_+(0,\infty)$ with $k(1)=1$, Temme et al.~\cite{TKRWV} introduced the
\emph{quantum $\chi^2$-divergence} of $\rho,\sigma\in\cSH$ with $\rho^0\le\sigma^0$ relative to
$k$ by
\[
\chi_k^2(\rho,\sigma):=\<\rho-\sigma,\Omega_\sigma^k(\rho-\sigma)\>_\HS
=\<\rho,\Omega_\sigma^k\rho\>_\HS-1,
\]
where $\Omega_\sigma^k:=k(\Delta_\sigma)R_\sigma^{-1}$. In fact, the symmetry
condition $k(x^{-1})=xk(x)$ was imposed in \cite{TKRWV}, which is though not essential since
$\chi_k^2(\rho,\sigma)=\chi_{k_\sym}^2(\rho,\sigma)$ for the symmetrized
$k_\sym(x):=(k(x)+x^{-1}k(x^{-1}))/2\in\OMD_+(0,\infty)$. Since $\Omega_\sigma^k\sigma=I_\cH$
and $\Omega_\sigma^k=J_h^{1/k}(\sigma,\sigma)^{-1}$ in our notation, we have
\begin{align}\label{F-4.8}
\chi_k^2(\rho,\sigma)=\gamma_\sigma^{1/k}(\rho,\rho)-1,
\end{align}
where $\gamma_\sigma^{1/k}(\rho,\rho)$ is defined as a monotone metric on
$\cB(\sigma^0\cH)^{++}$. By \eqref{F-4.2} we have the monotonicity
$\chi_k^2(\Phi(\rho),\Phi(\sigma))\le\chi_k^2(\rho,\sigma)$ under trace-preserving maps $\Phi$
such that $\Phi^*$ is a Schwarz map. Joint convexity of $\chi_k^2(\rho,\sigma)$ in $\rho,\sigma$
was shown in \cite{Han2}. The important case discussed in \cite{TKRWV} is
\[
k_\alpha(x):={x^{-\alpha}+x^{\alpha-1}\over2},\qquad
\chi_\alpha^2(\rho,\sigma):=\chi_{k_\alpha}^2(\rho,\sigma)
=\Tr\rho\sigma^{-\alpha}\rho\sigma^{\alpha-1}-1.
\]

The following are equality conditions for the monotonicity of $\chi^2$-divergences, which were
shown in \cite{Je} with a slightly stronger assumption on $\Phi$. They are, in view of \eqref{F-4.8},
rewritings of some conditions of Theorem \ref{T-4.1}, while those corresponding to (iii)--(v) are
omitted here. 

\begin{cor}\label{C-4.5}
For any trace-preserving map $\Phi:\cBH\to\cBK$ such that $\Phi^*$ is a Schwarz map, and for
every $\rho,\sigma\in\cSH$ with $\rho^0\le\sigma^0$, the following conditions are equivalent:
\begin{itemize}
\item[(i$'$)] $\chi^2_k(\Phi(\rho),\Phi(\sigma))=\chi_k^2(\rho,\sigma)$ for all $k\in\OMD_+(0,\infty)$
with $k(1)=1$.
\item[(ii$'$)] $\chi^2_k(\Phi(\rho),\Phi(\sigma))=\chi_k^2(\rho,\sigma)$ for some
$k\in\OMD_+(0,\infty)$ with $k(1)=1$ and
\begin{align}\label{F-4.9}
|\supp\lambda_k|\ge|\Sp(\Delta_\sigma)\cup\Sp(\Delta_{\Phi(\sigma)})|.
\end{align}
\item[(vi$'$)] $\Tr\Phi(\rho)\Phi(\sigma)^{-\alpha}\Phi(\rho)\Phi(\sigma)^{\alpha-1}=
\Tr\rho\sigma^{-\alpha}\rho\sigma^{\alpha-1}$, i.e.,
$\chi_\alpha^2(\Phi(\rho),\Phi(\sigma))=\chi_\alpha^2(\rho,\sigma)$ for some $\alpha\in(0,1)$.
\item[(vii$'$)] $\Phi^*(\Phi(\sigma)^{-z}\Phi(\rho)\Phi(\sigma)^{z-1})=\sigma^{-z}\rho\sigma^{z-1}$
for all $z\in\bC$.
\end{itemize}
\end{cor}

\section{Maps preserving quasi-entropies and monotone metrics}\label{Sec-5}

In this section we first determine the form of maps $\Phi:\cBH\to \cBK$ which satisfy equality
conditions given in Sec.~\ref{Sec-3.1} for a sufficient family of $K\in\cBK$ with some (fixed)
pair $\rho,\sigma\in\cBH^{++}$.

\begin{thm}\label{T-5.1}
Let $\Phi:\cBH\to \cBK$ be a trace-preserving map such that $\Phi^*$ is a Schwarz map, and
let $Q:=\Phi(I_\cH)^0$. Then the following conditions are equivalent:
\begin{itemize}
\item[(a)] $\Phi$ satisfies all of (ii)--(ix) given in Sec.~\ref{Sec-3.1} for all
$\rho,\sigma\in\Phi^*(\cBK^+)$ and all $K\in\cBK$.
\item[(b)] $\Phi$ satisfies (i) (equivalently (i$'$)) in Sec.~\ref{Sec-3.1} for all
$\rho,\sigma\in\Phi^*(\cBK^+)$ and all $K\in\cBK$ with $QK=KQ$.
\item[(c)] There exist $\rho,\sigma\in\cBH^{++}$ and a set $\cT\subset \cBK$ such that
$Q\cT Q\cup\{Q\}$ (where $Q\cT Q:=\{QKQ:K\in\cT\}$) generates $\cB(Q\cK)$ as a
$C^*$-subalgebra and $\Phi$ satisfies one of (ii)--(ix) in Sec.~\ref{Sec-3.1} for $\rho,\sigma$ and
all $K\in\cT$.
\item[(d)] The multiplicative domain of $\Phi^*|_{\cB(Q\cK)}$ is $\cB(Q\cK)$.
\item[(e)] $\cM_{\Phi^*}=\cB(Q\cK)\oplus\cB(Q^\perp\cK)$ $(=\{Y\in\cBK:QY=YQ\})$.
\item[(f)] There exist a decomposition $\cH=\cH_1\otimes\cH_2$ and an isometry $V:\cH_1\to\cK$
such that $VV^*=Q$ and
\begin{align}\label{F-5.1}
\Phi(X)=V(\T_{\cH_2}X)V^*,\qquad X\in\cBH.
\end{align}
\end{itemize}

If $\Phi$ is in addition unital, then $\Phi$ satisfies one (equivalently, all) of (a)--(f) if and only if
$\Phi$ is a unitary transformation, i.e., there exists a unitary $U:\cH\to\cK$ such that
$\Phi(X)=UXU^*$ for all $X\in\cBH$.
\end{thm}

\begin{proof}
Let $\hat\Phi:\cBH\to\cB(Q\cK)$ be defined as in \eqref{F-3.21} with $\sigma=I_\cH$ and so
$\tilde\sigma^0=Q$, that is, $\hat\Phi$ in the present case is $\Phi$ viewed as a map into
$\cB(Q\cK)$ instead of $\cBK$. By Lemma \ref{L-3.10} note that
$\hat\Phi^*=\Phi^*|_{\cB(Q\cK)}$ and $\Phi^*(Y)=\hat\Phi^*(QYQ)=\Phi^*(QYQ)$ for any
$Y\in\cBK$. Hence conditions (d) and (f) for $\Phi$ are the same as (d) and (f) for $\hat\Phi$,
respectively. Moreover, by Lemma \ref{L-3.11}, condition (c) for $\Phi$ implies (c) for $\hat\Phi$.
Therefore, we may and do assume $Q=I_\cK$ below to prove that (c)$\implies$(d)$\implies$(f).

(c)$\implies$(d).\enspace
Assume that (c) holds with $Q=I_\cK$. By Theorem \ref{T-3.2}\,(2) the equality in (vii) in
Sec.~\ref{Sec-3.1} holds for all $\alpha\in(0,1)$ and any $K\in\cT$. Letting $\alpha\searrow0$ and
$\alpha\nearrow1$ in the equality gives
\[
\Tr\sigma\Phi(K^*K)=\Tr\sigma\Phi(K)^*\Phi(K),\qquad
\Tr\rho\Phi(KK^*)=\Tr\rho\Phi(K)\Phi(K)^*,
\]
which imply that $\cT\subset\cM_{\Phi^*}$. Since $\cM_{\Phi^*}$ is a $C^*$-subalgebra of
$\cB(\cK)$, $\cM_{\Phi^*}=\cBK$ (i.e., (d) in the case $Q=I_\cK$) follows.

(d)$\implies$(f).\enspace
Since (d) (with $Q=I_\cK$) implies that $\Phi^*$ is a *-isomorphism from $\cBK$ into $\cBH$
(due to $\cBK$ being a factor), it follows (see, e.g., \cite[Lemma 2.2.2]{GDJ}) that
$\Phi^*(\cBK)'\cap \cBH$ is a factor and
$\cBH=\Phi^*(\cBK)\otimes(\Phi^*(\cBK)'\cap \cBH)$, where
$\Phi^*(\cBK)'\cap \cBH$ is the commutant of $\Phi^*(\cBK)$ in
$\cBH$. Hence one can write $\cH$ as a tensor product space
$\cH_1\otimes\cH_2$ so that $\Phi^*(\cBK)=\cB(\cH_1)\otimes I_2$ where $I_2:=I_{\cH_2}$, and
there exists a unitary $V:\cH_1\to\cK$ (this becomes an isometry when $Q\ne I_\cK$) such that
$\Phi^*(Y)=V^*YV\otimes I_2$ for any $Y\in\cBK$. For any $X\in\cBH$ one has
\begin{align}\label{F-5.2}
\<\Phi^*(Y),X\>_\HS=\Tr (V^*Y^*V\otimes I_2)X=\Tr Y^*V(\T_{\cH_2}X)V^*,
\qquad Y\in\cBK.
\end{align}
Therefore, \eqref{F-5.1} follows.

Let us now proceed to prove without the $Q=I_\cK$ assumption.

(e)$\implies$(d) is clear since $\Phi^*(Y)=\Phi^*(QYQ)$ for all $Y\in\cBK$.

(f)$\implies$(e).\enspace
Assume (f). Thanks to \eqref{F-5.1} and \eqref{F-5.2} (viewed in the converse direction)
it follows that
\begin{align}\label{F-5.3}
\Phi^*(Y)=V^*YV\otimes I_2,\qquad Y\in\cBK.
\end{align}
For every $Y\in\cBK$ one has
\begin{align*}
\Phi^*(Y^*Y)=\Phi^*(Y)^*\Phi(Y)
&\iff V^*Y^*YV=V^*Y^*VV^*YV \\
&\iff QY^*YQ=QY^*QYQ \\
&\iff(I_\cK-Q)YQ=0,
\end{align*}
and also $\Phi^*(YY^*)=\Phi^*(Y)\Phi^*(Y)^*\iff(I_\cK-Q)Y^*Q=0\iff QY(I_\cK-Q)=0$. Hence (e)
follows.

(a)$\implies$(c) and (b)$\implies$(c) are immediate by taking $\rho=\sigma=I_\cH$ and
$\cT=\cB(Q\cK)$ in (a) and (b).

(f)$\implies$(a).\enspace
Assume (f) and let $\rho,\sigma\in\Phi^*(\cBK^+)$; then in view of \eqref{F-5.3},
$\rho=\rho_1\otimes I_2$ and $\sigma=\sigma_1\otimes I_2$ for some
$\rho_1,\sigma_1\in\cB(\cH_1)^+$. To prove that $\Phi$ satisfies (ii)--(ix) in Sec.~\ref{Sec-3.1} for
all $\rho,\sigma\in\Phi^*(\cBK^+)$ and all $K\in\cBK$, it suffices by Lemma \ref{L-3.1} that $\Phi$
satisfies (ii) and (viii) for those $\rho,\sigma$ and $K$. Let $h$ be as stated in (ii). With the
spectral decompositions $\rho_1=\sum_aaP_a$ and $\sigma_1=\sum_bbQ_b$ we find by
\eqref{F-5.1} and \eqref{F-5.3} that
\begin{align*}
S_h^K(\Phi(\rho)\|\Phi(\sigma))
&=(\Tr I_2)S_h^K(V\rho_1V^*\|V\sigma_1V^*) \\
&=(\Tr I_2)\sum_{a,b>0}bh\Bigl({a\over b}\Bigr)\Tr K^*(VP_aV^*)K(VQ_bV^*) \\
&=(\Tr I_2)\sum_{a,b>0}bh\Bigl({a\over b}\Bigr)\Tr(V^*K^*V)P_a(V^*KV)Q_b \\
&=\sum_{a,b>0}bh\Bigl({a\over b}\Bigr)\Tr\Phi^*(K)^*(P_a\otimes I_2)\Phi^*(K)(Q_b\otimes I_2)
=S_h^{\Phi^*(K)}(\rho\|\sigma),
\end{align*}
where we have repeatedly used expression \eqref{F-2.4} with $h(0)=h'(\infty)=0$. Hence (ii) holds.
For any $K,Y\in\cBK$ and $z\in\bC$ we have
\begin{equation}\label{F-5.4}
\begin{aligned}
\sigma^0\Phi^*(Y\Phi(\rho)^zK\Phi(\sigma)^{-z})\sigma^0
&=\sigma^0\bigl[\bigl(V^*Y(V\rho_1^zV^*)K(V\sigma_1^{-z}V^*)V\bigr)\otimes I_2\bigr]\sigma^0 \\
&=\sigma^0\bigl[\bigl((V^*YV)\rho_1^z(V^*KV)\sigma_1^{-z}\bigr)\otimes I_2\bigr]\sigma^0 \\
&=\sigma^0\Phi^*(Y)\rho^z\Phi^*(K)\sigma^{-z},
\end{aligned}
\end{equation}
implying (viii). Hence (a) follows.

(f)$\implies$(b).\enspace
Assume (f) so that (a) and (e) hold as proved above. Let $\rho,\sigma\in\Phi^*(\cBK^+)$ and
$K\in\cBK$ with $QK=KQ$. Since $K\in\cM_{\Phi^*}$ by (e), condition (b) holds by (a) and
Theorems \ref{T-3.2}\,(3) and \ref{T-3.3}\,(1).

Therefore, it has been shown that (a)--(f) are all equivalent. Finally, if $\Phi$ is unital, then letting
$X=I_\cH$ in \eqref{F-5.1} gives $(\Tr I_2)Q=I_\cK$, and this is the case when $\dim\cH_2=1$ and
$V$ maps onto $\cK$. So the last assertion holds.
\end{proof}

\begin{remark}\label{R-5.2}\rm
In the situation where one (hence all) of conditions (a)--(f) in Theorem \ref{T-5.1} holds with
$Q\ne I_\cK$, one can never include condition (i) or (i$'$) (in Sec.~\ref{Sec-3.1}) in (a), as
immediately seen from (e) and Proposition \ref{P-3.13}.
\end{remark}

Finally, we are concerned with the special form of $2$-positive $\Phi:\cBH\to \cBK$ preserving a
standard $f$-divergence $S_f(\rho\|\sigma)=S_f^I(\rho\|\sigma)$, as well as a monotone
metric $\gamma_\sigma^h(K,K)$ and a $\chi^2$-divergence $\chi_k^2(\rho,\sigma)$ discussed in
Sec.~\ref{Sec-4}, for sufficiently many $\rho\in\cBH^+$ or $K\in\cBH$ with a fixed
$\sigma\in\cBH^{++}$.

\begin{thm}\label{T-5.3}
Let $\Phi:\cBH\to \cBK$ be a trace-preserving $2$-positive map, and let $Q:=\Phi(I_\cH)^0$.
\begin{itemize}
\item[(1)] The following conditions are equivalent:
\begin{itemize}
\item[(a)] $S_f(\Phi(\rho)\|\Phi(\sigma))=S_f(\rho\|\sigma)$ for all $\rho,\sigma\in\cBH^+$
and for all $f\in\OC[0,\infty)$.
\item[(b)] There exist a $\sigma\in\cBH^{++}$, a set $\cD\subset \cBH$ and an
$f\in\OC[0,\infty)$ with \eqref{F-3.5} for $\nu_f$ such that $\cD\cup\{\sigma\}$
generates $\cBH$ as a $C^*$-subalgebra and
\[
S_f(\Phi(\rho)\|\Phi(\sigma))=S_f(\rho\|\sigma),\qquad\rho\in\cD.
\]
\item[(c)] There exist a $\sigma\in\cBH^{++}$ and an $h\in\OM_+[0,\infty)$ such that
\begin{align}
&|\supp\lambda_{1/h}|\ge|\Sp(\Delta_\sigma)\cup\Sp(\Delta_{\Phi(\sigma)})|,
\label{F-5.5}\\
&\gamma_{\Phi(\sigma)}^h(\Phi(X),\Phi(X))=\gamma_\sigma^h(X,X),\qquad X\in\cBH.
\label{F-5.6}
\end{align}
\item[(d)] There exist a $\sigma\in\cBH^{++}$, a set $\cT\subset\cBH$ and an
$h\in\OM_+[0,\infty)$ with \eqref{F-5.5} such that $\cT\cup\{\sigma\}$ generates $\cBH$ as a
$C^*$-subalgebra and
\[
\gamma_{\Phi(\sigma)}^h(\Phi(K),\Phi(K))=\gamma_\sigma^h(K,K),\qquad K\in\cT.
\]
\item[(e)] There exist a $\sigma\in\cSH^{++}$ and a $k\in\OMD_+(0,\infty)$ with $k(1)=1$ and
\eqref{F-4.9} such that
\[
\chi_k^2(\Phi(\rho),\Phi(\sigma))=\chi_k^2(\rho,\sigma),\qquad\rho\in\cSH.
\]
\item[(f)] There exists a $\sigma\in\cBH^{++}$ such that $\Phi_\sigma^*\circ\Phi=\id_{\cBH}$
where $\Phi_\sigma^*$ is the Petz recovery map of $\Phi$ relative to $\sigma$.
\item[(g)] There exist a decomposition $Q\cK=\cK_1\otimes\cK_2$, a unitary $U$ from $\cH$
onto $\cK_1$, and an $\eta\in\cB(\cK_2)^{++}$ with $\Tr\eta=1$ such that
\begin{align}\label{F-5.7}
\Phi(X)=(UXU^*\otimes\eta)\oplus0,\qquad X\in\cBH,
\end{align}
where $0$ is the zero operator on $Q^\perp\cK$.
\end{itemize}

\item[(2)] If (g) holds, then
\begin{align}
\Phi^*(Y)&=U^*\bigl[\T_{\cK_2}((I_{\cK_1}\otimes\eta^{1/2})QYQ
(I_{\cK_1}\otimes\eta^{1/2}))\bigr]U,\qquad Y\in\cBK, \label{F-5.8}\\
\cM_{\Phi^*}&=(\cB(\cK_1)\otimes I_{\cK_2})\oplus\cB(Q^\perp\cK), \label{F-5.9}
\end{align}
and $\Phi$ satisfies all of (i)--(x) and (i$'$) in Sec.~\ref{Sec-3.1} for all $\rho,\sigma\in\cBH^+$
and all $K\in\cM_{\Phi^*}$.

\item[(3)] If $\Phi$ is in addition unital, then $\Phi$ satisfies one (equivalently, all) of (a)--(g) if and
only if $\Phi$ is a unitary transformation.
\end{itemize}
\end{thm}

\begin{proof}
(1)\enspace
(a)$\implies$(b) is trivial.

(c)$\iff$(e).\enspace
Assume (e) and let $h:=1/k$. Then by \eqref{F-4.8} one has \eqref{F-5.6} for all $X\in\cSH$, which
extends to all $X\in\cBH$ by polarization formula. The converse is also easy since we may assume
$h(1)=1$ in (c).

(a)$\iff$(f).\enspace
Since $\Phi$ is $2$-positive, this is immediate from \cite[Theorem 5.1]{HMPB}.

(c)$\implies$(f).\enspace
Assume (c). Then by (ii)$\implies$(vii) (for $z=1/2$) of Theorem \ref{T-4.1} one has
\[
\Phi^*(\Phi(\sigma)^{-1/2}\Phi(X)\Phi(\sigma)^{-1/2})=\sigma^{-1/2}X\sigma^{-1/2},
\qquad X\in\cBH.
\]
Hence $\Phi_\sigma^*(\Phi(X))=X$ for all $X\in\cBH$.

(g)$\implies$(c),\,(d).\enspace
Assume (g), and let us show more strongly that \eqref{F-5.6} holds for any $\sigma\in\cBH^{++}$
and any $h\in\OM_+[0,\infty)$. Let $\sigma\in\cBH^{++}$ and $X\in\cBH$ be arbitrary. Since
$\gamma_{\Phi(\sigma)}^h(\Phi(X),\Phi(X))$ is defined on $\cB(Q\cK)$, we may and do
assume that $Q=I_\cK$ so that $\Phi(\sigma)\in\cBK^{++}$. For each $n\in\bN\cup\{0\}$ note that
\begin{align*}
\Delta_{\Phi(\sigma)}^nR_{\Phi(\sigma)}^{-1}\Phi(X)
&=L_{\Phi(\sigma)}^nR_{\Phi(\sigma)}^{-n-1}\Phi(X) \\
&=(U\sigma U^*\otimes\eta)^n(UXU^*\otimes\eta)(U\sigma U^*\otimes\eta)^{-n-1} \\
&=U\sigma^nX\sigma^{-n-1}U^*\otimes I_\cK
=U(\Delta_\sigma^nR_\sigma^{-1}X)U^*\otimes I_\cK.
\end{align*}
Approximating $1/h$ uniformly on an interval $[a,b]$ ($0<a<b$) including
$\Sp(\Delta_\sigma)\cup\Sp(\Delta_{\Phi(\sigma)})$ by polynomials, one has
\begin{align*}
J_h(\Phi(\sigma),\Phi(\sigma))^{-1}\Phi(X)
&=(1/h)(\Delta_{\Phi(\sigma)})R_{\Phi(\sigma)}^{-1}\Phi(X)
\quad(\mbox{see \eqref{F-2.3}}) \\
&=U((1/h)(\Delta_\sigma)R_\sigma^{-1}X)U^*\otimes I_\cK \\
&=U(J_h(\sigma,\sigma)^{-1}X)U^*\otimes I_\cK
\end{align*}
so that
\begin{align*}
\gamma_{\Phi(\sigma)}^h(\Phi(X),\Phi(X))
&=\Tr\Phi(X)^*J_h(\Phi(\sigma),\Phi(\sigma))^{-1}\Phi(X)
\quad(\mbox{see \eqref{F-4.1}}) \\
&=\Tr(UX^*U^*\otimes\eta)(U(J_h(\sigma,\sigma)^{-1}X)U^*\otimes I_\cK) \\
&=\Tr X^*J_h(\sigma,\sigma)^{-1}X\cdot\Tr\eta=\gamma_\sigma^h(X,X).
\end{align*}

(b)$\implies$(g).\enspace
Since $\Phi$ is $2$-positive, by \cite[Theorem 3.19]{HM} (see, e.g., \cite{Mos} for more details)
we have decompositions
\begin{align}\label{F-5.10}
\cH=\bigoplus_{j=1}^r\cH_{j,L}\otimes\cH_{j,R},\qquad
Q\cK=\bigoplus_{j=1}^r\cK_{j,L}\otimes\cK_{j,R},
\end{align}
operators $\omega_j\in\cB(\cH_{j,R})^{++}$, unitaries $U_j:\cH_{j,L}\to\cK_{j,L}$, and
trace-preserving $2$-positive maps $\Phi_j:\cB(\cH_{j,R})\to\cB(\cK_{j,R})$ such that
\begin{align}
\bigl(\cF_{\Phi_\sigma^*\circ\Phi}\bigr)^+
&=\bigoplus_{j=1}^r\cB(\cH_{j,L})^+\otimes\omega_j, \label{F-5.11}\\
\Phi(X_{j,L}\otimes X_{j,R})&=U_jX_{j,L}U_j^*\otimes\Phi_j(X_{j,R}) \label{F-5.12}
\end{align}
for all $X_{j,L}\in\cB(\cH_{j,L})$, $X_{j,R}\in\cB(\cH_{j,R})$. (To be more more precise,
the equalities in \eqref{F-5.12} hold up to unitary conjugation.)
Here, $\cF_{\Phi_\sigma^*\circ\Phi}$ denotes the set of $X\in\cBH^+$ such that
$\Phi_\sigma^*\circ\Phi(X)=X$. By \cite[Theorem 5.1]{HMPB} condition (b) implies that
$\rho,\sigma\in\bigl(\cF_{\Phi_\sigma^*\circ\Phi}\bigr)^+$ (for all $\rho\in\cD$) so that we have
by \eqref{F-5.11}
\[
\rho=\bigoplus_{j=1}^r\rho_j\otimes\omega_j,\qquad
\sigma=\bigoplus_{j=1}^r\sigma_j\otimes\omega_j
\qquad(\rho_j,\sigma_j\in\cB(\cH_{j,L})^+).
\]
Let $P_1$ be any eigen-projection of $\omega_1$. Then the above expressions imply
in particular that all $\rho\in\cD$ and $\sigma$ commute with $I_{\cH_{1,L}}\otimes P_1$. Since
$\cD\cup\{\sigma\}$ generates $\cBH$ as a $C^*$-subalgebra, it follows that
$I_{\cH_{1,L}}\otimes P_1=I_\cH$. This means that $r=1$ and $\dim\cH_{1,R}=1$.
Condition (g) then follows immediately from \eqref{F-5.12}, by letting $\eta:=\Phi_1(1)$ where $1$
is the identity of $\cB(\cH_{1,R})\cong\bC$.

(d)$\implies$(g).\enspace
Let us make use of the same decompositions as in \eqref{F-5.10}--\eqref{F-5.12}.
Assume (d). By Remark \ref{R-4.4} we find that
\[
\cD_{\sigma,\cT}:=\{\sigma+s\Re K+t\Im K:K\in\cT,\,s,t\in\bR\}\cap\cBH^{++}
\]
is included in $(\cF_{\Phi_\sigma^*\circ\Phi})^+$. Since the assumption means that $\cD_{\sigma,\cT}$
generates $\cBH$ as a $C^*$-subalgebra, (g) can be proved as in the above proof of (b)$\implies$(g).

(2)\enspace
Assume that (g) holds. For any $X\in\cBH$ and $Y\in\cBK$, by \eqref{F-5.7} one has
\begin{align*}
\<\Phi(X),Y\>_\HS&=\Tr(UXU^*\otimes\eta)^*QYQ \\
&=\Tr(UX^*U^*\otimes I_2)(I_1\otimes\eta^{1/2})QYQ(I_1\otimes\eta^{1/2}) \\
&=\Tr(UX^*U^*)\T_{\cK_2}((I_1\otimes\eta^{1/2})QYQ(I_1\otimes\eta^{1/2})) \\
&=\Tr X^*U^*\bigl[\T_{\cK_2}((I_1\otimes\eta^{1/2})QYQ(I_1\otimes\eta^{1/2}))\bigr]U,
\end{align*}
implying \eqref{F-5.8}. Next, if $Y=(Y_1\otimes I_2)\oplus Y_0$ where $Y_1\in\cB(\cK_1)$
and $Y_0\in\cB(Q^\perp\cK)$, then
\begin{align*}
\Phi^*(Y^*Y)
&=U^*\bigl[\T_{\cK_2}((I_1\otimes\eta^{1/2})(Y_1^*Y_1\otimes I_2)
(I_1\otimes\eta^{1/2}))\bigr]U \\
&=U^*\Tr_{\cK_2}(Y_1^*Y_1\otimes\eta)U=U^*Y_1^*Y_1U \\
&=U^*Y_1^*UU^*Y_1U=\Phi^*(Y)^*\Phi^*(Y)
\end{align*}
and similarly $\Phi^*(YY^*)=\Phi^*(Y)\Phi^*(Y)^*$. Hence
\begin{align}\label{F-5.13}
\cM_{\Phi^*}\supset(\cB(\cK_1)\otimes I_2)\oplus\cB(Q^\perp\cK).
\end{align}
On the other hand, define a unital CP map $\Psi:\cB(Q\cK)\to\cB(\cK_1)$ by
\[
\Psi(Y):=\T_{\cK_2}((I_1\otimes\eta^{1/2})Y(I_1\otimes\eta^{1/2})),
\qquad Y\in\cB(Q\cK).
\]
If $Y\in\cBK$ belongs to $\cM_{\Phi^*}$, then it follows from \eqref{F-5.8} that
\[
\Psi(QY^*YQ)=\Psi(QY^*Q)\Psi(QYQ),\qquad
\Psi(QYY^*Q)=\Psi(QYQ)\Psi(QY^*Q).
\]
Since $\Psi(QY^*Q)\Psi(QYQ)\le\Psi(QY^*QYQ)$, one has
$\Psi(QY^*(I_\cK-Q)YQ)=0$ so that $(I_\cK-Q)YQ=0$ since $\Psi$ is faithful on $\cB(Q\cK)^+$.
Similarly $(I_\cK-Q)Y^*Q=0$, and hence $QY=YQ$, which implies that
\begin{align}\label{F-5.14}
\cM_{\Phi^*}\subset\cB(Q\cK)\oplus\cB(Q^\perp\cK).
\end{align}
By \eqref{F-5.13} and \eqref{F-5.14}, to show \eqref{F-5.9} we need to prove that
$\cA:=\cM_{\Phi^*}\cap\cB(Q\cK)=\cB(\cK_1)\otimes I_2$. Note that $\cA$ is a von Neumann
subalgebra of $\cB(\cK_1)\otimes\cB(\cK_2)$ and $\cA'\subset I_1\otimes\cB(\cK_2)$ by
\eqref{F-5.13}, where $\cA'$ is the commutant of $\cA$. Hence $\cA'=I_1\otimes\cB$ for some
von Neumann subalgebra $\cB$ of $\cB(\cK_2)$. By von Neumann's double commutant theorem
we have $\cA=\cB(\cK_1)\otimes\cB'$. Now let $Y_2\in\cB'$; then $I_1\otimes Y_2\in\cM_{\Phi^*}$.
Since
\[
\Phi^*(I_1\otimes Y_2)=U^*\T_{\cK_2}(I_1\otimes\eta^{1/2}Y_2\eta^{1/2})U
=(\Tr\eta Y_2)I_\cH,
\]
it follows that $\Tr\eta Y_2^*Y_2=|\Tr\eta Y_2|^2$, from which it must hold that
$Y_2=(\Tr\eta Y_2)I_2$. Therefore, $\cB'=\bC I_2$ and \eqref{F-5.9} has been shown.
Finally, for any $\rho,\sigma\in\cBH^+$ and $K\in\cM_{\Phi^*}$ so that
$K=(K_1\otimes I_2)\otimes K_0$ where $K_1\in\cB(\cK_1)$ and $K_0\in\cB(Q^\perp\cK)$,
one has $\Phi^*(K)=U^*[\T_{\cK_2}(K_1\otimes\eta)]U=U^*K_1U$. Hence for any $\alpha\in(0,1)$,
\begin{align*}
\Tr\Phi^*(K)^*\rho^\alpha\Phi^*(K)\sigma^{1-\alpha}
&=\Tr U^*K_1^*U\rho^\alpha U^*K_1U\sigma^{1-\alpha} \\
&=\Tr(K_1^*\otimes I_2)(U\rho^\alpha U^*\otimes\eta^\alpha)
(K_1\otimes I_2)(U\sigma^{1-\alpha}U^*\otimes\eta^{1-\alpha}) \\
&=\Tr K^*\Phi(\rho)^\alpha K\Phi(\sigma)^{1-\alpha},
\end{align*}
so that condition (vii) in Sec.~\ref{Sec-3.1} holds for these $\rho,\sigma$ and $K$. Therefore, the
last assertion of (2) holds due to Theorems \ref{T-3.2}\,(3) and \ref{T-3.3}\,(1).

(3)\enspace
If $\Phi$ is unital, then \eqref{F-5.7} gives $(I_{\cK_1}\otimes\eta)\oplus0=I_\cK$, and this is the
case when $Q=I_\cK$ and $\dim\cK_2=1$. So the assertion holds.
\end{proof}

It is worth noting that the forms of $\Phi$ in \eqref{F-5.1} of Theorem \ref{T-5.1} and in \eqref{F-5.7}
of Theorem \ref{T-5.3} have some relationship like duality. Indeed, $\Phi^*$ in \eqref{F-5.8} is
quite similar to $\Phi$ in \eqref{F-5.1}. This might be natural from a certain duality between
quasi-entropies and monotone metrics found in their definitions, see \eqref{F-2.2} and \eqref{F-4.1}.


\begin{remark}\label{R-5.4}\rm
It is well known that $\cBK$ with $\dim\cK=n$ for any $n\in\bN$ is singly generated. Indeed, if
$K\in\cBK$ has $n$ distinct eigenvalues and no two eigenvectors corresponding to different
eigenvalues are orthogonal (for example, if $A$ has a matrix representation
\[
A=\begin{bmatrix}\lambda_1&1&&&\\&\lambda_2&1&&\mbox{\LARGE0}&\\
&&\ddots&\ddots&\\&\mbox{\LARGE0}&&\lambda_{n-1}&1\\&&&&\lambda_n\end{bmatrix}
\]
with $\lambda_1<\lambda_2<\dots<\lambda_n$), then $\cBK$ is generated by a single $K$ (see
\cite{Pea}). Take such a generator $K$ of $\cB(Q\cK)$. Let $\rho=\sigma=I_\cH$ and
$\cT=\{K\}$ in (c) of Theorem \ref{T-5.1}; then conditions (a)--(f) of Theorem \ref{T-5.1} are
equivalent to a single equality (the $\alpha=1/2$ case of condition (vii) in Sec.~\ref{Sec-3.1})
\[
\Tr K^*\Phi(I)^{1/2}K\Phi(I)^{1/2}=\Tr\Phi^*(K)^*\Phi^*(K).
\]
Also, when $\Phi$ is $2$-positive, let $\sigma=I_\cH$, $\cT=\{K\}$ and $h(x)=x^{1/2}$ in (d) of
Theorem \ref{T-5.3}; then (a)--(g) of Theorem \ref{T-5.3} are equivalent to a single equality
\[
\Tr\Phi(K)^*\Phi(I)^{-1/2}\Phi(K)\Phi(I)^{-1/2}=\Tr K^*K.
\]
\end{remark}

\begin{remark}\label{R-5.5}\rm
A similar result to Theorem \ref{T-5.3}\,(1)  was given in \cite{MNS} from a different point of view,
where the authors proved that if $f$ is a strictly convex function on $[0,\infty)$ and a
transformation (i.e., a bijection) $\phi$ on the set $\cSH$ satisfies
$S_f(\phi(\rho)\|\phi(\sigma))=S_f(\rho\|\sigma)$ for all $\rho,\sigma\in\cSH$, then
$\phi(\rho)=U\rho U^*$ ($\rho\in\cSH$) with either a unitary or an antiunitary $U$ on $\cH$. Since
the bijectivity of $\phi$ is essential in their proof based on Wigner's theorem, it does not seem
possible to apply the method in \cite{MNS} to our case.
\end{remark}

\section*{Acknowledgments}

The author thanks Anna Jen\v cov\'a, Mil\'an Mosonyi, and Yoshimichi Ueda for discussions which
helped to improve the paper. In particular, he is indebted to Jen\v cov\'a (\cite{Je6}) for
the idea of proving (b)$\implies$(g) and (d)$\implies$(g) of Theorem \ref{T-5.3}.

\end{document}